\documentclass[onefignum,onetabnum]{siamonline190516}
\usepackage{amsmath, amssymb, mathrsfs}
\usepackage{physics}
\usepackage{lipsum}
\usepackage{amsfonts}
\usepackage{graphicx,mathtools}
\usepackage{algorithmic}
\usepackage{booktabs}
\usepackage[titletoc]{appendix}
\usepackage{bm}

\usepackage{threeparttable}
\usepackage{subcaption} 
\usepackage{siunitx}
\usepackage{multirow}

\usepackage{enumitem}
\setlist[enumerate]{leftmargin=.5in}
\setlist[itemize]{leftmargin=.5in}

\newtheorem{Th}{Theorem}[section]
\newtheorem{Lemma}[Th]{Lemma}

\newtheorem{Rem}[Th]{Remark}
\newtheorem{?}[Th]{Problem}

\newtheorem{Ass}[Th]{Assumption}

\def\R{{\mathbb R}}

\def\N{{\mathbb N}}
\def\E{{\mathbb E}}
\def\R{{\mathbb R}}

\def\S{{\mathcal S}}
\def\W{{\bm W}}
\def\H{{\bm H}}
\def\P{{\mathbb P}}
\def\A{{\bm A}}
\def\ZZ{{\bm Z}}

\def\e{{\varepsilon}}

\def\ex{{\text{epr}}}
\def\ext{{\text{ept}}}

\DeclareMathOperator*{\argmin}{arg\,min}
\newcommand{\indep}{\perp \!\!\! \perp}

\def\x{{\bm x}}
\def\y{{\bm y}}
\def\z{{\bm z}}

\def\ex{{\text{epr}}}
\def\ext{{\text{ept}}}

\usepackage{xspace}
\usepackage{bold-extra}
\usepackage[most]{tcolorbox}

\colorlet{texcscolor}{blue!50!black}
\colorlet{texemcolor}{red!70!black}
\colorlet{texpreamble}{red!70!black}
\colorlet{codebackground}{black!25!white!25}

\lstdefinestyle{siamlatex}{%
  style=tcblatex,
  texcsstyle=*\color{texcscolor},
  texcsstyle=[2]\color{texemcolor},
  keywordstyle=[2]\color{texemcolor},
  moretexcs={cref,Cref,maketitle,mathcal,text,headers,email,url},
}

\tcbset{%
  colframe=black!75!white!75,
  coltitle=white,
  colback=codebackground, 
  colbacklower=white, 
  fonttitle=\bfseries,
  arc=0pt,outer arc=0pt,
  top=1pt,bottom=1pt,left=1mm,right=1mm,middle=1mm,boxsep=1mm,
  leftrule=0.3mm,rightrule=0.3mm,toprule=0.3mm,bottomrule=0.3mm,
  listing options={style=siamlatex}
}

\newtcblisting[use counter=example]{example}[2][]{%
  title={Example~\thetcbcounter: #2},#1}

\newtcbinputlisting[use counter=example]{\examplefile}[3][]{%
  title={Example~\thetcbcounter: #2},listing file={#3},#1}

\DeclareTotalTCBox{\code}{ v O{} }
{ 
  fontupper=\ttfamily\color{black},
  nobeforeafter,
  tcbox raise base,
  colback=codebackground,colframe=white,
  top=0pt,bottom=0pt,left=0mm,right=0mm,
  leftrule=0pt,rightrule=0pt,toprule=0mm,bottomrule=0mm,
  boxsep=0.5mm,
  #2}{#1}

\patchcmd\newpage{\vfil}{}{}{}
\begin{tcbverbatimwrite}{tmp_\jobname_header.tex}
\title{An approximate control variates approach to multifidelity distribution estimation\thanks{Submitted to the editors.
\funding{RH is partially supported by the Hong Kong Research Grants Council (grant no. 14301821) and Start-up Fund for New Recruits, The Hong Kong Polytechnic University.
AN is partially supported by NSF DMS-1848508 and AFOSR FA9550-20-1-0338.
BK and DL are supported by the Ministry of Trade, Industry and Energy (MOTIE) and the Korea Institute for Advancement
of Technology (KIAT) through the International Cooperative R\&D program (grant no. P0019804, Digital twin based intelligent unmanned facility inspection solutions).
The views expressed in this article do not necessarily represent the views of Wells Fargo.}}}

\author{
Ruijian Han\thanks{Department of Applied Mathematics, The Hong Kong Polytechnic University}
\and Boris Kramer\thanks{Department of Mechanical and Aerospace Engineering, University of California San Diego}
\and Dongjin Lee\footnotemark[3]
\and Akil Narayan\thanks{Scientific Computing and Imaging Institute, and Department of Mathematics, University of Utah}
\and Yiming Xu\thanks{Corporate Model Risk, Wells Fargo
  (\email{yiming2358@gmail.com}).}}

\headers{Multifidelity distribution estimation}{R. Han and B. Kramer and D. Lee and A. Narayan and Y. Xu}
\end{tcbverbatimwrite}
\title{An approximate control variates approach to multifidelity distribution estimation\thanks{Submitted to the editors.
\funding{RH is partially supported by the Hong Kong Research Grants Council (grant no. 14301821) and Start-up Fund for New Recruits, The Hong Kong Polytechnic University.
AN is partially supported by NSF DMS-1848508 and AFOSR FA9550-20-1-0338.
BK and DL are supported by the Ministry of Trade, Industry and Energy (MOTIE) and the Korea Institute for Advancement
of Technology (KIAT) through the International Cooperative R\&D program (grant no. P0019804, Digital twin based intelligent unmanned facility inspection solutions).
The views expressed in this article do not necessarily represent the views of Wells Fargo.}}}

\author{
Ruijian Han\thanks{Department of Applied Mathematics, The Hong Kong Polytechnic University}
\and Boris Kramer\thanks{Department of Mechanical and Aerospace Engineering, University of California San Diego}
\and Dongjin Lee\footnotemark[3]
\and Akil Narayan\thanks{Scientific Computing and Imaging Institute, and Department of Mathematics, University of Utah}
\and Yiming Xu\thanks{Corporate Model Risk, Wells Fargo
  (\email{yiming2358@gmail.com}).}}

\headers{Multifidelity distribution estimation}{R. Han and B. Kramer and D. Lee and A. Narayan and Y. Xu}

\ifpdf
\hypersetup{pdftitle={Guide to Using  SIAM'S \LaTeX\ Style} }
\fi

\usepackage{tikz,pgfplots}
\usetikzlibrary{backgrounds}
\usetikzlibrary{arrows.meta}
\usetikzlibrary{positioning,fit,calc}
\tikzset{fit margins/.style={/tikz/afit/.cd,#1,
    /tikz/.cd,
    inner xsep=\pgfkeysvalueof{/tikz/afit/left}+\pgfkeysvalueof{/tikz/afit/right},
    inner ysep=\pgfkeysvalueof{/tikz/afit/top}+\pgfkeysvalueof{/tikz/afit/bottom},
    xshift=-\pgfkeysvalueof{/tikz/afit/left}+\pgfkeysvalueof{/tikz/afit/right},
    yshift=-\pgfkeysvalueof{/tikz/afit/bottom}+\pgfkeysvalueof{/tikz/afit/top}},
    afit/.cd,left/.initial=2pt,right/.initial=2pt,bottom/.initial=2pt,top/.initial=2pt}
\tikzset{base/.style={rectangle, rounded corners, draw=black, thick, text centered, fill=blue!20},
         innerbase/.style={base, draw=black!30},
         eebox/.style={base, dashed, thick, fill=black!08, inner xsep=0.5cm},
         explotation/.style={base, fill=red!40!white},
         eephase/.style={font=\large\bfseries},
         connector/.style={very thick,-{Latex[width=3mm]}, rounded corners}
         }

\begin{document}
\maketitle

\begin{abstract}
Forward simulation-based uncertainty quantification that studies the distribution of quantities of interest (QoI) is a crucial component for computationally robust engineering design and prediction. There is a large body of literature devoted to accurately assessing statistics of QoIs, and in particular, multilevel or multifidelity approaches are known to be effective, leveraging cost-accuracy tradeoffs between a given ensemble of models. However, effective algorithms that can estimate the full distribution of QoIs are still under active development. In this paper, we introduce a general multifidelity framework for estimating the cumulative distribution function (CDF) of a vector-valued QoI associated with a high-fidelity model under a budget constraint. Given a family of appropriate control variates obtained from lower-fidelity surrogates, our framework involves identifying the most cost-effective model subset and then using it to build an approximate control variates estimator for the target CDF.  We instantiate the framework by constructing a family of control variates using intermediate linear approximators and rigorously analyze the corresponding algorithm. Our analysis reveals that the resulting CDF estimator is uniformly consistent and asymptotically optimal as the budget tends to infinity, with only mild moment and regularity assumptions on the joint distribution of QoIs. The approach provides a robust multifidelity CDF estimator that is adaptive to the available budget, does not require \textit{a priori} knowledge of cross-model statistics or model hierarchy, and applies to multiple dimensions. We demonstrate the efficiency and robustness of the approach using test examples of parametric PDEs and stochastic differential equations including both academic instances and more challenging engineering problems.   
\end{abstract}



\section{Introduction}
Physical systems are often modeled with computational simulations or emulators, and as such, understanding the error in these constructed approximations is of utmost importance. One particular source of uncertainty in the output is due to the input uncertainty in these models, either through uncertainty in model parameters (which can be finite- or infinite-dimensional) or through modeled stochasticity in the system, e.g., systems driven with white noise processes. 
To make the resulting models trustworthy, it is crucial to quantify the resulting uncertainty in QoIs; that is, to estimate the QoI's distribution or some statistical summary of it. One popular approach for achieving this is through Monte Carlo (MC) simulation, which is easy to implement and provides robust results but has a slow convergence rate.
A typical MC procedure requires drawing a large number of samples or running repeated experiments, which is expensive given the increasing complexity of computational simulations.

To address this issue, methods based on multilevel \cite{giles2008multilevel, giles2015multilevel} and multifidelity modeling \cite{Peherstorfer_2016, PKW17MFCE, qian2018multifidelity, HKTQ18CVaRROMS, peherstorfer2019multifidelity, Gorodetsky_2020, HKT2020_Adaptive_ROM_CVAR_estimation, Schaden_2020, farcas2020context, xu2022bandit, schaden2021asymptotic, gruber2022multifidelity, croci2023multi, farcas2023context} have been developed to estimate the statistics of QoIs associated with the (high-fidelity) model. 
The core idea behind multilevel/multifidelity methods lies in leveraging models of different accuracies and costs to improve computational efficiency.
However, a major limitation of the existing literature is that it predominantly focuses on the estimation of the statistical mean of the QoIs (or other scalar-valued descriptive statistics such as quantiles or conditional expectations), providing only partial insight into the uncertainty of the QoIs. A more comprehensive understanding would require assessing, for example, higher-order statistics of the QoI, or even the entire distribution. 

Existing methods to estimate CDFs in the multilevel and multifidelity setup have seen notable success \cite{giles2015multilevel, lu2016improved, giles2017adaptive, krumscheid2018multilevel, ayoul2022quantifying, xu2021budget}.
In \cite{giles2015multilevel}, the authors proposed a multilevel approach to computing the CDFs of univariate random variables arising from stochastic differential equations and derived an upper bound for the cost in terms of the error.
The methodology in \cite{giles2015multilevel} was further developed and applied in several subsequent works \cite{lu2016improved, krumscheid2018multilevel, giles2017adaptive, ayoul2022quantifying}.
In particular,
\cite{lu2016improved} designed an \textit{a posteriori} optimization strategy to calibrate the smoothing function and showed its superiority over MC in oil reservoir simulations; 
\cite{krumscheid2018multilevel} generalized the ideas in \cite{giles2015multilevel} to approximate more general parametric expectations such as characteristic functions; \cite{giles2017adaptive} applied an adaptive approach for parameter selection that yields an improved cost bound;
\cite{ayoul2022quantifying} provides a novel computable error estimator to enhance algorithm tuning. 
Despite the substantive contributions of these approaches, nearly all of them make relatively restrictive assumptions regarding model hierarchy (e.g., the model cost versus accuracy tradeoffs), and do not immediately extend to the general non-hierarchical multifidelity setup.   
For this more general multifidelity estimation of CDFs, the only work we are aware of is the adaptive explore-then-commit algorithm for distribution learning (AETC-d) \cite{xu2021budget}. However, the large-budget performance of AETC-d is restricted by its own set of statistical assumptions that are often too stringent to satisfy in practice. Moreover, the QoI in all the above references is assumed to be a scalar.

An outline of the paper is as follows.  The remainder of this section lists our contributions, introduces overall notation, and summarizes the main theory and algorithmic advances.
Sections \ref{sec:2} through \ref{sec:4} describe the necessary mathematical and statistical background for our method: 
Section \ref{sec:2} gives a brief overview of the control variates method; 
Section \ref{sec:3} introduces a multifidelity CDF estimation framework based on approximate control variates estimators.
Section \ref{sec:4} provides a computational construction for the control variates through linear approximators.
Section \ref{sec:5} develops our new meta algorithm (cvMDL) that accomplishes autonomous model selection together with an algorithmic correction to preserve the monotonicity of the resulting CDF estimators. The meta algorithm cvMDL itself does not specify how to compute the control variates: A specialization to using the linear approximations from Section \ref{sec:4} yields a computationally explicit algorithm that we study in detail, establishing both uniform consistency and budget-asymptotic optimality. Section \ref{sec:6} contains a detailed simulation study and showcases applications that use estimated CDFs to compute probabilistic risk metrics.

\subsection{Contributions}

The main goal of this article is to provide novel solutions that mitigate the deficiencies described above.
We develop an efficient algorithm for estimating the CDF in a general non-hierarchical multifidelity approximation setting under computational budget constraints. The proposed method satisfies the following criteria: 1) it requires as input neither cross-model statistics nor model hierarchy; 2) it can provide distributional estimates for vector-valued QoIs, and 3) it is empirically robust and enjoys theoretical guarantees.  Although our approach uses a similar meta algorithm as in \cite{xu2022bandit, xu2021budget} (all borrowing ideas from the explore-then-commit algorithm in bandit learning \cite{lattimore2020bandit}), it contains a substantial number of new ingredients that extend applicability and improve robustness. In more technical language, our contributions are twofold:

\begin{itemize}[topsep=0pt,itemsep=0pt]
\item We propose a control variates-based exploration-exploitation strategy for multifidelity CDF estimation under a budget constraint.
  The \textit{exploration} step leverages statistical estimation to select a subset of low-fidelity models for the control variates construction, followed by the \textit{exploitation} step that utilizes the learned information to build an approximate control variates estimator for the target CDF. This procedure is initialized with no \textit{a priori} oracle information\footnote{In this article, oracle information refers to model statistics that we treat as exact. These statistics may be exactly computed, but more often are approximations identified through simulations with a large enough computational expense so that the approximations are treated as ground truth.} about model relationships, in contrast to several methods that require such information as input. In addition, our estimator for the CDF applies to both scalar-valued and vector-valued QoI, which differentiates it from existing methods that apply only to scalar-valued QoI.
 \item Through examination of the average weighted-$L^2$ loss that balances errors in exploration and exploitation, we design a new meta algorithm, the control variates multifidelity distribution learning algorithm (``cvMDL'', summarized in \Cref{fig:flowchart} and detailed in \Cref{alg2-detailed}), that accomplishes model (subset) selection and CDF estimation. Using control variates constructed from linear approximators, we establish both uniform consistency and asymptotic optimality of the estimator produced by cvMDL as the budget approaches infinity (\Cref{main}). Our analysis illustrates that the proposed procedure significantly ameliorates the restrictive model assumptions in \cite{xu2021budget}.
\end{itemize}

A verbatim usage of our approaches produces a CDF estimator that enjoys the previously-mentioned theoretical guarantees but is not necessarily monotonic and hence may not be itself a distribution function. To mitigate this artifact, we utilize an empirical algorithmic correction that restores the monotonicity of the estimated CDFs and additionally makes its manipulation more computationally convenient (e.g. for extraction of quantiles and conditional expectations); see \Cref{alg-sort}.
We observe that in some cases this empirical correction further reduces errors. 
 
\subsection{Notation}\label{notation}

For $n\in\N$, let $\{1:n\}:=\{1, \ldots, n\}$. 
We use bold upper-case and lower-case letters to denote matrices and vectors, respectively. The Euclidean ($\ell^2$) norm on a vector $\bm{v}$ is denoted $\|\bm{v}\|_2$. For a matrix $\A$, $\A^\top$ is the transpose and $\A^\dagger$ is the pseudoinverse; $\A^\dagger$ coincides with the regular inverse $\A^{-1}$ when $\A$ is invertible. 
The $i$th column of $\A$ is denoted by $\A^{(i)}$. 
The Frobenius norm of $\A$ is denoted by $\|\A\|_F = (\sum_{i}\|\A^{(i)}\|_2^2)^{1/2}$. 
We use $\otimes$ to denote the tensor product operator. 
For a set $\mathcal T\subseteq\R^d$, we denote its interior as $\mathcal T^\circ$, and $\mathbf 1_{\mathcal T}(\x) := \mathbf 1_{\{\x\in \mathcal T\}}$ as the indicator function on $\mathcal T$. 
For two vectors $\x = (\x^{(1)}, \ldots, \x^{(d)})^\top$ and $\y = (\y^{(1)}, \ldots, \y^{(d)})^\top$, we use $\vee$ and $\wedge$ to denote the componentwise \emph{max} and \emph{min} operators, respectively, i.e., 
\begin{align*}
\x\vee\y &:= \left(\max\{\x^{(1)}, \y^{(1)}\}, \ldots, \max\{\x^{(d)}, \y^{(d)}\}\right)^\top\\
\x\wedge\y &:= \left(\min\{\x^{(1)}, \y^{(1)}\}, \ldots, \min\{\x^{(d)}, \y^{(d)}\}\right)^\top.
\end{align*}
Moreover, we say $\x\leq \y$ if $\x^{(i)}\leq \y^{(i)}$ for all $i\in \{1:d\}$. 
We consider the QoIs from computational models as random variables that jointly lie in some common probability space $(\Omega, \mathcal F, \P)$. For a random vector $X\in\R^d$, we let $F_X(\x) = \P(X^{(1)}\leq \x^{(1)}, \ldots, X^{(d)}\leq \x^{(d)})$ denote its CDF. 
For two sequences of random variables $\{a_m(\omega)\}$ and $\{b_m(\omega)\}$ where $\omega \in \Omega$ is a probabilistic event, we write $a_m(\omega)\lesssim b_m(\omega)$ if almost surely (a.s.), $a_m(\omega)\leq \eta(\omega)b_m(\omega)$ for all $m\in\N$, where the constant $\eta(\omega)$ is independent of $m$. For convenience, we let 
\begin{align*}
&Y = (Y^{(1)}, \ldots, Y^{(d)})^\top\in\R^{d}& X_i = (X_i^{(1)}, \ldots, X_i^{(d_i)})^\top\in\R^{d_i}\ \ \ \ \ \ \ i\in \{1:n\}
\end{align*}
denote the high-fidelity and the $i$th low-fidelity QoIs, respectively.
Here $d, d_i\in\N$ are the corresponding dimensions of $Y$ and $X_i$.
There are $n$ low-fidelity models in total. 
We use $\E[\cdot ], \text{Var}[\cdot]$/$\text{Cov}[\cdot ]$, and $\text{Corr}[\cdot ]$ to denote the expectation, variance/covariance, and correlation operators respectively.
We use $\indep$ to represent probabilistic independence. 

\subsection{Model assumptions}

We assume the sampling costs for $Y$ and $X_1, \ldots, X_n$, denoted by positive numbers $c_0$ and $c_1, \ldots, c_n$, are deterministic and known.
For $\S\subseteq \{1:n\}$, let $c_\S = \sum_{i\in\S}c_i$, corresponding to the cost of sampling all (low-fidelity) models from subset $\S$. 
We let $B>0$ be the total budget (deterministic and known) that is available to expend on sampling the models.
Moreover, for $\S\subseteq \{1:n\}$, we let $X_\S = (X_i^\top)^\top_{i\in\S}\in\R^{d_\S}$, where $d_\S = \sum_{i\in\S}d_i$, and $X_{\S^+} = (1, X_\S^\top)^\top$, where the latter is used when considering a linear model approximation with the intercept/bias term. 

The central goal in the rest of the article is to develop a multifidelity estimator for $F_Y(\x)$ through drawing samples of $(Y, X_{\{1:n\}})$ and of $X_\S$ for some adaptively-determined $\S \subseteq \{1:n\}$, subject to the sampling cost not exceeding the total budget constraint $B>0$. No other high-level assumptions are made. In other words, we assume only that $Y$ is a known high-fidelity model; we do not assume any ordering/hierarchy in the models $X_{\{1:n\}}$, and we do not assume known statistics (e.g., correlations) between any models. 
While such generality is sufficient for algorithmic purposes, our theoretical guarantees require additional technical assumptions that are articulated in \Cref{ssec:theory}. These technical assumptions are mild regularity conditions, related to finite moments of random variables and CDF functional regularity.

The notation we have introduced is enough to present the overall cvMDL algorithm in the next section.
The actual computations that make the algorithm practical, however, require more technical details which are provided in \Cref{sec:2} through \Cref{sec:5}.

\subsection{Summary of the algorithm}\label{ssec:cvmdl-summary}

The proposed cvMDL meta algorithm is shown in \Cref{fig:flowchart}. In summary, we first gather $m$ full joint samples of $(Y, X_{\{1:n\}})$ through an \textit{exploration phase} that identifies (i) how models are related, (ii) which model subset $\S$ optimally balances cost versus accuracy, and (iii) whether more samples $m$ are needed to certify a robust exploration or whether the choice of $\S$ is statistically robust enough to proceed with exploitation. 
Exploration is followed by the \textit{exploitation phase}, where we exhaust the remaining computational budget to sample the optimal model subset $X_\S$. Exploitation corresponds to exercise of a particular approximate control variates estimator for $F_Y$. A more detailed description is as follows:
\begin{itemize}[topsep=3pt,itemsep=0pt,leftmargin=0pt]
  \item[] \textbf{Exploration phase}
    \begin{itemize}[topsep=0pt, itemsep=0pt, leftmargin=15pt]
    \item \textit{Mininum exploration}: This step ensures that the number of exploration samples $m$ is set large enough so that non-degenerate empirical statistics can be computed.
    \item \textit{Analyze low-fidelity models}: We are interested in estimating the minimal loss associated with an estimated CDF that utilizes the model subset $\S$. For such a goal, this step identifies for each model subset $\S$ both an estimated number of optimal exploration samples $\widehat{m}^\ast_\S$ along with the corresponding loss function minimum $\widehat{L}(m\vee \widehat{m}^\ast_\S; m)$. The value $\widehat{L}(z; m)$ is an estimator with the currently-available $m$ exploration samples and measures the estimated loss if we eventually use $z$ exploration samples. When evaluating the minimum loss, we require the input $z \gets m\vee \widehat{m}^\ast_\S$ since if $m > \widehat{m}^\ast_\S$ then the number of exploration samples should be $m$, and not $\widehat{m}^\ast_\S$ (we cannot take fewer exploration samples than already committed and we assume $\widehat{L}(z; m)$ is convex and has a unique minimizer). The definitions of $\widehat{L}$, $\widehat{m}_\S^\ast$, and $\widehat{\S}^\ast$ are given in \eqref{lm} and \eqref{eq:Shatast}.
    \item \textit{Select optimal model}: The estimated optimal model subset $\widehat{\S}^\ast$ is computed by choosing the subset $\S$ with minimal loss from the previous step.
    \item \textit{Continue exploration}: If the current number of exploration samples $m$ is smaller than the estimated optimal number of samples $\widehat{m}^\ast_{\widehat{\S}^\ast}$ required for the optimal subset $\widehat{\S}^\ast$, then we continue exploration, with the precise number of additional exploration samples determined by the function $Q(\cdot,\cdot)$ that is defined in \eqref{eq:Q-def}. If $m \geq \widehat{m}^\ast_{\widehat{\S}^\ast}$, then exploration terminates and we move to the exploitation phase.
  \end{itemize}
  \item[] \textbf{Exploitation phase} 
  \begin{itemize}[topsep=0pt, itemsep=0pt, leftmargin=15pt]
    \item \textit{Expend budget}: After exploration terminates and an ``optimal'' model subset $\widehat{\S}^\ast$ has been identified, we expend the remaining computational budget on sampling $X_{\widehat{\S}^\ast}$.
    \item \textit{Estimate CDF}: Using the collected samples, we construct the CDF estimator $\widetilde{F}_{\widehat{\S}^\ast}$ for $F_Y$, which is defined in \eqref{fug}.
  \end{itemize}
\end{itemize}
The precise details of how the loss function is computed and the CDF estimator is constructed is the topic of \Cref{sec:5}, with \Cref{sec:2,sec:3,sec:4} serving to make requisite mathematical and statistical definitions.

A more detailed version of the algorithm is given in \Cref{alg2-detailed}, which lists more explicit computational steps that must be taken. The coming sections are devoted to the theoretical construction of quantities in \Cref{fig:flowchart}; in particular \Cref{sec:2,sec:3} provide a construction of a loss function that is the integral part of exploration decision-making.

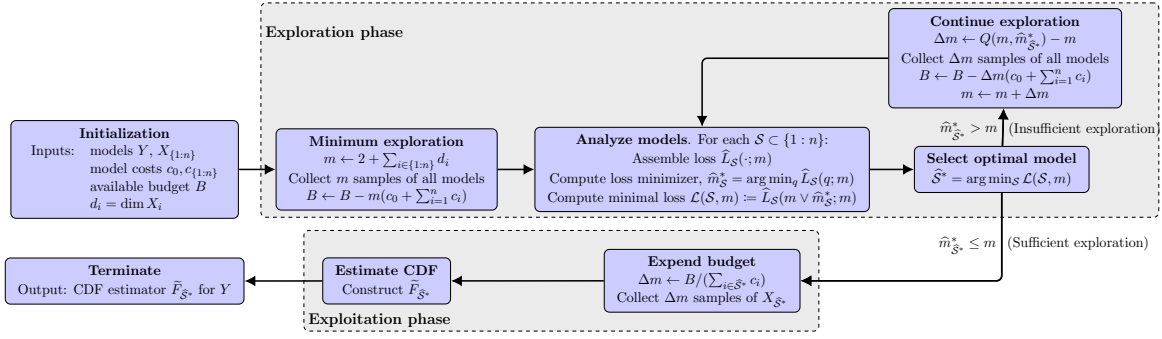
\begin{figure}
  \resizebox{\textwidth}{!}{
    \begin{tikzpicture}

  \node (start) [base] {\begin{tabular}{c} \textbf{Initialization} \\ \begin{tabular}{rl} Inputs: & models $Y$, $X_{\{1:n\}}$ \\ & model costs $c_0, c_{\{1:n\}}$ \\ & available budget $B$  \\ & $d_i = \dim X_i$ \end{tabular}\end{tabular}};
  \node (minexplore) [base, right=of start.east, anchor=west] {\begin{tabular}{c} \textbf{Minimum exploration} \\ $m \gets 2 + \sum_{i \in \{1:n\}} d_i$ \\
    Collect $m$ samples of all models \\
    $B \gets B - m (c_0 + \sum_{i=1}^n c_i)$ \end{tabular}};

    \node (modelanalysis) [base, right=of minexplore.east] {\begin{tabular}{c} \textbf{Analyze models}. For each $\S \subset \{1:n\}$: \\ Assemble loss $\widehat{L}_{\S}(\cdot;m)$ \\ Compute loss minimizer, $\widehat{m}^\ast_{\S} = \arg\min_q \widehat{L}_{\S}(q;m)$ \\ Compute minimal loss $\mathcal{L}(\S,m) \coloneqq \widehat{L}_{\S}(m\vee \widehat{m}^\ast_{\S};m)$ \end{tabular}};

    \node (modelselect) [base, right=of modelanalysis.east, xshift=0.2cm] {\begin{tabular}{c}\textbf{Select optimal model} \\ $\widehat{\S}^\ast = \arg\min_{\S} \mathcal{L}(\S,m)$ \end{tabular}};

    \node (exploremore) [base, above=of modelanalysis.north east, anchor=south west, xshift=0.5cm, yshift=-0.5cm] {\begin{tabular}{c} \textbf{Continue exploration} \\ $\Delta m \gets Q(m,\widehat{m}^\ast_{\widehat{\S}^\ast}) - m$ \\ Collect $\Delta m$ samples of all models \\ $B \gets B - \Delta m (c_0 + \sum_{i=1}^n c_i)$ \\ $m \gets m + \Delta m$ \end{tabular}};

    \node (exploit) [base, below=of modelanalysis.south, anchor=north] {\begin{tabular}{c} \textbf{Expend budget} \\ $\Delta m \gets B/(\sum_{i \in \widehat{\S}^\ast} c_i)$ \\
    Collect $\Delta m$ samples of $X_{\widehat{\S}^\ast}$\end{tabular}};

    \node (estimate) [base, anchor=base, yshift=-1mm] at (minexplore |- exploit.west) {\begin{tabular}{c} \textbf{Estimate CDF} \\ Construct $\widetilde{F}_{\widehat{\S}^\ast}$ \end{tabular}};
    \node (end) [base, anchor=base, yshift=-1mm] at (start |- estimate.west) {\begin{tabular}{c} \textbf{Terminate} \\ Output: CDF estimator $\widetilde{F}_{\widehat{\S}^\ast}$ for $Y$ \end{tabular}};

  \begin{pgfonlayer}{background}
    \node (exploration) [eebox, fit margins={left=0.2cm,right=0.5cm}, fit={(minexplore) (modelanalysis) (exploremore)}] {};
    \node (exploitation) [eebox, inner ysep=0.5cm, fit={(exploit) (estimate)}] {};
  \end{pgfonlayer}

  \node [anchor=north west, yshift=-0.5cm, eephase] at (exploration.north west) {Exploration phase};
  \node [anchor=south west, eephase] at (exploitation.south west) {Exploitation phase};

  \draw [connector] (start) -- (minexplore);
  \draw [connector] (minexplore) -- (modelanalysis);
  \draw [connector] (modelanalysis) -- (modelselect);
  \draw [connector] (modelselect) -- node [pos=0.4, left] {$\widehat{m}^\ast_{\widehat{\S}^\ast} > m$} (exploremore);
  \draw [connector] (modelselect) -- node [pos=0.44, right] {(Insufficient exploration)} (exploremore);
  \draw [connector] (exploremore) -| (modelanalysis);
  \draw [connector] (modelselect) |- node [pos=0.30, left] {$\widehat{m}^\ast_{\widehat{\S}^\ast} \leq m$} (exploit);
  \draw [connector] (modelselect) |- node [pos=0.29, right] {(Sufficient exploration)} (exploit);
  \draw [connector] (exploit) -- (estimate);
  \draw [connector] (estimate) -- (end);

\end{tikzpicture}
  }
  \caption{Flowchart illustration of the cvMDL algorithm. More details of the steps are discussed in \Cref{ssec:cvmdl-summary}. The full algorithm is presented in \Cref{alg2-detailed}.}\label{fig:flowchart}
\end{figure}

\section{Background: control variates}\label{sec:2}

We first introduce the control variates method, which is a standard approach for variance reduction in MC simulation. 
For a random variable $X$ with bounded variance $\sigma_X^2>0$, the size-$m$ MC estimator for $\E[X]$ based on i.i.d. data $X_\ell$, $\widehat{x} = \sum_{\ell\in \{1:m\}} X_\ell/m$, is unbiased and has variance $\sigma_X^2/m$. 
Given a random vector $Z = (Z^{(1)}, \ldots, Z^{(d)})^\top\in\R^d$ that lives in the same probability space as $X$, one may use joint i.i.d samples of $(X, Z)$, i.e., $(X_\ell, Z_\ell^\top) = (X_\ell, Z^{(1)}_{\ell}, \ldots,  Z^{(d)}_{\ell})$ for $\ell\in \{1:m\}$, to construct a control variates estimator $\widehat{x}_\text{cv}$ for $\E[X]$:
\begin{align*}
\widehat{x}_{\text{cv}} = \frac{1}{m}\sum_{\ell\in \{1:m\}}X_\ell -  \frac{1}{m}\sum_{\ell\in \{1:m\}}(Z_\ell^\top\beta - \E[Z]^\top\beta),
\end{align*}
where $\beta\in\R^d$ is some appropriately chosen vector and $\E[Z]$ is assumed known. The estimator $\widehat{x}_{\text{cv}}$ is also unbiased and has variance 
\begin{align*}
\sigma_{\text{cv}}^2 = \text{Var}[\widehat{x}_{\text{cv}}] = \frac{\text{Var}[X-Z^\top\beta]}{m}.
\end{align*}
This variance is minimized when $\beta$ is the least-squares coefficient for centered linear regression, i.e., for regressing $(X-\E[X])$ on $(Z-\E[Z])$,
{\small
\begin{align}\label{mybeta}
\beta = \text{Cov}[Z]^{-1}\text{Cov}[Z, X] \hskip 5pt \Longrightarrow \hskip 5pt
\sigma_{*}^2 = \min_{\beta\in\R^d}\sigma_{\text{cv}}^2  = \frac{(1-\rho^2)\sigma_X^2}{m}, \hskip 10pt \rho = \text{Corr}(X, Z^\top\text{Cov}[Z]^{-1}\text{Cov}[Z, X]).
\end{align}
}
When $|\rho|\approx 1$, the variance reduction is significant, in which case 
$Z^\top\beta$ 
accounts for most of $\text{Var}[X]$. 

When $\E[Z]$ is unknown, one may consider the following \emph{approximate control variates} estimator that uses an independent size-$N$ MC estimator in place of $\E[Z]$ using samples $\widetilde{Z}_j$: 
\begin{align}
&\widehat{x}_{\text{acv}} = \frac{1}{m}\sum_{\ell\in \{1:m\}}X_\ell -  \frac{1}{m}\sum_{\ell\in \{1:m\}}\left(Z_\ell^\top\beta - \frac{1}{N}\sum_{j\in \{1:N\}}\widetilde{Z}_{j}^\top\beta\right)&
  (\ell, j)\in \{1:m\}\times \{1:N\}\label{xiaodaihua},
\end{align}
where we assume $Z_{\ell}\indep \widetilde{Z}_{j}$. Then this has variance 
\begin{align}
\sigma^2_{\star} = \sigma_{*}^2 + \frac{\text{Var}[Z^\top\beta]}{N} = \frac{(1-\rho^2)\sigma_X^2}{m} + \frac{\rho^2\sigma_X^2}{N}.\label{1}
\end{align}
Construction of such approximate control variates estimators has been recently studied in the multifidelity estimation of first-order statistics \cite{Gorodetsky_2020, xu2022bandit}. 
The terms $\text{Cov}[Z]^{-1}$ and $\text{Cov}[Z, X]$ may be estimated empirically at the cost of incurring higher-order trajectory-wise statistical errors in $m$ \cite{glasserman2004monte, owen2013monte}.

\section{Variance reduction for CDF estimation}\label{sec:3}

Control variates can be more generally applied to CDF estimation of nonlinear functions of random variables \cite{glasserman2000efficient, hesterberg1998control}. 
For example, in risk management applications \cite{glasserman2000efficient}, the authors considered using the \emph{delta-gamma} approximation\footnote{Here we refer to the ``full'' delta-gamma approximation. The more commonly used delta-gamma approximation in practice does not consider the second-order cross terms.} (i.e. the second-order Taylor expansion) of a loss function $L$ at a given position $\bm x$ along random market move direction $\bm\eta$ as control variates to compute its quantiles.
More precisely, one uses a quadratic function of $\bm\eta$ to approximate the loss at $\bm x$:
\begin{align*}
-(L(\bm x + \bm\eta)-L(\bm x)) \eqqcolon \ell(\bm x) \approx \widehat{\ell}(\bm x)\coloneqq -\nabla L (\bm x)^\top\bm\eta -\frac{1}{2}\bm\eta^\top\nabla^2 L (\bm x)\bm\eta.
\end{align*} 
Fixing a scalar $C$, $\bm 1_{\{\widehat{\ell}(\bm x)\leq C\}}$ can be used as a control variate for $\bm 1_{\{\ell(\bm x)\leq C\}}$ to compute the latter's expectation, which in particular provides a way to compute CDFs.
More advanced approximation techniques have been introduced in \cite{hesterberg1998control} to construct other control variates in the value-at-risk computation. 

We apply a similar idea in the proposed multifidelity setup here. 
In our setup, a specific functional form may be computationally difficult to produce, and Taylor-like approximations can be inaccurate outside local regions.  
Our alternative strategy is to employ a global emulator for $Y$ based on linear combinations of $X_{\{1:n\}}$, which can be effective when the correlation between these quantities is high.
For example, this situation is often true when modeling parametric PDEs.
In the rest of the section, we introduce a general multifidelity approach to estimate $F_Y(x)$ subject to a budget constraint.

\subsection{Control variates for multifidelity CDF estimation}\label{sec:3.1}
 
In developing the proposed method, we frequently resort to the simple observation that  
\begin{align*}
F_Y(\bm x) = \E[\mathbf 1_{\{Y\leq \bm x\}}], \qquad\bm x\in\R^{d}.
\end{align*}
If we fix $\S\subseteq \{1:n\}$, the control variate based on $X_\S$ that minimizes variance (and hence is optimal) is $\E[\mathbf 1_{\{Y\leq \bm x\}}|X_\S]$ \cite{rao1973linear}.
This quantity requires the orthogonal projection of $Y$ onto the sigma-field generated by $X_\S$, which is computationally intractable without special assumptions (e.g. joint normality). In order to approximate $\E[\mathbf 1_{\{Y\leq \bm x\}}|X_\S]$,
we use $h(X_\S; \bm x)$ to denote a general $X_\S$-measurable function that serves as the control variates for $\mathbf 1_{\{Y\leq \bm x\}}$.  We make a particular choice for $h$ in \Cref{sec:4}.

Analogous to \eqref{xiaodaihua}, we construct an approximate control variates estimator for $F_Y(\x)$, where the $m$ and $N$ in \eqref{xiaodaihua} are related by the budget constraint (the cost of sampling $Y$ and $X_\S$).
Since different subsets $\S$ are considered simultaneously, we take a \emph{uniform exploration policy} that first collects $m$ i.i.d joint exploration samples of the \emph{full} model
for variance reduction and then commits the remainder of the budget to collect $N_\S$ samples of a selected model subset $\S$ of low-fidelity models to compute the control variates mean.
The exploration samples and exploitation samples under a uniform exploration policy are denoted:
\begin{align}
  \text{Exploration samples:} &\; \{(X^\top_{\ex, \ell, 1}, \ldots, X^\top_{\ex, \ell, n}, Y^\top_{\ex, \ell})^\top\}_{\ell\in \{1:m\}}\subset \R^{d + \sum_{i=1}^n d_i}\label{onta1}\\
  \text{Exploitation samples:} &\; \{X_{\ext, j, \S}\}_{j\in \{1:N_\S\}},\label{onta2}
\end{align}
where the subscripts ``$\ex$'' and ``$\ext$'' specify the stage where a sample is used. 
The parameters $m$ and $N_\S$ are related by the budget constraint:  
\begin{align}\label{eq:NS-def}
  N_\S &= \frac{B-c_{\text{epr}}m}{c_\S}& c_{\text{epr}} &= \sum_{i=0}^nc_i, &  c_\S &= \sum_{i\in\S}c_i, 
\end{align}
where we ignore integer rounding effects to simplify the discussion. 
The control variates estimator for $F_Y(\x)$ based on $h(X_\S; \x)$ is
\begin{align}
\widehat{F}_\S(\x) = \frac{1}{m}\sum_{\ell\in \{1:m\}}\mathbf 1_{\{Y_{\ex,\ell}\leq \x\}} -  \frac{1}{m}\sum_{\ell\in \{1:m\}}\alpha(\x)\left(h(X_{\ex, \ell, \S}; \x) - \frac{1}{N_\S}\sum_{j\in \{1:N_\S\}}h(X_{\ext, j, \S}; \x)\right),\label{panger}
\end{align}
where $\alpha(\x)$ is the optimal scaling coefficient as in \eqref{mybeta}:
\begin{align}
\alpha(\x) = \text{Cov}[h(X_{\S}; \x)]^{-1}\text{Cov}[\mathbf 1_{\{Y\leq \x\}}, h(X_{\S}; \x)].\label{alpha}
\end{align}
Note $\alpha(\x)$ is undefined if $\text{Cov}[h(X_{\S}; \x)]=0$.
In this case, the value of the estimator $\widehat{F}_\S(\x)$ does not depend on $\alpha(\x)$, and we set $\alpha(\x)$ to $0$ for convenience. 
The quantity $\widehat{F}_\S(\x)$ is an unbiased estimator for $F_Y(\x)$ with variance 
\begin{align*}
\text{Var}[\widehat{F}_\S(\x)] = \frac{(1-\rho_\S^2(\x))F_Y(\x)(1-F_Y(\x))}{m} + \frac{\rho_\S^2(\x)F_Y(\x)(1-F_Y(\x))}{N_\S},
\end{align*}
where
\begin{align}\label{eq:rhoS-def}
\rho_\S(\x) = \text{Corr}[\mathbf 1_{\{Y\leq \x\}}, h(X_{\S}; \x)].
\end{align} 

\subsection{A control variates loss function}\label{ssec:loss}

To measure the overall accuracy of $\widehat{F}_\S(\x)$, we introduce the loss $L_\S$ defined by the average $\omega(\x)$-weighted $L^2$-norm square of $\widehat{F}_\S(\x)-F_Y(\x)$:
\begin{align}\label{eq:loss-def}
L_\S \coloneqq \E\left[\int_{\R^d}\omega(\x)|\widehat{F}_\S(\x)-F_Y(\x)|^2 \text{d}\x\right],
\end{align}
where $\omega(\x): \R^d\to\R_{\geq 0}$ is a weight function. 
The $\omega(\x)$-weighted $L^2$-norm square is related to other more widely used metrics on distributions, 
e.g., it reduces to the Cram\'{e}r--von Mises distance when $\omega(\x)\text{d}\x=\text{d}F_Y(\x)$. 
To estimate $L_\S$, note that by Tonelli's theorem, we have, 
\begin{align}
  L_\S = \int_{\R^d} \omega(\x) \text{Var}[\widehat{F}_\S(\x)] \text{d}\x =  \frac{k_1(\S)}{m} + \frac{k_2(\S)}{B-c_{\text{epr}}m} \label{loss},
\end{align}
where 
\begin{align}\label{eq:k12-def}
k_1(\S) &= \int_{\R^d}\omega(\x)(1-\rho_\S^2(\x))F_Y(\x)(1-F_Y(\x)) \text{d}\x \nonumber\\
k_2(\S) &= c_\S\int_{\R^d}\omega(\x)\rho_\S^2(\x)F_Y(\x)(1-F_Y(\x))\text{d}\x.
\end{align}
Since $k_1(\S)$ and $k_2(\S)$ are nonnegative, a sufficient and necessary condition for $k_1(\S)$ and $k_2(\S)$ being well-defined (i.e. finite) is 
\begin{align}
k_1(\S) + c_\S^{-1}k_2(\S) = \int_{\R^d}\omega(\x) F_Y(\x)(1-F_Y(\x))\text{d}\x <\infty, \label{pangwa}
\end{align}
but this need not hold for arbitrary choice of $\omega$. 
For instance, when $\omega(\x)\equiv 1$, \eqref{pangwa} is true when $d = 1$ if $\E[|Y|^{1+\delta}]<\infty$ for some $\delta>0$.
However, when $d\geq 2$, \eqref{pangwa} is generally not true when the support for the distribution of $Y$ is unbounded since $F_Y^{-1}([\e, 1-\e])$ may have infinite Lebesgue measure in $\R^d$ for some $\e>0$.  
For such scenarios, requiring that $\omega(\x)$ is integrable ensures \eqref{pangwa}, i.e.,
\begin{align*}
\int_{\R^d}\omega(\x) F_Y(\x)(1-F_Y(\x))\text{d}\x<\int_{\R^d}\omega(\x)\text{d}\x<\infty. 
\end{align*}
Some typical choices for integrable $\omega(\x)$ include $\omega(\x) = \mathbf 1_{\mathcal T}$ where $\mathcal T\subset \R^d$ is a bounded domain of interest or $\omega(\x)$ with reasonably fast decaying tails as $\|\x\|_2\to\infty$. In the following discussion, we assume \eqref{pangwa} holds (and later codify this as \Cref{ass:omega}). We make different choices for $\omega$ in our numerical results of \Cref{sec:6}.

\subsection{Exploration-exploitation trade-off}\label{sec:22}

Equation \eqref{loss} is similar to the exploration-exploitation loss trade-off that was originally formulated in \cite{xu2022bandit}, where $k_1$ and $k_2$ measure the errors committed by the exploration and the exploitation, respectively. 
Note that $L_\S$ is a strictly convex function for a valid exploration rate $m$, i.e., for $0 < m < B/{c_{\textrm{epr}}}$, and achieves its unique minimum at $m^*_\S$ with corresponding minimum loss $L_\S^*$:
\begin{align}\label{optm}
  m^*_\S &= \frac{B}{c_{\text{epr}}+\sqrt{\frac{c_{\text{epr}}k_2(\S)}{k_1(\S)}}}, & 
  L_\S^* \coloneqq \min_{0<m<\frac{B}{c_{\text{epr}}}}L_\S(m) &= \frac{(\sqrt{c_{\text{epr}}k_1(\S)} + \sqrt{k_2(\S)})^2}{B} 
      \eqqcolon \frac{\gamma_\S}{B}. 
\end{align}
An optimal subset $\S$ is the one that minimizes the $m$-optimized loss value,
\begin{align}\label{eq:Sopt-oracle}
\S^* = \arg\min_{\S\subseteq \{1:n\}}\gamma_\S.
\end{align}
A uniform exploration policy is called \emph{optimal} if it collects $m^*_{\S^*}$ joint samples for exploration and uses model $\S^*$ for exploitation. 
This is, in effect, a model selection procedure, as an optimal exploration policy selects the model subset that yields the smallest error via optimally balancing the trade-off between exploration and exploitation. In the following discussion, we assume $\S^*$ is unique.   

As a benchmark to the procedure above (with oracle information), one can consider an empirical (ECDF) procedure that devotes the full budget to sampling the high-fidelity model $Y$, ignoring the lower-fidelity models. The following result relates the error between these two approaches.
\begin{Lemma}\label{lem:rel-efficiency}
  With $L^\ast_{\S^\ast} = \gamma_{\S^*}/B$ the minimum error achieved by a uniform exploration policy as described above, and $c_\ex\int_{\R^d}\omega(\x) F_Y(\x)(1-F_Y(\x))\text{d}\x/B$ the expected error achieved by an ECDF estimator for $F_Y$, then 
  \begin{align*}
    \frac{c_\ex\int_{\R^d}\omega(\x) F_Y(\x)(1-F_Y(\x))\text{d}\x/B}{\gamma_{\S^*}/B} \geq \frac{1}{2\left(\frac{c_\S}{c_\ex} + \E_{Z}[(1-\rho_{\S^*}^2(Z))]\right)} \geq\frac{1}{4}
  \end{align*}
  where $Z$ is a random variable with (unnormalized) density $ \omega(\z) F_Y(\z)(1-F_Y(\z))$.  
\end{Lemma}

\begin{proof}
We have
\begin{align*}
\frac{c_\ex\int_{\R^d}\omega(\x) F_Y(\x)(1-F_Y(\x))\text{d}\x/B}{\gamma_{\S^*}/B}&= \frac{c_\ex\int_{\R^d}\omega(\x)F_Y(\x)(1-F_Y(\x))\text{d}\x}{\left(\sqrt{c_{\text{epr}}k_1(\S^*)}+ \sqrt{k_2(\S^*)}\right)^2}\nonumber\\
&\geq \frac{c_\ex\int_{\R^d}\omega(\x)F_Y(\x)(1-F_Y(\x))\text{d}\x}{2\left(c_{\text{epr}}k_1(\S^*) + k_2(\S^*)\right)}\nonumber\\
&\geq \frac{1}{2\left(\frac{c_\S}{c_\ex} + \E_{\x}[(1-\rho_{\S^*}^2(\x))]\right)} \geq\frac{1}{4}, 
\end{align*}
where the expectation $\E_\x[\cdot]$ is taken with respect to 
\begin{align*}
    \x\sim \frac{\omega(\x) F_Y(\x)(1-F_Y(\x))\text{d}\x}{\int_{\R^d} \omega(\z) F_Y(\z)(1-F_Y(\z))\text{d}\z},
\end{align*}
and the last inequality follows by noting $c_\S\leq c_\ex$ and $0\leq \E_{\x}[(1-\rho_{\S^*}^2(\x))]\leq 1$. 
\end{proof}
Hence, the relative efficiency of a uniform exploration policy compared to the ECDF estimator is unconditionally bounded below by $1/4$, and hence the uniform exploration policy can at worst realize a loss value of $4$ times a naive ECDF procedure. On the other hand, the relative efficiency is $\gg 1$ if both $c_\S/c_\ex$ and $\E_{\x}[(1-\rho_{\S^*}^2(\x))]$ are small. 
This happens, for instance, if $X_{\S^*}$ has a much smaller sampling cost than $Y$ and $h(X_\S; \x)$ are ``good'' control variates for $\mathbf 1_{\{Y\leq \x\}}$ uniformly for all $\x\in\R^d$, both of which are realistic occurrences in multifidelity applications.  

\section{Choosing control variates from linear approximations}\label{sec:4}

We propose a procedure for selecting the control variate $h$, which boils down to constructing approximations of $\E[\mathbf 1_{\{Y\leq\x\}}|X_\S]$ that both retain high correlation with $Y$ and are budget-friendly.
While one may generate special forms for approximations in particular cases, our goal is a simple and generic choice that is useful for many practical applications.

Recall that $X_{\S^+} = (1, X^\top_i)_{i\in\S}^\top\in\R^{d_\S + 1}$.
For $i\in \{1:d\}$, let $\bm{\beta}^{(i)}_{\S^+}$ be the optimal linear projection coefficients for estimating the $i$th component of $Y$ using $X_{\S^+}$:
\begin{align}\label{eq:betaplus}
  \bm{\beta}^{(i)}_{\S^+} &= (\E[X_{\S^+}X_{\S^+}^\top])^{-1}\text{Cov}[X_{\S^+}, Y^{(i)}]\in\R^{d_\S + 1}, & 
  \bm B_{\S^+} &= [\bm{\beta}^{(1)}_{\S^+}, \cdots, \bm{\beta}^{(d)}_{\S^+}]\in\R^{(d_\S+1)\times d},
\end{align}
The least squares approximation of $Y$ using linear combinations of $X_\S$ and $1$ is given by
\begin{align*}
  H_\S(X_\S) := (X_{\S^+}^\top\bm B_{\S^+})^\top = \begin{bmatrix}
X_{\S^+}^\top \bm\beta^{(1)}_{\S^+}\\
\vdots\\
X_{\S^+}^\top \bm\beta^{(d)}_{\S^+}
\end{bmatrix}\in\R^d. 
\end{align*}
When all quantities are scalars, i.e., $d=d_1 = \cdots = d_n =1$, one can directly manipulate $H_\S$ to estimate the statistics of $Y$ \cite{xu2022bandit, xu2021budget}. 
Such an approach is easy to implement and enjoys certain robustness for first-order statistics, but is more prone to model misspecification effects (e.g. expressibility of the linear model, noise assumption, etc.) when the whole distribution of $Y$ is to be learned due to the limitation of linear approximation \cite{xu2021budget}.

To address the issue, we take an additional nonlinear step beyond $H_\S$. 
In particular, we consider the following family of control variates that slice the estimator $H_\S$:
\begin{align}
h(X_\S; \x) = \mathbf 1_{\{H_\S\leq \x\}}.\label{linear:h}
\end{align} 
Intuitively, we may expect $\mathbf 1_{\{Y\leq \x\}}$ and $h(X_\S; \x)$ to be correlated if $\E[\|Y-H_\S\|^2_2]$ is small.
However, this may not be true for $\x$ approaching the tails of $Y$.
For instance, assuming $d=1$ and a standard joint Gaussian random vector $(X, Y)$ with correlation $\rho$, 
\begin{align*}
  \lim_{x\to-\infty}\text{Corr}(\mathbf 1_{\{Y\leq x\}}, \mathbf 1_{\{X\leq x\}}) &= \lim_{x\to 0}\frac{C_{X,Y}(x,x)}{x} = \left\{ \begin{array}{rl} 1, & |\rho| = 1, \\ 0, & |\rho| < 1 \end{array}\right.
\end{align*}  
where $C_{X, Y}(x, y) = \P\left(\Phi^{-1}(X)\leq x, \Phi^{-1}(Y)\leq y\right)$ is the Gaussian copula and $\Phi^{-1}$ is the quantile of a standard normal distribution; see \cite{mcneil2015quantitative}. Hence, $\mathbf 1_{\{X \leq x\}}$ is not a good control variate for $\mathbf 1_{\{Y\leq x\}}$ when $|x|\to\infty$ unless $|\rho| = 1$, i.e., only if $X \propto Y$.
Nevertheless, our experiments in Section \ref{sec:6} show that in practice $h(X_\S; \x)$ provides a reasonable control variates choice for many scenarios in multifidelity simulations, and thus suggests that situations described above are less common for the applications of our interest. We discuss computational aspects of using \eqref{linear:h} as control variates in Section \ref{sorted-comp}.

Choosing $h$ as in \eqref{linear:h}, the coefficient $\alpha(\x)$ in \eqref{alpha} can be explicitly computed as 
\begin{align}\label{here}
\alpha(\x) = 
\begin{cases}
\frac{F_{Y\vee H_\S}(\x) - F_{Y}(\x)F_{H_\S}(\x)}{F_{H_\S}(\x)(1-F_{H_\S}(\x))}&\x\in\text{supp}(F_{H_\S}(\x))^\circ
\\
0& \text{otherwise}
\end{cases}
.
\end{align}
One useful technical result is that $\alpha(\x)$ is bounded.
\begin{Lemma}\label{kang}
Let $\alpha(\x)$ be given as in \eqref{here}. Then, $|\alpha(\x)|\leq 1$. 
\end{Lemma}
\begin{proof}
It suffices to check that for $\x\in\text{supp}(F_{H_\S}(\x))^\circ$, $\alpha(\x)\leq 1$ and $-\alpha(\x)\leq 1$ hold simutaneously: 
\begin{align*}
\frac{F_{Y\vee H_\S}(\x) - F_{Y}(\x)F_{H_\S}(\x)}{F_{H_\S}(\x)(1-F_{H_\S}(\x))}&\leq\frac{F_{Y}(\x)\wedge F_{H_\S}(\x) - F_{Y}(\x)F_{H_\S}(\x)}{F_{H_\S}(\x)(1-F_{H_\S}(\x))}\\
&= \frac{F_Y(\x)}{F_{H_\S}(\x)}\wedge \frac{1-F_Y(\x)}{1-F_{H_\S}(\x)}\leq 1,
\end{align*}
and
\begin{align*}
\frac{F_{Y}(\x)F_{H_\S}(\x) - F_{Y\vee H_\S}(\x)}{F_{H_\S}(\x)(1-F_{H_\S}(\x))}&\leq\frac{F_{Y}(\x)F_{H_\S}(\x) - F_{Y}(\x) + 1 - F_{H_\S}(\x)}{F_{H_\S}(\x)(1-F_{H_\S}(\x))}\wedge \frac{F_{Y}(\x)F_{H_\S}(\x)}{F_{H_\S}(\x)(1-F_{H_\S}(\x))}\\
&\leq\frac{1-F_Y(\x)}{F_{H_\S}(\x)}\wedge\frac{F_Y(\x)}{1-F_{H_\S}(\x)}\leq 1.  
\end{align*}
\end{proof}

\section{Algorithms}\label{sec:5}
We revisit the cvMDL algorithm in \Cref{fig:flowchart}: the loss function $L_\S$ in \eqref{loss} is the desired loss function to optimize over but requires oracle statistics (i.e. $k_1(\S)$ and $k_2(\S)$). Thus, we replace it with an approximation $\widehat{L}_\S$ that we describe in this section. Additionally, the computations in the ``Analyze Models'' step are now more transparent: The oracle computations are given by \eqref{eq:Sopt-oracle} and \eqref{optm}. In a practical algorithmic setting, we replace these with consistent approximate computations, which is the topic of this section.

When using approximate quantities to compute $L_\S$, the explicit exploration-exploitation loss decomposition in \eqref{loss} may no longer be true. Nevertheless, if the quantities we estimate are sufficiently accurate, then such a decomposition is expected to be approximately valid. Thus, in devising practical algorithms, we use the oracle loss form \eqref{loss} (with estimated coefficients) instead of \eqref{eq:loss-def} as the criteria for model selection. 
We present in the numerical section some empirical evidence that such a replacement has little impact on model selection. 

Since the proposed estimators change when new exploration samples are collected, the dependence on this number of exploration samples must be made explicit. For $\S\subseteq \{1:n\}$, we let $\widehat{L}_\S(\cdot;t)$ denote the estimated loss function $L_\S$ after having collected $t$ exploration samples. We then let $\widehat{m}^*_\S$ be the corresponding estimator for the optimal exploration sample size $m^\ast_\S$. 
Summarizing this: the intuition behind the cvMDL algorithm is that we use currently collected exploration data ($t$ samples) to find the estimated optimal model ($\widehat{\S}^\ast$) and the corresponding exploration rate ($\widehat{m}^\ast_{\widehat{\S}^\ast}$). Based on the value of $\widehat{m}^\ast_{\widehat{\S}^\ast}$ relative to $t$, we decide whether to continue to explore or to switch to exploitation.

\subsection{Estimators for oracle quantities}\label{sorted-comp}

In this section, we discuss how to estimate $L_{\S}$, $m^*_\S$, and $\alpha(\x)$ from exploration data when instantiating cvMDL using the linear approximators as introduced in Section \ref{sec:4}.
The control variates $h(X_\S; \x) = \mathbf 1_{\{H_\S(X_\S)\leq \x\}} $ belong to a parametric family characterized by $\bm\beta^{(i)}_{\S^+}, i\in \{1:d\}$ from \eqref{eq:betaplus}, which can be estimated from exploration data. 

Recall from \eqref{onta1} that the $\ell$th exploration sample of all low-fidelity models in $\S$ is denoted by $X_{\ex, \ell, \S}$.
Similarly, we define $X_{\ex, \ell, \S^+} := (1, X_{\ex, \ell, \S}^\top)^\top$.  
To estimate $\bm\beta_{\S^+}^{(i)}$, we use the least-squares estimator:
\begin{align}\label{eq:betahatS}
&\widehat{\bm\beta}^{(i)}_{\S^+} = \ZZ_\S^\dagger\bm Y^{(i)}& \ZZ_\S = \begin{bmatrix}
X^\top_{\ex, 1, \S^+} \\
\vdots\\
X^\top_{\ex, m, \S^+} \\
\end{bmatrix}\in\R^{m\times (d_\S + 1)}
\ \ \ \bm Y^{(i)} = \begin{bmatrix}
Y^{(i)}_{\ex, 1}\\
\vdots\\
Y^{(i)}_{\ex, m}
\end{bmatrix}\in\R^m,
\end{align}
where $(X^\top_{\ex, \ell, 1}, \ldots, Y^\top_{\ex,\ell})^\top_{\ell\in \{1:m\}}$ are joint exploration samples, and the design matrix $\ZZ_\S$ is assumed to have full column rank\footnote{This motivates the minimal exploration size condition in Algorithm \ref{fig:flowchart}, which is a neccesary condition for full rank here.}.

For $\x\in\R^d$, $h(X_\S; \x)$ can be estimated as 
\begin{align}\label{eq:hShat-HShat}
  &\widehat{h}(X_\S; \x) = \mathbf 1_{\{\widehat{H}_\S(X_\S)\leq \x\}}&\widehat{H}_\S(X_\S) \coloneqq \widehat{\bm B}_{\S^+}^\top X_{\S^+} \eqqcolon 
\begin{bmatrix}
X^\top_{\S^+}\widehat{\bm\beta}^{(1)}_{\S^+}\\
\vdots\\
X^\top_{\S^+}\widehat{\bm\beta}^{(d)}_{\S^+}
\end{bmatrix},
\end{align}
For ease of notation, we write $\widehat{H}_\S(X_\S)$ and $\widehat{h}_\S(X_\S; \x)$ as $\widehat{H}_\S$ and $\widehat{h}_\S$ when $X_{\S^+}$ is generic and not necessarily related to the exploration and exploitation data. 
We introduce some additional notation for quantities involving both estimated coefficients and empirical CDFs using exploration data:
\begin{align*}
\widehat{F}_Y(\x) &= \frac{1}{m}\sum_{\ell\in \{1:m\}}\mathbf 1_{\{Y_{\ex, \ell}\leq \x\}}\\
\widehat{F}_{\widehat{H}_\S}(\x) &= \frac{1}{m}\sum_{\ell\in \{1:m\}}\widehat{h}(X_{\ex, \ell, \S}; \x) = \frac{1}{m}\sum_{\ell\in \{1:m\}}\mathbf 1_{\{\widehat{H}_\S(X_{\ex, \ell, \S})\leq \x\}}\\
\widehat{F}_{Y\vee\widehat{H}_\S}(\x) & = \frac{1}{m}\sum_{\ell\in \{1:m\}}\mathbf 1_{\{{Y_{\ex, \ell} \vee \widehat{H}_\S(X_{\ex, \ell, \S})}\leq \x\}}.
\end{align*}

To compute the loss function approximation, we build approximations to $k_1$ and $k_2$ in \eqref{eq:k12-def}, which requires us to compute $\rho_\S^2$ in \eqref{eq:rhoS-def}. For this purpose, observe that
\begin{align*}
(1-\rho_\S^2(\x))F_Y(\x)(1-F_Y(\x)) = \E[(\mathbf 1_{\{Y\leq\x\}}-F_Y(\x))-\alpha(\x)(\mathbf 1_{\{H_\S\leq \x\}}-F_{H_\S}(\x))]^2,
\end{align*}
where $\alpha(\x)$ is defined in \eqref{here}. 
The quantity in the expectation is the mean squared regression residual between two centered Bernoulli random variables $\mathbf 1_{\{Y\leq\x\}}$ and $\mathbf 1_{\{H_\S\leq \x\}}$. 
Thus, a natural estimator for $(1-\rho_\S^2(\x))F_Y(\x)(1-F_Y(\x))$ is to compute an empirical mean-squared difference between $\mathbf 1_{\{Y \leq x\}}$ and a regressor with covariates $\mathbf 1_{\{\widehat{H}_\S \leq x\}}$, which requires data for $Y$. Since we have (uncentered) data for $Y$ on the exploration samples $Y_{\ex,j}$ for $j \in \{1:m\}$, we can evaluate a regressor for $Y$ with covariates $\widehat{H}_\S$ together with the intercept term on the exploration data sites. This results in the following estimator $\mathcal{K}_1$ for $(1-\rho_\S^2(\x))F_Y(\x)(1-F_Y(\x))$
\begin{align}
&\mathcal K_1(\x) = \frac{1}{m}\sum_{\ell\in \{1:m\}}(\mathbf 1_{\{Y_{\ex, \ell}\leq \x\}} - r_\ell(\x))^2& 
\begin{bmatrix}
r_1(\x)\\
\vdots\\
r_m(\x)
\end{bmatrix} = \W_\S\W_\S^\dagger \begin{bmatrix}
\bm 1_{\{Y_{\ex, 1}\leq \x\}}\\
\vdots\\
\bm 1_{\{Y_{\ex, m}\leq \x\}}
\end{bmatrix},
\label{my1}
\end{align}
where 
\begin{align*}
&\W_\S =  \begin{bmatrix}
1 &\widehat{h}_\S(X_{\ex,1,{\S}}; \x)\\
\vdots & \vdots\\
1 & \widehat{h}_\S(X_{\ex,m,{\S}}; \x)
\end{bmatrix}
.
\end{align*}
This allows us to estimate $\rho_\S^2(\x)F_Y(\x)(1-F_Y(\x))$ as
\begin{align}\label{eq:calK2}
\mathcal K_2(\x) &= \widehat{F}_Y(\x)(1- \widehat{F}_Y(\x)) - \mathcal K_1(\x).
\end{align}
Consequently, we can estimate $k_1(\S)$ and $k_2(\S)$ as 
\begin{align}
&\widehat{k}_1(\S) = \int_{\R^d}\mathcal \omega(\x)\mathcal K_1(\x) \text{d}\x&\widehat{k}_2(\S) = c_\S\int_{\R^d}\omega(\x)\mathcal K_2(\x) \text{d}\x.\label{lovepang}
\end{align}
The above estimators for $k_1(\S)$ and $k_2(\S)$ are positive and coincide with empirical estimators for these quantities whenever defined (see \Cref{appx:3}), which is a crucial realization for our consistency results later. 
Plugging the above estimates into \eqref{loss} and \eqref{optm} yields estimates for $L_{\S}$ and $m^*_\S$:
\begin{align}
&\widehat{L}_\S(z; m) = \frac{\widehat{k}_1(\S)}{z} + \frac{\widehat{k}_2(\S)}{B-c_{\text{epr}}z}&\widehat{m}^*_\S = \frac{B}{c_{\text{epr}}+\sqrt{\frac{c_{\text{epr}} \widehat{k}_2(\S)}{\widehat{k}_1(\S)}}}.\label{lm}
\end{align} 
Note that $\widehat{L}_\S(z; m)$ has a parameter $m$ indicating the number of exploration samples used to compute $\widehat{k}_1(\S)$ and $\widehat{k}_2(\S)$, and a variable $z$ denoting the exploration rate where to evaluate $\widehat{L}_\S$. We define $\widehat{\S}^\ast$ as the optimal model selected by this estimator,
\begin{align}\label{eq:Shatast}
  \widehat{\S}^\ast = \argmin_{\S \subseteq \{1:n\}} \widehat{L}_\S(\widehat{m}^\ast_\S; m),
\end{align}
which parallels the oracle choice \eqref{eq:Sopt-oracle}. We have described all quantities needed to complete the exploration phase of \Cref{fig:flowchart}. What remains is to describe how the CDF estimator $\widetilde{F}_{\widehat{\S}^\ast}$ in the exploitation phase of \Cref{fig:flowchart} is generated.

Our exploitation goal is to generate an estimator for \eqref{panger}, and so we need to estimate $\alpha(\x)$:
\begin{align}\label{myalpha}
&\widehat{\alpha}(\x) =  \frac{\widehat{F}_{Y\vee\widehat{H}_\S}(\x) - \widehat{F}_{Y}(\x)\widehat{F}_{\widehat{H}_\S}(\x)}{\widehat{F}_{\widehat{H}_\S}(\x)(1-\widehat{F}_{\widehat{H}_\S}(\x))}&\x\in\text{supp}(\widehat{F}_{\widehat{H}_\S}(\x))^\circ,
\end{align}
and $\widehat{\alpha}(\x) = 0$ zero otherwise.
By a similar reasoning as in Lemma \ref{kang}, one has 
\begin{align}
|\widehat{\alpha}(\x)|\leq 1. \label{bd2}
\end{align}

Finally, the \textit{exploitation} estimator $\widetilde{F}_\S(\x)$ for $F_Y(\x)$ based on estimated parameters, utilizes $N_\S$ exploitation samples (i.e., exhausts the remaining budget $B$) and is given by,
\begin{align}
  \widetilde{F}_{\S}(\x) \coloneqq
  \widehat{F}_Y(\x) -  \frac{1}{m}\sum_{\ell\in \{1:m\}}\left(\widehat{\alpha}(\x)\widehat{h}_\S(X_{\ex, \ell, \S} ;\x)- \frac{1}{N_\S}\sum_{j\in \{1:N_\S\}}\widehat{\alpha}(\x)\widehat{h}_\S(X_{\ext, \ell, \S}; \x)\right),\label{fug}
\end{align}
where $\S = \widehat{\S}^\ast$ is the selected model based on $\widehat{k}_1(\S)$ and $\widehat{k}_2(\S)$. 
By inspection, we observe that $\widetilde{F}_\S(\x)$ is a piecewise affine correction of $\widehat{F}_Y$, where the correction is based on the control variates $\widehat{h}_\S$.
\begin{Rem}
  The estimator $\widehat{\alpha}(\x)$ is undefined and manually set to zero for $\x$ outside the support of $\widehat{F}_{\widehat{H}_\S}$, as in that case the denominator vanishes. 
   Alternatively, one can define $\widehat{\alpha}(\x)$ for $\x$ outside the support of $\widehat{F}_{\widehat{H}_\S}$ as $\widehat{\alpha}(\x')$ for some $\x'$ inside the support of $\widehat{F}_{\widehat{H}_\S}$ that can be accurately estimated yet remains close to $\x$. To illustrate, consider $d=1$.
Assuming $\alpha(x)$ is a continuous function of $x$ in $\text{supp}(F_{H_\S})$ and $\text{supp}(\widehat{F}_{\widehat{H}_\S}(x)) = [x_{\min}, x_{\max}]$ for some $x_{\min}<x_{\max}$, for $x\in\text{supp}(\widehat{F}_{\widehat{H}_\S}(x))^c$, we may estimate $\alpha(x)$ outside $[x_{\min}, x_{\max}]$ as
\begin{align}\label{myalpha2}
\widehat{\alpha}(x) = \begin{cases}
\frac{\widehat{F}_{Y\vee\widehat{H}_\S}(x(\tau)) - \widehat{F}_{Y}(x(\tau))\tau}{\tau (1-\tau)} & x\leq x_{\min} \\
\frac{\widehat{F}_{Y\vee\widehat{H}_\S}(x(1-\tau)) - \widehat{F}_{Y}(x(1-\tau))(1-\tau)}{\tau (1-\tau)} & x\geq x_{\max},
      \end{cases}
\end{align}
where $x(\tau)$ and $x(1-\tau)$ are the $\tau$ and $1-\tau$ quantiles of $\widehat{F}_{\widehat{H}_\S}$ for some small $\tau\in (0, 1)$:
\begin{align*}
&x(\tau) = \widehat{F}^{-1}_{\widehat{H}_\S}(\tau)&x(1-\tau) = \widehat{F}^{-1}_{\widehat{H}_\S}(1-\tau).
\end{align*}
This allows us to get nontrivial estimates of $F_Y$ outside $[x_{\min}, x_{\max}]$, i.e. in the tail regime. 
When $d\geq 2$, one may generalize the ideas above by projecting the points in the tail regime to some bounded set in $\R^d$ that contains most of the measure in the domain.  
\end{Rem}

\subsection{Monotonicity of the exploitation CDF estimator}\label{ssec:monotonicity}
By construction, $\widetilde{F}_\S(\x)$ is a piecewise constant function on some $d$-dimensional rectangular partition of $\R^d$, but is not necessarily a monotone nondecreasing function in each direction due to the fluctuations of estimators used in the construction.
To address this issue, we introduce a dimension-wise recursive-sorting post-processing procedure on values in the range of $\widetilde{F}_\S$ to recover the desired monotonicity and ensure that we compute an actual distribution function. 
To represent $\widetilde{F}_\S(\x)$ as a $d$-dimensional tensor, we introduce the index set $\bm I = \otimes_{i\in \{1:d\}}(z_{i, 1}, \ldots, z_{i, M_i})$, $-\infty = z_{i, 1}\leq \cdots\leq z_{i, M_i} = +\infty$, where $z_{i, j}$ is the $j$th order statistic of the projected partition points associated with $\widetilde{F}_\S(\x)$, and $M_i$ is the total number of such projected points.
Using this notation, we express $\widetilde{F}_\S(\x)$ as a tensor $\bm T$, where 
\begin{align*}
&\widetilde{F}_\S(\x) = \bm T_{z_{1, s_1}, \cdots, z_{d, s_d}}& \x\in \prod_{i\in \{1:d\}}[z_{i, s_i}, z_{i, s_{i}+1}).
\end{align*}
The desired monotonicity in each dimension can be recovered by an alternating dimension-wise sorting of entries in $\bm T$ until convergence.
The details are given in \Cref{alg-sort}.
\begin{algorithm}
\hspace*{\algorithmicindent} \textbf{Input}: a tensor $\bm T$ that represents the estimated CDF $\widetilde{F}_\S(\x)$\\
    \hspace*{\algorithmicindent} \textbf{Output}: a sorted tensor sorted($\bm T$) with desired monotonicity
 \begin{algorithmic}[1]
   \STATE Initialize sorted($\bm T$) as a all-zeros tensor with the same size as $\bm T$
 \WHILE{sorted($\bm T$)$\neq \bm T$}{
 \STATE $\text{sorted}(\bm T)\gets \bm T$
   \FOR{$i\in \{1: d\}$}{
   \FOR{$\bm z :=(z_{1, s_1}, \ldots, z_{i-1, s_{i-1}}, z_{i+1, s_{i+1}}, \ldots, z_{d, s_{d}})\in\otimes_{j\in \{1:d\}\setminus\{i\}}(z_{j, 1}, \ldots, z_{j, M_i})$}{
   \STATE $\bm T[\bm z, :]\gets\text{sort}(\bm T[\bm z, :])$, where $\bm T[\bm z, :]:=(\bm T_{z_{1, s_1}, \ldots, z_{i, j}, \ldots, z_{d, s_{d}}})_{j\in \{1: M_i\}}$
   }
   \ENDFOR
}
   \ENDFOR}
      \ENDWHILE
     \STATE{Return $\text{sorted}(\bm T)$}
 \end{algorithmic}
\caption{Alternating sorting.} 
\label{alg-sort}
\end{algorithm}

An example when $d=2$ is given below:
\begin{align*}
\begin{bmatrix}
0.7 & 0.4 & 0\\
0.3 & 0.5 & 0.2\\
1 & 0.8 & 0.6
\end{bmatrix}
\xrightarrow{\text{sort rows}}
\begin{bmatrix}
0 & 0.4 & 0.7\\
0.2 & 0.3 & 0.5\\
0.6 & 0.8 & 1
\end{bmatrix}
\xrightarrow{\text{sort columns}}
\begin{bmatrix}
0 & 0.3 & 0.5\\
0.2 & 0.4 & 0.7\\
0.6 & 0.8 & 1
\end{bmatrix}
\\
\begin{bmatrix}
0.7 & 0.4 & 0\\
0.3 & 0.5 & 0.2\\
1 & 0.8 & 0.6
\end{bmatrix}
\xrightarrow{\text{sort columns}}
\begin{bmatrix}
0.3 & 0.4 & 0\\
0.7 & 0.5 & 0.2\\
1 & 0.8 & 0.6
\end{bmatrix}
\xrightarrow{\text{sort rows}}
\begin{bmatrix}
0 & 0.3 & 0.4\\
0.2 & 0.5 & 0.7\\
0.6 & 0.8 & 1
\end{bmatrix}
\end{align*}
As shown above, sorting ends up in some stationary point with desired monotonicity after a finite number of steps (see Theorem \ref{thm:sort}), but different orders of sorting may lead to different sorted CDF representations when $d\geq 2$. 
However, in our case, $\widetilde{F}_\S(\x)$ is itself a perturbation of the CDF of $Y$, so the sorting procedure is often beneficial for stabilizing the algorithm.
A more detailed empirical study on this is given in Section \ref{sec:6}. The sorting procedure described converges (i.e., achieves monotonicity in the values of $\bm{T}$) in a finite number of iterations.
\begin{Th}\label{thm:sort}
Assume that all the entries in $\bm T$ are distinct. 
The alternating sorting algorithm described in \Cref{alg-sort} converges to a stationary point with desired monotonicity within a finite number of iterations. 
\end{Th}
\begin{proof}
See \Cref{ssec:sorting-proof}.
\end{proof}

\subsection{Exploration sampling}
We next describe the precise action taken when we decide to continue exploring. 
In particular, we need to define the function $Q(m,\widehat{m}_{\widehat{\S}^\ast})$ in \Cref{fig:flowchart}. When the current number $m$ of exploration samples is smaller than the estimated optimal number of samples $\widehat{m}^\ast_{\widehat{\S}^\ast}$, the function $Q$ determines how to increase $m$.
A natural choice for $Q$ is $Q(m,\widehat{m}^\ast_{\widehat{\S}^\ast}) = m+1$, i.e., simply increase by a single additional exploration sample. In practice, we observe that this behavior can be overly conservative and time-consuming when $B$ is large. As an alternative, one could use a more aggressive strategy, say $Q(m,\widehat{m}^\ast_{\widehat{\S}^\ast}) = \frac{1}{2}\left(m + \widehat{m}^\ast_{\widehat{\S}^\ast}\right)$, which more quickly closes the gap between $m$ and $\widehat{m}_{\widehat{\S}^\ast}$. However, there are situations when this is too aggressive. For example, if $m$ is small (such as at initialization) then estimated quantities can be poor approximations, and in some cases $\widehat{m}_{\widehat{\S}^\ast}$ is significantly overestimated, and thus increasing $m$ to $\frac{1}{2}\left(m + \widehat{m}^\ast_{\widehat{\S}^\ast}\right)$ can actually result in substantially overshooting the oracle value of $m^\ast_{\S^\ast}$. The probability of such an event is often positive and does not vanish as $B$ increases. 

As a compromise between these conservative and aggressive behaviors, we choose the following form:
\begin{align}\label{eq:Q-def}
  Q(m, \widehat{m}^\ast_{\widehat{\S}^\ast}) = \left\{
    \begin{array}{rl} 
    2 m, & \ \ \ \ 1 \leq m < \frac{\widehat{m}^\ast_{\widehat{\S}^\ast}}{2} \\
      \frac{1}{2} \left(m + \widehat{m}^\ast_{\widehat{\S}^\ast} \right), & \frac{\widehat{m}^\ast_{\widehat{\S}^\ast}}{2} \leq m < \widehat{m}^\ast_{\widehat{\S}^\ast}
  \end{array}\right.
  .
\end{align}
Since $\widehat{m}^\ast_{\widehat{\S}^\ast}$ is proportional to $B$, the above choice ensures that there is a sufficient amount of time for the algorithm to take exponential exploration whose growth manner is independent of the value of $\widehat{m}^\ast_{\widehat{\S}^\ast}$, which ensures both efficiency and accuracy of the algorithm.  
We note that neither the exponential rate two nor taking the average between $m$ and $\widehat{m}^\ast_{\widehat{\S}^\ast}$ in \eqref{eq:Q-def} is special, and can respectively be replaced with other rates greater than one and nonuniform averaging operations subject to appropriate modifications. Both the theoretical conclusions and numerical simulations in the subsequent sections assume that $Q$ has the form above, but other reasonable choices for $Q$ do not change the theoretical conclusions.

We have completed all technical descriptions of \Cref{fig:flowchart}. 
A more fleshed-out pseudocode version is given in \Cref{alg2-detailed} that provides more details for every step. 
Next, we establish that the proposed algorithm enjoys optimality guarantees relative to model selection and budget allocation strategies produced by an oracle.

\begin{algorithm}
\hspace*{\algorithmicindent} \textbf{Input}: $B$: total budget, model costs $c_0, c_1, \ldots, c_n$\\
    \hspace*{\algorithmicindent} \textbf{Output}: an estimator for $F_Y(\x)$
 \begin{algorithmic}[1]
   \STATE Initialize $\textsf{exploration} = \text{TRUE}$
   \STATE Initialize $m = \sum_{i\in \{1:n\}}d_i + 2$
   \STATE Generate $m$ exploration samples of $(Y, X_{\{1:n\}})$
 \WHILE{\textsf{exploration} = TRUE}{
   \FOR{$\S\subseteq \{1:n\}$}{
     \STATE Compute regression coefficients $\widehat{\bm{\beta}}_{\S^+}^{(i)}$, $i \in\{1:d\}$ from \eqref{eq:betahatS}
     \STATE Construct $\widehat{H}_\S(X_\S)$ and $\widehat{h}(X_\S; \bm{x})$ from \eqref{eq:hShat-HShat}
     \STATE Compute regression coefficients $r_j(\bm{x})$, $j \in \{1:m\}$ from \eqref{my1}
     \STATE Construct $\mathcal{K}_1$ and $\mathcal{K}_2$ from \eqref{my1} and \eqref{eq:calK2}, respectively
     \STATE Evaluate $\widehat{k}_1(\S)$ and $\widehat{k}_2(\S)$ using \eqref{lovepang} and a quadrature rule on $\R^d$
   \STATE Compute $\widehat{m}^*_\S$ and $\widehat{L}_\S(\ \cdot\ ; m)$ from \eqref{lm}
   \STATE Compute the minimal expected loss $\widehat{L}_\S(m\vee \widehat{m}^*_\S; m)$ from \eqref{lm}
}
   \ENDFOR
   \STATE{Choose $\widehat{\S}^* = \argmin_{\S\subseteq \{1:n\}}\widehat{L}_\S(m\vee \widehat{m}^*_\S; m)$};
   \IF{$m < \widehat{m}^\ast_{\widehat{\S}^\ast}$}
     \STATE{Generate $Q(m, \widehat{m}^\ast_{\widehat{\S}^\ast})-m$ additional samples of $(Y, X_{\{1:n\}})$, where $Q$ is given in \eqref{eq:Q-def}}
     \STATE{Increase $m$: $m \gets Q(m,\widehat{m}^\ast_{\widehat{\S}^\ast})$}
   \ELSE
      \STATE{\textsf{exploration} = FALSE}
    \ENDIF
     }
     \ENDWHILE
     \STATE{Generate $N_{\widehat{\S}^\ast}$ samples of $X_{\widehat{\S}^\ast}$, with $N_\S$ given in \eqref{eq:NS-def}}
     \STATE Construct $\widehat{\alpha}(\bm{x})$ for $\S \gets \widehat{\S}^\ast$ using \eqref{myalpha}
     \STATE{Generate $\widehat{\S}^\ast$ exploitation estimator $\widetilde{F}_{\widehat{\S}^\ast}$ using \eqref{fug}.}
 \end{algorithmic}
\caption{The detailed cvMDL algorithm.} 
\label{alg2-detailed}
\end{algorithm}

\subsection{Model consistency and optimality}\label{ssec:theory}

We now provide theoretical guarantees for \Cref{alg2-detailed} (corresponding to the flowchart in \Cref{fig:flowchart}). In summary, we show that as the budget $B$ tends to infinity, the model subset $\widehat{\S}^\ast$ chosen along with the number of exploration samples $m$ taken in \Cref{alg2-detailed}, both converge to the oracle optimal model $\S^\ast$ and the optimal number of exploration samples $m^\ast_{\S^\ast}$, respectively. 

We need some technical assumptions in order to proceed with our results. Since we estimate quadratic moments, we require quadratic moments to exist.
We also require that there are no pairs of low-fidelity QoIs that are perfectly correlated. These are codified in the following two assumptions.
\begin{Ass}\label{ass1}
  The models $X_{\{1:n\}}$ and $Y$ have bounded second moments:
\begin{align}
  \E[\| X_{\{1:n\}} \|_2^2 + \E \|Y\|_2^2]  < \infty.\label{ass}
\end{align}
\end{Ass}

\begin{Ass}\label{ass2}
  The uncentered second moment matrix $\E [X_{\{1:n\}+} X^\top_{\{1:n\}+}]$ is invertible, where $X_{\{1:n\}+} = (1, X_{\{1:n\}}^\top)^\top$.
\end{Ass}

\Cref{ass1} is the minimal moment condition on QoIs that we require to make oracle quantities well-defined. Random variables that violate \Cref{ass2} exhibit perfect multicollinearity and in practice are relatively rare. \Cref{ass2} being violated does not cause any conceptual breakdown of our procedure; the only consequence is that all the linear regression procedures suffer from a lack of identifiability of optimal covariates. 
While there are numerous standard procedures to remedy multicollinearity, such as covariate removal or regularization, violation of this assumption did not surface in our experiments, so we do not utilize any of these remedies.

The model selection procedure requires estimating the average $\omega$-weighted $L^2$ norm. 
This requires us to make certain assumptions about $\omega$.
\begin{Ass}\label{ass:omega}
  The weight $\omega(\x)$ is chosen so that \textit{either} of the following conditions is true:
  \begin{enumerate}[label=(\alph*)]
  \item $\|\omega\|_{L^\infty(\R^d)}<\infty$ (e.g. $\omega(\x)\equiv 1$) and $d=1$; or
  \item $\|\omega\|_{L^1(\R^d)}<\infty$.
  \end{enumerate}
\end{Ass}

The final more technical assumption we require involves some regularity on distribution functions. In particular, we show pointwise convergence in $\x$ of the estimator $\widehat{\alpha}(\x)$ to the oracle parameter $\alpha(\x)$, and to accomplish this we require bounds on the local variations of $F_{H_\S}$ and $F_{H_\S \vee Y}$ constructed in the model selection procedure. More technically, a sufficient assumption is a bounded local variations condition involving CDFs of certain $d$-dimensional sketches of $X_{\{1:n\}}$ and $Y$.

\begin{Ass}\label{ass4}
  Define
  \begin{align*}
    V(\bm A) &:= 
  (X_{\S^+}^\top\bm A)^\top\in\R^d
  &\bm A := [\bm\A^{(1)}, \cdots, \bm\A^{(d)}]\in\R^{d_\S\times d}, 
  \end{align*}
  and recall the optimal coefficient matrix $\bm B_{\S^+}$ in \eqref{eq:betaplus}.
We assume the CDFs of $V(\bm A)$ and $V(\bm A)\vee Y$, denoted by $F_{V(\bm A)}$ and $F_{V(\bm A)\vee Y}$, are globally Lipschitz near $\bm B_{\S^+}$ for all $\S$.
In particular, we assume that there exists $\e>0$ such that 
\begin{align*}
\max_{\S\subseteq \{1:n\}}\sup_{\bm A: \|\bm A-\bm B_{\S^+}\|_F\leq \e}\left\{\|F_{V(\bm A)}\|_{\text{Lip}} + \|F_{V(\bm A)\vee Y}\|_{\text{Lip}}\right\} = C<\infty,
\end{align*}
where $\|\cdot\|_{\text{Lip}}$ is the Lipschitz constant defined as 
\begin{align*}
&\|f\|_{\text{Lip}} = \sup_{\x\neq \x'}\frac{|f(\x)-f(\x')|}{\|\x-\x'\|_2}& f: \R^d\to\R.
\end{align*}
\end{Ass}
This final assumption is less transparent than our previous ones. 
An unnecessarily stronger sufficient condition to ensure that \Cref{ass4} holds is to assume that both $Y$ and all unit linear combinations of $X_{\{1:n\}}$ have uniformly bounded densities, and that every high-fidelity covariate $Y^{(i)}$ is correlated with every low-fidelity covariate $X^{(r)}_j$, i.e., $\min_{i,j,r} |\mathrm{Corr}(Y^{(i)}, X_j^{(r)})| > 0$. Alternatively, one could assume that the same bounded density condition, and the rather reasonable condition that the oracle regression coefficients ${\bm B}_{\S^+}$ select at \textit{least} one non-deterministic covariate for every $\S$.

We can now present our main results regarding applying the cvMDL algorithm with $h(X_\S; \x)$ constructed using linear approximations, with the corresponding loss function parameters estimated from \eqref{lovepang} and \eqref{lm}. In particular, we have that the adaptive exploration rate $m(B)$ asymptotically matches the optimal (oracle) exploration rate $m^*_{\S^*}$ defined in Section \ref{sec:22}, and the selected model $\S(B)$ converges to the optimal (oracle) model $\S^*$ as $B\to\infty$:

\begin{Th}[Uniform consistency and asymptotic optimality of cvMDL in \Cref{alg2-detailed}]\label{main}
  Let $h(X_\S; \x)$ be defined in \eqref{linear:h}, i.e., we use the linear approximation estimators from Section \ref{sec:4}, and assume the model parameters are estimated via \eqref{lovepang} and \eqref{lm}. Then consider \Cref{alg2-detailed} with an input budget $B$, and let
  \begin{itemize}[itemsep=0pt]
    \item $m(B)= \widehat{m}_{\widehat{\S}^*}$ be the number of exploration samples chosen by \Cref{alg2-detailed},
    \item $\S(B)= \widehat{\S}^*$ be the model selected for exploitation,
    \item $\widetilde{F}(\x;B) = \widetilde{F}_{\widehat{\S}^\ast}(\x)$ be the CDF estimator for $F_Y$.
  \end{itemize}
  Under Assumptions \ref{ass1}, \ref{ass2}, \ref{ass:omega}, and \ref{ass4}, then with probability one,
\begin{subequations}
\begin{align}\label{rr1}
  \lim_{B\to\infty}\frac{m(B)}{m^*_{\S^*}} &= 1, \\\label{rr2}
  \lim_{B\to\infty}\S(B) &= \S^*,\\\label{rr3}
  \lim_{B\to\infty} \sup_{\x\in\R^d}|\widetilde{F}(\x; B)-F_Y(\x)|&= 0,
\end{align}
\end{subequations}
  where $\S^*$ and $m^*_{\S^*}$ are the unique optimal (oracle) model choice and exploration sample size defined in Section \ref{sec:22}.
\end{Th}

The proof is given in \Cref{aa1}. The result \eqref{rr3} should not come as a surprise since uniform consistency is generally true for empirical CDF estimators. Therefore, while \eqref{rr1} and \eqref{rr2} show that cvMDL in \Cref{alg2-detailed} exhibits optimality (relative to an oracle) for the choice of exploration samples and sample allocation across models, \eqref{rr3} is not evidence that the multifidelity estimator $\widetilde{F}(\x; B)$ is superior to the empirical CDF estimator that uses only the high-fidelity samples, although it confirms that $\widetilde{F}(\x; B)$ behaves as expected. The major difference that distinguishes $\widetilde{F}(\x; B)$ from a standard empirical CDF estimator is a constant term resulting from the mean $\omega$-weighted $L^2$ convergence rate; see the discussion near the end of \Cref{sec:22}.

The statements in Theorem \ref{main} and \cite[Theorem 5.2]{xu2021budget} are similar, but in the former, the requisite assumptions are \textit{much} weaker and the guarantees are stronger. In fact, for \cite[Theorem 5.2]{xu2021budget} to hold, one must assume that $\E[Y|X_\S]$ is a linear function of $X_\S$ and $(Y-\E[Y|X_\S])\indep X_\S$ for all $\S\subseteq \{1:n\}$.
However, none of these assumptions is needed in Theorem \ref{main}. Additionally, \Cref{main} ensures convergence for a \textit{multivariate} distribution function instead of the univariate convergence statements in \cite[Theorem 5.2]{xu2021budget}.

\subsection{A brief view into proving \Cref{main}}

While we leave the technical parts of proving \Cref{main} to \Cref{aa1}, we can summarize the crucial intermediate results that allow the proof to succeed. The major results we need revolve around the consistency of various estimators as $m$ and/or $N_\S$ approach infinity. The following two sets of results leverage the assumptions to conclude the consistency of intermediate computations in the algorithm.

The first collection of results shows that the finite-sample estimators for quantities computed in the exploration phase are consistent as the number of exploration samples $m$ tends to infinity.
\begin{Lemma}[Asymptotic consistency of exploration estimators]\label{lemma:parameter-consistency}
  We have the following technical estimates and consistency results for all $\S\subseteq \{1:n\}$:
  \begin{enumerate}[label={(\roman*)},itemsep=0pt]
    \item\label{lemma:item:0} 
      Under Assumptions \ref{ass1} and \ref{ass4}, then with probability one,
      \begin{subequations}\label{0020}
      \begin{align}
      \sup_{\|\bm A-\bm B_{\S^+}\|_F<\e}\sup_{\x\in\R^d}|F_{V(\bm A)}(\x) - F_{V(\bm B_{\S^+})}(\x)|\lesssim\|\bm A - \bm B_{\S^+}\|_F^{2/3}\lesssim\e^{2/3}\label{0021}\\
      \sup_{\|\bm A-\bm B_{\S^+}\|_F<\e}\sup_{\x\in\R^d}|F_{V(\bm A)\vee Y}(\x) - F_{V(\bm B_{\S^+})\vee Y}(\x)|\lesssim\|\bm A - \bm B_{\S^+}\|_F^{2/3}\lesssim\e^{2/3}\label{0022}
      .
      \end{align}
      \end{subequations}
    \item \label{lemma:item:1} 
      Under Assumptions \ref{ass1} and \ref{ass2}, then with probability one, 
            \begin{align*}
        \lim_{m \to \infty} \widehat{\bm B}_{\S^+}&= \bm B_{\S^+}. 
      \end{align*}
    \item \label{lemma:item:2} Under Assumptions \ref{ass1}, \ref{ass2}, and \ref{ass4},  then with probability one,
      \begin{align*}
        \lim_{m\to\infty} \sup_{\x\in\R^d}|\widehat{F}_{\widehat{H}_\S}(\x)- F_{H_\S}(\x)| &= 0 & \lim_{m \to \infty} \sup_{\x\in\R^d}|\widehat{F}_{Y\vee\widehat{H}_\S}(\x)- {F}_{Y\vee{H}_\S}(\x)| &=  0.
      \end{align*}
    \item \label{lemma:item:3}
      Under Assumptions \ref{ass1}, \ref{ass2}, and \ref{ass4}, then almost surely as $m \to \infty$ we have that,
      \begin{align*}
      &\mathcal K_1(\x)\to (1-\rho_\S^2(\x))F_Y(\x)(1-F_Y(\x))\ \ \ \mathcal K_2(\x)\to\rho_\S^2(\x)F_Y(\x)(1-F_Y(\x))
      \end{align*}
      for all $\x\in\R^d$. 
    \item \label{lemma:item:4}
      Under Assumptions \ref{ass1}, \ref{ass2}, \ref{ass:omega}, and \ref{ass4}, then with probability one, $\lim_{m\to\infty} \widehat{k}_1(\S) = k_1(\S)$ and $\lim_{m\to\infty} \widehat{k}_2(\S) = k_2(\S)$. 
    \item\label{lemma:item:5}
      Under Assumptions \ref{ass1}, \ref{ass2}, and \ref{ass4}, for $\x\in (\text{supp}(F_{H_\S}))^\circ$, $\widehat{\alpha}(\x)$ is a consistent estimator of $ \alpha(\x) $ almost surely, i.e., $\lim_{m\to\infty} \widehat{\alpha}(\x) = \alpha(\x)$ for every $\x \in (\text{supp}(F_{H_\S}))^\circ$.
  \end{enumerate}
\end{Lemma}
\noindent The proof is given in \Cref{sec:lemma-proof}. 
\begin{Rem}
Note that $\widehat{\alpha}(\x)$ may not be consistent outside $(\text{supp}(F_{H_\S}))^\circ$, where the value of $\alpha(\x)$ is set to be zero in the definition for convenience; see \eqref{here}.  
However, this has no impact on the accuracy of the exploitation estimator as $\mathbf 1_{\{H_\S\leq \x\}}$ is constant in this region. 
\end{Rem}

Our second intermediate result shows that the exploitation estimator for the CDF of $Y$ is consistent asymptotically in both the exploration sample count $m$ and the exploitation sample count $N_{\S}$.

\begin{Lemma}[Uniform asymptotic consistency of the exploitation CDF estimator]\label{unifexploi}
Under Assumptions \ref{ass1}, \ref{ass2}, and \ref{ass4}, then with probability one, $\sup_{\x\in\R^d}|\widetilde{F}_\S(\x)-F_Y(\x)|\to 0$ as $m, N_\S \rightarrow \infty$. 
\end{Lemma}
See \Cref{appx:6} for the proof. The proof of our main result, \Cref{main}, is in \Cref{aa1}, which leverages the results in \Cref{lemma:parameter-consistency} and \Cref{unifexploi}. One additional high-level step needed to prove \Cref{main} is to show that cvMDL in \Cref{alg2-detailed} for asymptotically large budget $B$ results in both $m$ and $N_\S$ going to infinity. This is the first part of the proof presented in \Cref{aa1}.

\section{Numerical simulations}\label{sec:6}

In this section, we compare cvMDL and its variants with other algorithms on three forward uncertainty quantification scenarios. In Section~\ref{sec:num1}, we examine a scalar-valued parametric linear elasticity PDE problem, followed by a vector-valued stochastic differential equation problem concerning the extrema of a geometric Brownian motion over a finite interval in Section~\ref{sec:num2}. Lastly, in Section~\ref{sec:num3} we evaluate the cvMDL method on a scalar-valued practical engineering problem of brittle fracture in a fiber-reinforced matrix. 
We label algorithms under consideration as follows: 
\begin{itemize}[leftmargin=3.2cm,itemsep=-0pt]
  \item[(ECDF)] The empirical CDF estimator for $F_Y$ using the high-fidelity samples only;
   \item[(AETC-d)] The AETC-d algorithm from \cite{xu2021budget};
   \item[(cvMDL)] The cvMDL algorithm in \Cref{alg2-detailed}; 
   \item[(cvMDL-sorted)] cvMDL with the exploitation CDF monotonicity fix using \Cref{alg-sort};
  \item[(cvMDL*)] cvMDL that estimates $\widehat{\alpha}(\x)$ in the tail using \eqref{myalpha2} with $\tau = 0.05$ when $d=1$; 
  \item[(cvMDL*-sorted)] cvMDL*  with the CDF monotonicity fix using \Cref{alg-sort}.
\end{itemize}

For the weight function in the cvMDL algorithm and its variants, we choose $\omega(x)\equiv 1$ for all $x\in\R$ when $Y$ is scalar-valued, but in a case-dependent manner when $Y$ is vector-valued. 
Since the estimators produced by the cvMDL-type and AETC-d algorithms are random (depending on the exploration data), for every budget value $B$, we repeat the experiment $100$ times and report both the average of the mean $\omega$-weighted $L^2$ error and the corresponding $5\%$-$95\%$ quantiles to measure the uncertainty.

\subsection{Linear elasticity}\label{sec:num1}
We consider a suite of models with varying fidelities associated with a parametric elliptic equation, where lower-fidelity models are identified through mesh coarsening. 
The setup is taken from \cite[Section 7.1]{xu2022bandit}.
The elliptic PDE governs displacement in linear elasticity over a square spatial domain $\mathcal D = [0,1]^2$, with $\bm x = (x_1, x_2)^\top$; see \Cref{fig:struct}.
The parametric version of this problem equation seeks the displacement field $\bm{u} = (u, v)^\top$ that satisfies the PDE 
\begin{align*}
  -\nabla \cdot \left(\kappa(\bm{p}, \bm{x})\; \bm{\sigma}(\bm{x})\right) = \bm{F}(\bm x), \hskip 10pt \forall (\bm p,\bm x) \in \mathcal{P} \times \mathcal D 
   \end{align*}
where $\bm p$ is a random input vector that parameterizes the scalar $\kappa$.
Moreover, $\bm{\sigma}$ is the Cauchy stress tensor, given as 
\begin{align*}
       \bm{\sigma}(\bm x) = \begin{bmatrix}
       \sigma_1(\bm x)& \sigma_{12}(\bm x)\\
      \sigma_{12}(\bm x)& \sigma_2(\bm x)
       \end{bmatrix}, \hskip 20pt
         \begin{bmatrix}
          \sigma_1(\bm x) \\ \sigma_2(\bm x) \\ \sigma_{12}(\bm x) 
          \end{bmatrix} = \frac{1}{1 - \nu^2} \begin{bmatrix}
          \frac{\partial u(\bm x)}{\partial x_1} + \frac{\partial v(\bm x)}{\partial x_2} \\ \frac{\partial v(\bm x)}{\partial x_2} + \nu \frac{\partial u(\bm x)}{\partial x_1} \\ \frac{1-\nu}{2} (\frac{\partial u(\bm x)}{\partial x_1} + \frac{\partial v(\bm x)}{\partial x_2}) 
     \end{bmatrix}
   \end{align*}
where we set the Poisson ratio to $\nu = 0.3$. Here, $\kappa(\bm p, \bm x)$ is a scalar modeled with a Karhunen-Lo\`{e}ve expansion with four modes, given by 
\begin{align*}
  \kappa(\bm p, \bm x) = 1 + 0.5 \sum_{i=1}^{4} \sqrt{\lambda_i} \phi_i(\bm x) p_i,
\end{align*}
where $(\lambda_i, \phi_i)$ are ordered eigenpairs of the exponential covariance kernel $K$ on $\mathcal D$, i.e., 
\begin{align*}
  K(\bm x, \bm y) = 
  \exp(-\| \bm x - \bm y\|_1/\eta),
\end{align*}
where $\|\cdot\|_1$ is the $\ell^1$-norm on vectors, and we choose $\eta = 0.7$. 
Hence, $\bm p \in \R^{4}$ is a random input vector with independent components uniformly distributed on $[-1,1]$.

The displacement $\bm u$ is used to compute a scalar QoI, the structural \emph{compliance} or energy norm of the solution, which is the measure of elastic energy absorbed in the structure as a result of loading:
\begin{align*}
  E(\bm u; \bm p) \coloneqq \int_D (\bm{u}(\bm x) \cdot \bm{F}(\bm x)) \ \text{d}\bm x.
 \end{align*}
We solve the above system for each fixed $\bm p$ via the finite element method with standard bilinear square isotropic finite elements on a rectangular mesh \cite{andreassen2011efficient}. 

\begin{figure}[htbp]
\begin{minipage}{0.30\textwidth} 
 \includegraphics[width=\textwidth]{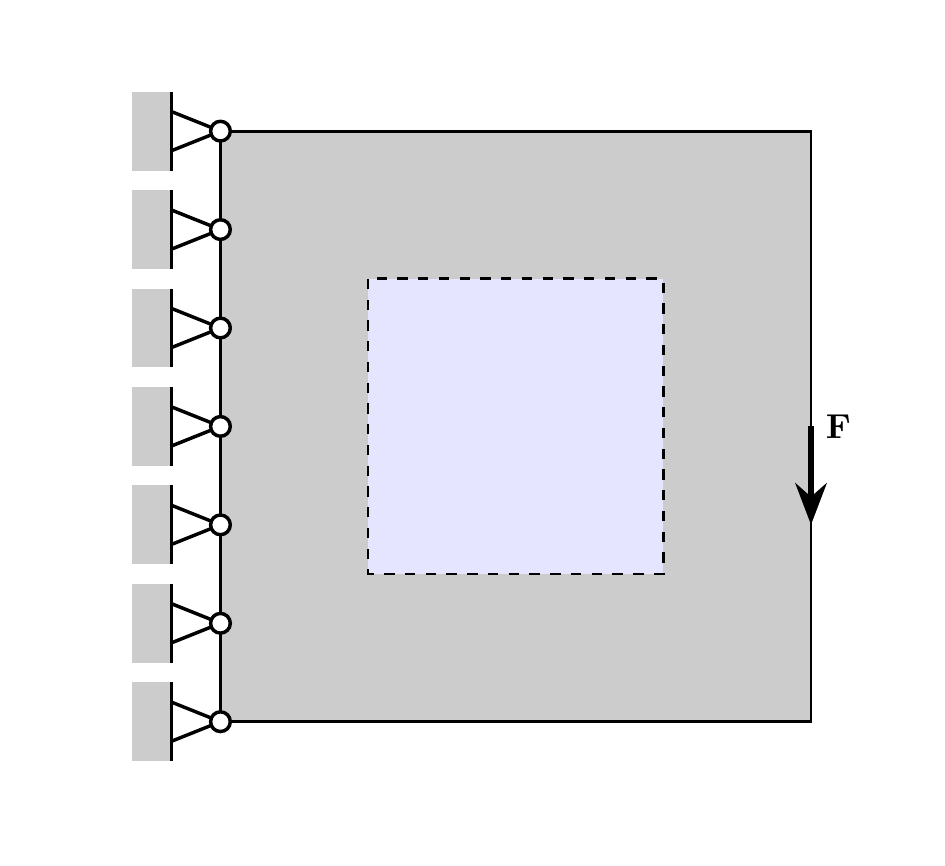} 
\end{minipage}
\begin{minipage}{0.60\textwidth} 
  \resizebox{\textwidth}{!}{
  \begin{tabular}{l*{7}{c}r}
  Model  $\S$             & $\{1\}$ &$\{2\}$ & $\{3\}$ & $\{4\}$ & $\{1, 2\}$& $\{1, 3\}$ & $\{1,4\}$ & $\{2,3\}$ \\
  \hline
  $\gamma_\S$  & 123.8 & 149.3 & 203.9 & 304.8 & 25.2 &  48.6 & 93.7  & 62.2 \\
  \hline
  $m^*_\S$ when $B=10^7$  & 1998 & 2231 & 2337 & 2390 & 1253 &  1657 & 1909  & 2054
  \end{tabular} 
}
  \\[12pt]
  \resizebox{\textwidth}{!}{
  \begin{tabular}{l*{7}{c}r}
  Model $\S$            & $\{2,4\}$ & $\{3,4\}$ & $\{1,2,3\}$ & $\{1,2,4\}$ & $\{1,3,4\}$ & $\{2,3,4\}$ & $\textbf{\{1,2,3,4\}}$  \\
  \hline
  $\gamma_\S$  & 107.6 & 129.7& 11.8 & 11.7 & 14.3 &11.9& \textbf{11.5}  \\
  \hline
  $m^*_\S$ when $B=10^7$   & 2175 & 2292 & 669 & 734 & 976 &1540 & \textbf{638} 
  \end{tabular} 
}
\end{minipage}
\caption{\small Linear elasticity. Left: Geometry and boundary conditions for the linear elastic structure. Right: Oracle scaled loss $\gamma_\S$ \eqref{optm} and optimal exploration sample count $m_\S^*$ \eqref{optm} for different choices of subsets of low-fidelity model indices $\S$. The optimal model subset $\S$ is typed in boldface. Oracle statistics are computed with 50,000 samples.}
\label{fig:struct}
\end{figure}

In this example, we form a multifidelity hierarchy through mesh coarsening via the mesh parameter $h$.
The compliance $E$ is our scalar-valued QoI for every model, i.e., $d = d_i = 1$ for $i=1,\ldots, 4$.
A mesh parameter of $h =2^{-7}$ corresponds to the high-fidelity model $Y$, and $h = 2^{-4}, 2^{-3}, 2^{-2}, 2^{-1}$ correspond to lower-fidelity models $X_1, \ldots, X_4$, respectively.

The cost for each model is the computational time, which we take to be inversely proportional to the mesh size squared, i.e., $h^2$. This corresponds to using a linear solver of optimal linear complexity.
We normalize cost so that the model with the lowest fidelity has unit cost, i.e., $(c_0, c_1, c_2, c_3, c_4)= (4096, 64, 16, 4, 1)$. 
The correlations between the QoIs of $Y$ and $X_1, X_2, X_3, X_4$ are $0.976$, $0.940$, $0.841$, $-0.146$, respectively.  
The total budget $B$ is taken on the interval $[10^5, 10^7]$. 

\begin{figure}[htbp]
  \centering 
\begin{subfigure}{0.325\textwidth}{\includegraphics[width=\linewidth, trim={0.2cm 0.2cm 1cm 2cm},clip]{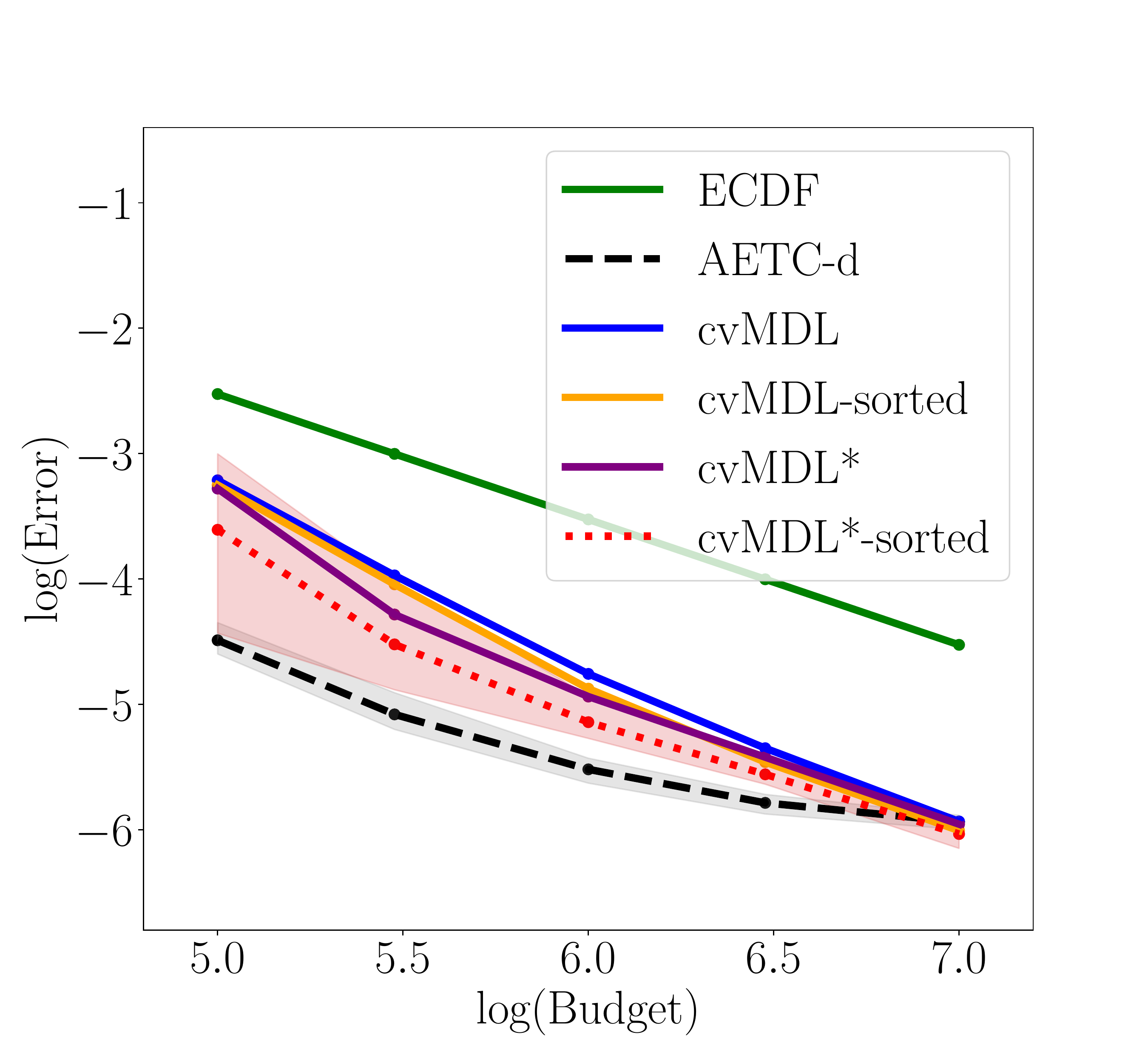}}\caption{}\label{num1:(a)}
\end{subfigure} 
\begin{subfigure}{0.325\textwidth}{\includegraphics[width=\linewidth, trim={0.2cm 0.2cm 1cm 2cm},clip]{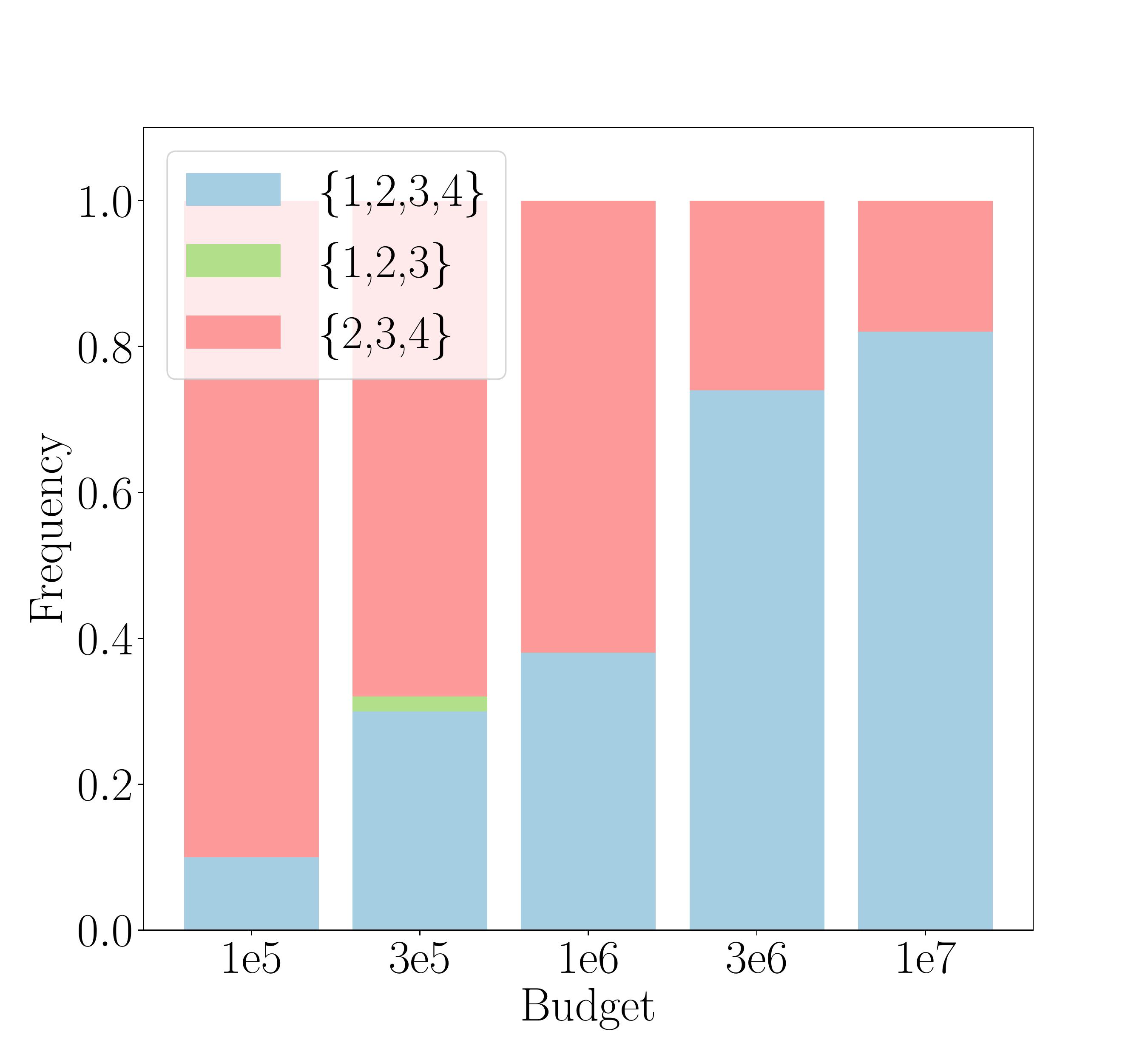}}\caption{}\label{num1:(b)}
\end{subfigure} 
\begin{subfigure}{0.325\textwidth}{\includegraphics[width=\linewidth, trim={0.05cm 0.2cm 1cm 2cm},clip]{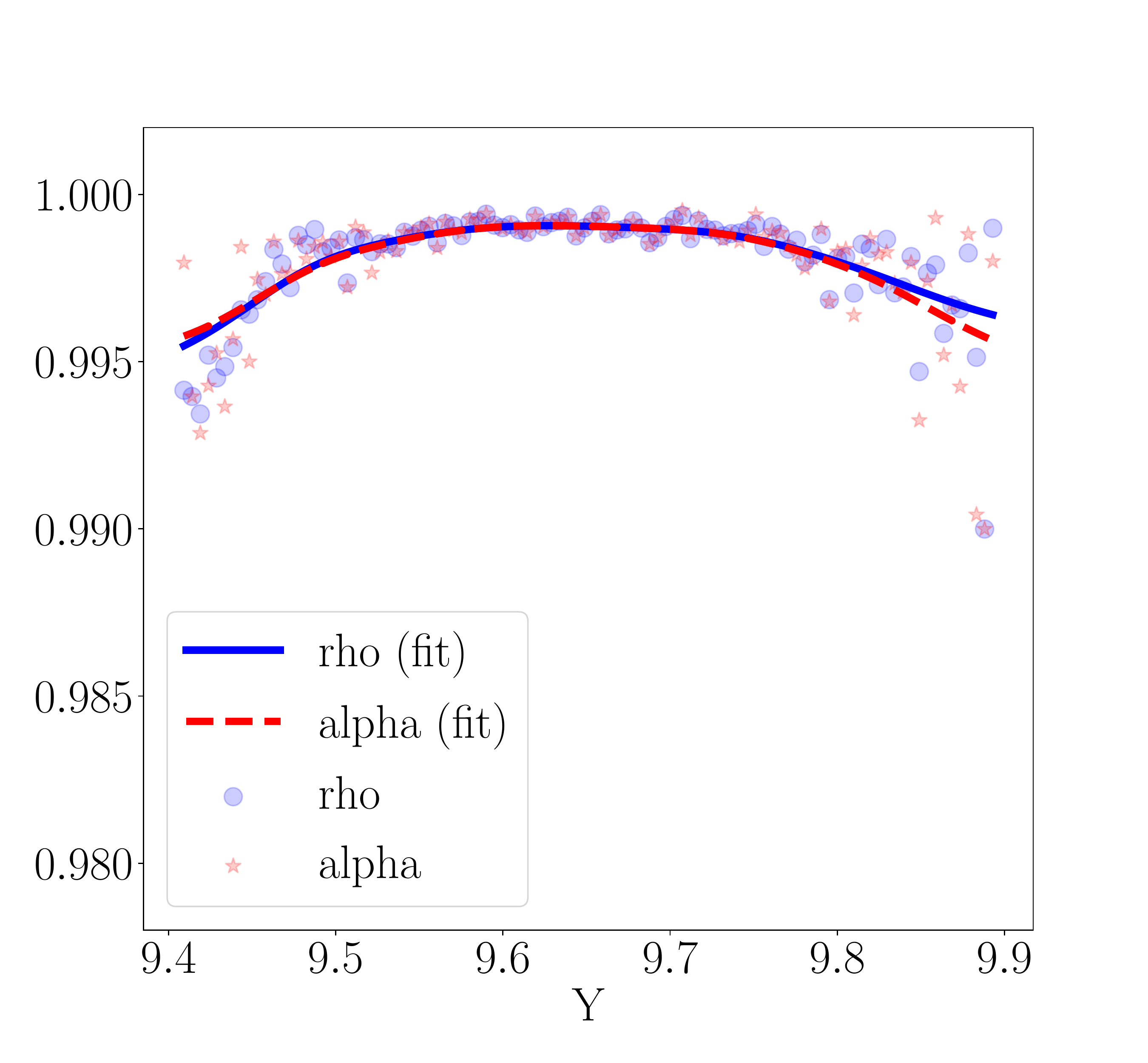}}\caption{}\label{num1:(c)}
\end{subfigure} 
\caption{Linear elasticity. (a). Mean $\omega(x)$-weighted $L^2$ error between $F_Y$ and the estimated CDFs given by ECDF, AETC-d, cvMDL, and its variants, with the $5\%$-$95\%$ quantiles (for ease of visualization, we only plot the quantiles for AETC-d and cvMDL*-sorted) to measure the uncertainty.
(b). Frequency of different models selected by cvMDL.  
(c). Scatter plot of the estimated $\alpha(x)$ and $\rho(x)$ when $\S = \{1,2,3,4\}$ using 50,000 i.i.d. samples in the $1\%$-$99\%$-quantile regime of $Y$. Gaussian kernel smoothing is applied to both data Gaussian kernel with bandwidth $\text{bd} = 0.0358$ chosen using $5$-fold cross-validation. } \label{pde-comp}
\end{figure} 

\subsubsection{Results for estimating the CDF}

\Cref{pde-comp} shows the simulation results given by different multifidelity estimators as well as more refined statistics of the cvMDL-related algorithms.  
In \Cref{pde-comp}(a), we see that AETC-d has the smallest error for smaller budgets but its asymptotic convergence is constrained by the model misspecification effects (associated with theoretical assumptions on the applicability of AETC-d), i.e., the error curve starts to plateau when $B$ exceeds $10^6$. 
Although this can be mitigated by including additional nonlinear (e.g. polynomial) terms as additional covariates, trustworthy practical guidance is still lacking for this approach. 
On the other hand, both cvMDL and its variants demonstrate superior performance over ECDF, with cvMDL*-sorted achieving a result competitive to AETC-d without the plateau effect.

In \Cref{pde-comp}(b), we note that as the budget increases, the model $\widehat{\S}^\ast$ selected by cvMDL converges to $\{1,2,3,4\}$, which is the same as the optimal model computed under oracle statistics in Figure \ref{fig:struct} (right). 
We note that the suboptimal model $\S = \{2, 3, 4\}$ is selected often by cvMDL, but 
not other models whose $\gamma_\S$ is close to that of $\{1,2,3,4\}$.
We believe this occurrence is due to the aggressive exploration steps taken by cvMDL, in particular when we double exploration samples causing suboptimal models $\S$  with large values of $m^\ast_{\S}$ (e.g., $\S = \{2, 3, 4\}$) become the preferred model.

The significant error reduction achieved by cvMDL is indicated by near-unity values of $\rho_{\S}(x) = \text{Corr}[\mathbf 1_{\{Y\leq x\}}, h(X_{\S}; x)]$ where $\S = \{1,2,3,4\}$, shown in \Cref{pde-comp}(c).
For cvMDL variants, either estimating $\alpha(x)$ in the tail regime through \eqref{myalpha} (cvMDL*) or sorting CDF values to ensure monotonicity (cvMDL/cvMDL*-sorted) can help further reduce the errors.
The former is particularly helpful in the small-budget regime where exploration data are not sufficient to estimate the full support of the QoI.

The weight function $\omega(x)$ in this scenario is constant on $\R$ thus the estimates produced by cvMDL-type estimators are expected to capture the global structure of $F_Y$ (e.g. bulk and tails).  
To inspect this, we fix $B = 10^7$ and plot the estimated CDFs in the tail and bulk regimes separately. 
The CDFs of the pointwise errors (at $1000$ discretization points) in the three regimes are shown in Figure \ref{fig:1r}.
It can be seen that cvMDL*-sorted has the smallest errors in all three regimes.
 
\begin{figure}[htbp]
\centering 
\begin{subfigure}{0.325\textwidth}{\includegraphics[width=\linewidth, trim={0.2cm 0.2cm 1cm 2cm},clip]{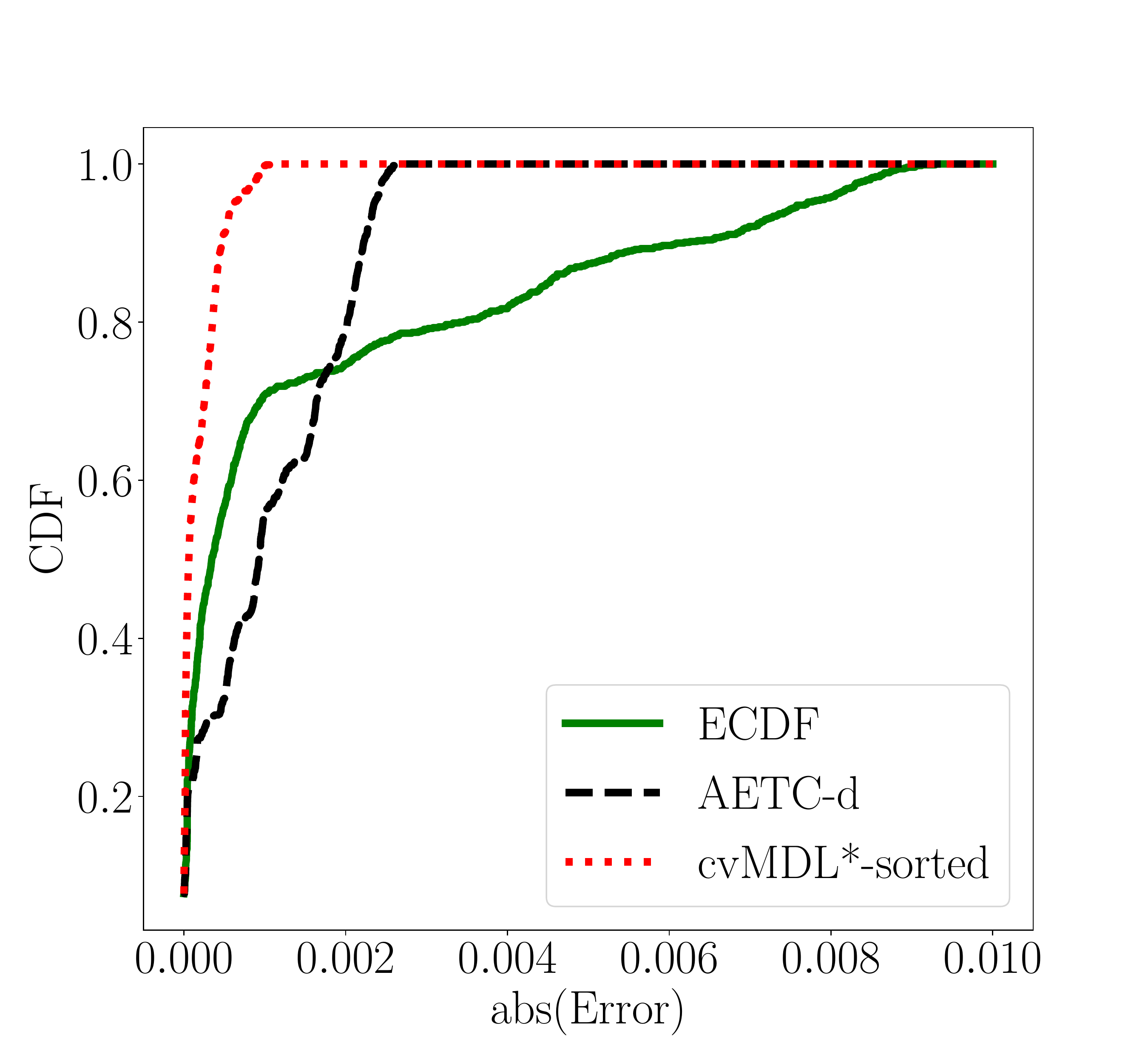}}\caption{}\label{Pe:(a)}
\end{subfigure} 
\begin{subfigure}{0.325\textwidth}{\includegraphics[width=\linewidth, trim={0.2cm 0.2cm 1cm 2cm},clip]{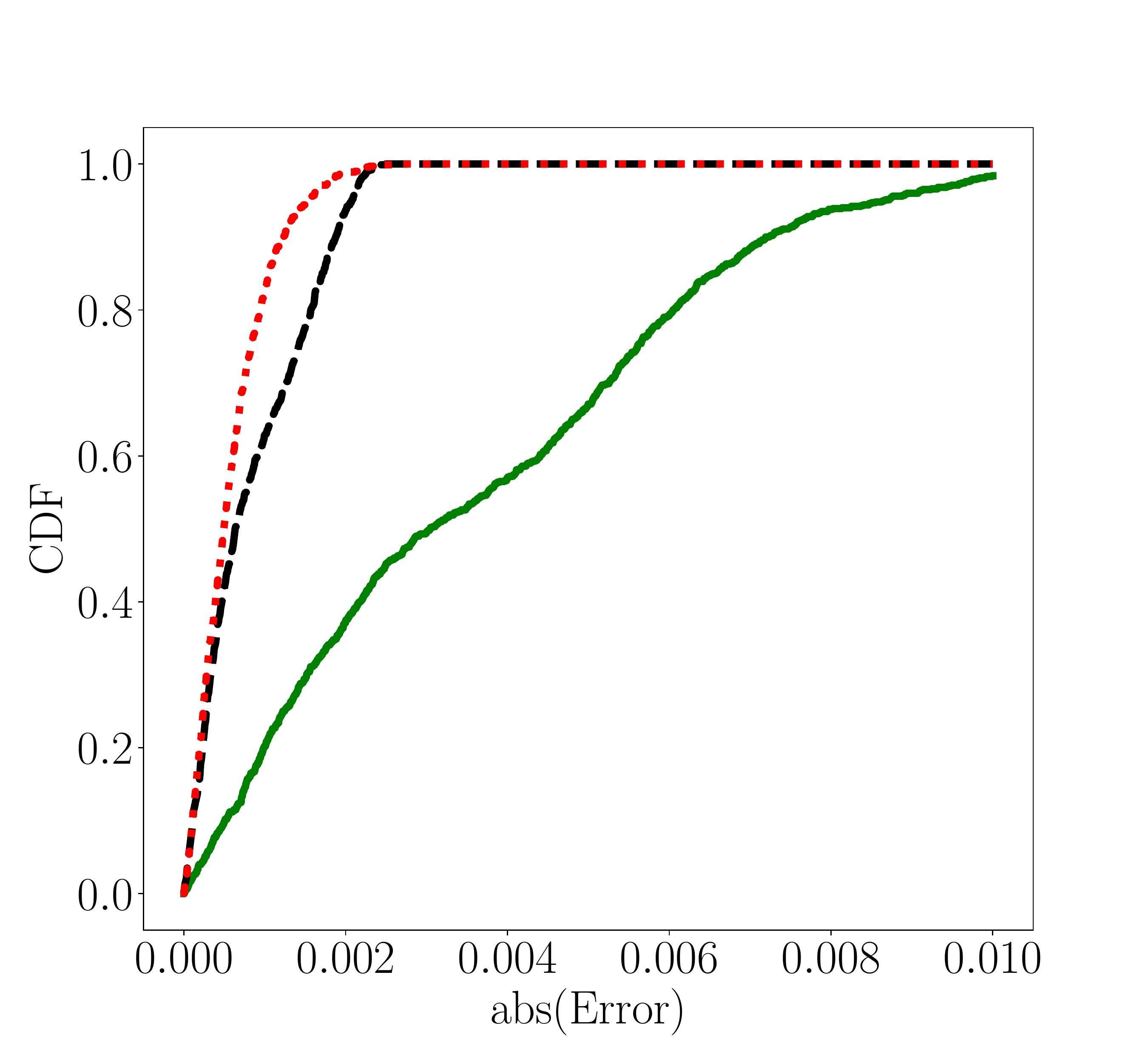}}\caption{}\label{Pe:(b)}
\end{subfigure} 
\begin{subfigure}{0.325\textwidth}{\includegraphics[width=\linewidth, trim={0.2cm 0.2cm 1cm 2cm},clip]{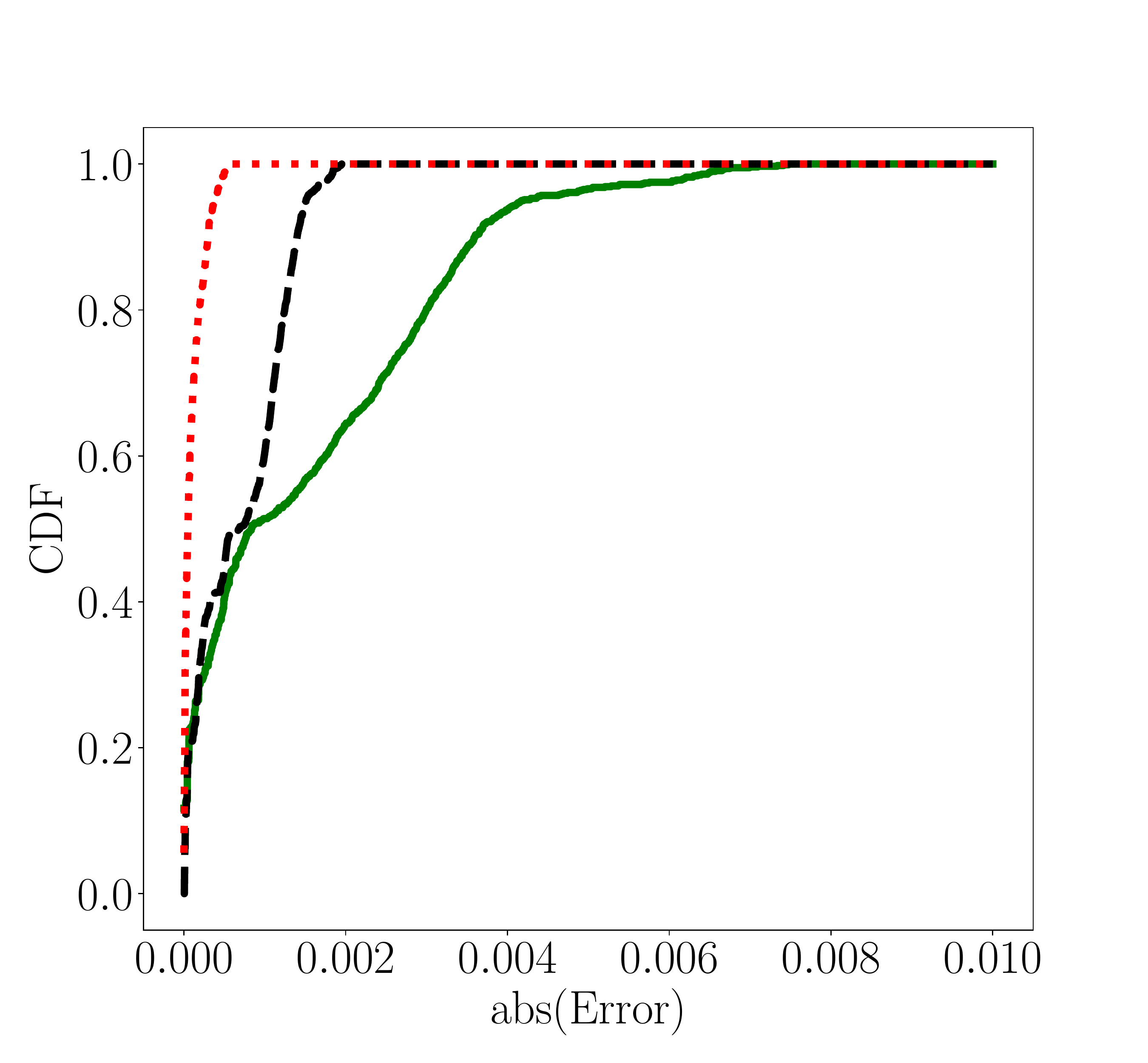}}\caption{}\label{Pe:(c)}
\end{subfigure} 
\caption{Linear elasticity. One realization of the distribution of pointwise CDF errors computed by cvMDL*-sorted, AETC-d, and ECDF for budget $B = 10^7$. We plot CDFs of errors in three different regimes: (a) the lower tail of $Y$ defined by the $0-0.05$ quantile region, (b) the bulk defined by the $0.05 - 0.95$ quantile region, (c) the upper tail defined by the $0.95 - 1.00$ quantile region.}
\label{fig:1r}
\end{figure} 

\subsubsection{Risk metrics}
We now compare some risk metrics of the estimated CDFs.
For example, one frequently used metric is the conditional value-at-risk (CVaR), also called the expected shortfall, which is defined as the conditional expectation of $Y$ in a tail regime (here, $Y$ being large):
\begin{align*}
&\textsf{CVaR}_a(Y) := \E[Y|F_Y(Y)\geq a] = \frac{1}{1-a}\int_{a}^1 F_Y^{-1}(x)\text{d}x & 0<a<1,
\end{align*}
where $a$ is the quantile level. 
Assuming $F_Y$ is known, $\textsf{CVaR}$ can be numerically computed using linear interpolation of $F^{-1}_Y$.
Fixing $B = 10^7$ as before, we use the estimated CDFs by ECDF, AETC-d, and cvMDL*-sorted to compute the CVaR of $Y$ for $a = 0.95$ and $0.99$, respectively. 
The experiment is repeated $50$ times, and the corresponding statistics are summarized using boxplots in Figure \ref{optimal} (a)-(b). 
For both choices of $a$, cvMDL*-sorted outperforms the other methods by a noticeable margin. 
It is worth noting that although AETC-d and cvMDL*-sorted have similar errors under the tested budget globally (Figure \ref{pde-comp} (a)), the model misspecification effects result in the former estimates being biased upward. 
The cvMDL*-sorted estimates, on the other hand, remain unbiased.

\begin{figure}[htbp]
 \begin{center}
\begin{subfigure}{0.327\textwidth}{\includegraphics[width=\linewidth, trim={0.01cm 0.2cm 1cm 2cm},clip]{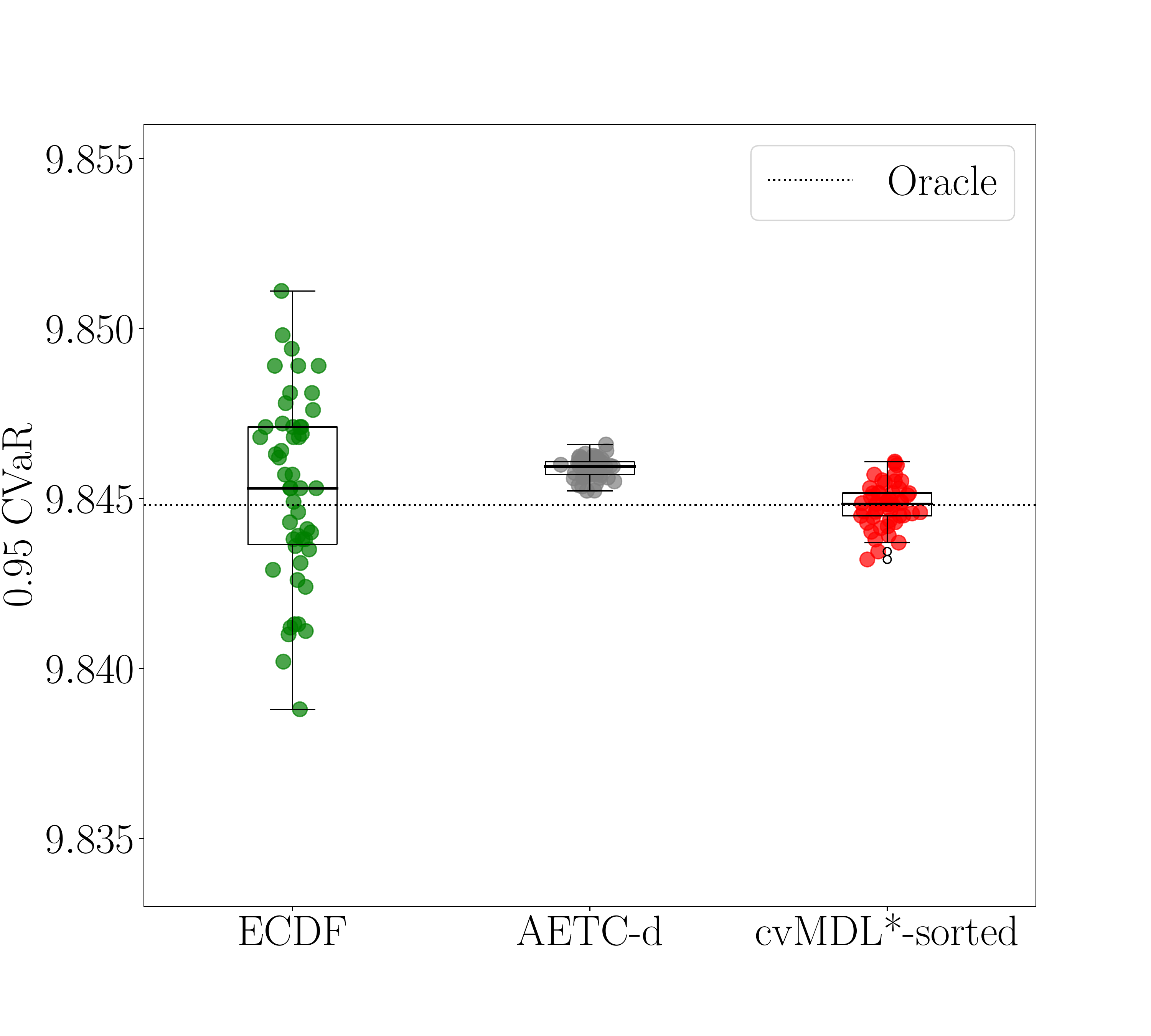}}\caption{}\label{cvar_1:(a)}
\end{subfigure} 
\begin{subfigure}{0.327\textwidth}{\includegraphics[width=\linewidth, trim={0.01cm 0.2cm 1cm 2cm},clip]{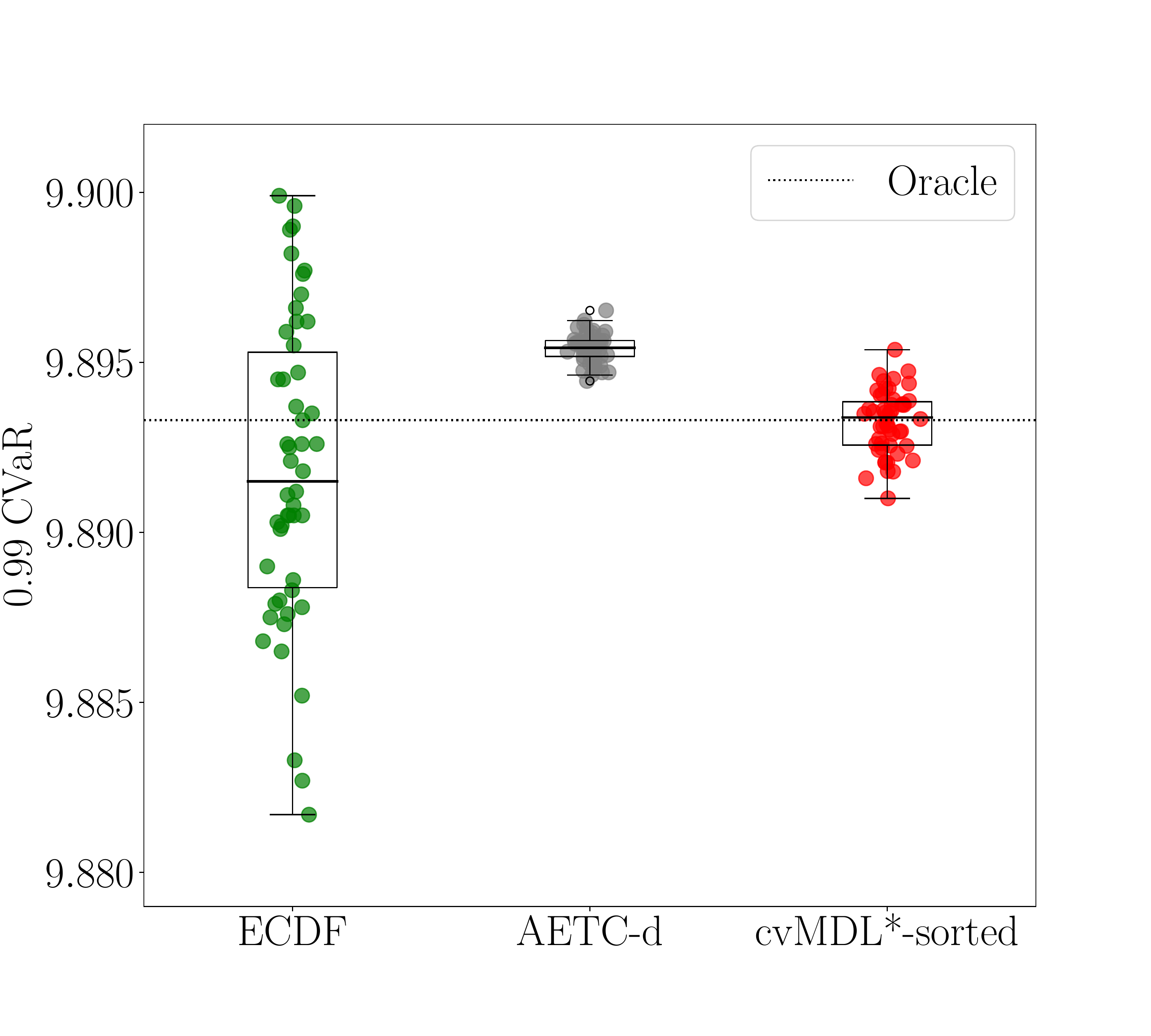}}\caption{}\label{cvar_2:(b)}
\end{subfigure} 
\begin{subfigure}{0.327\textwidth}{\includegraphics[width=\linewidth, trim={0.2cm 0.2cm 1cm 2cm},clip]{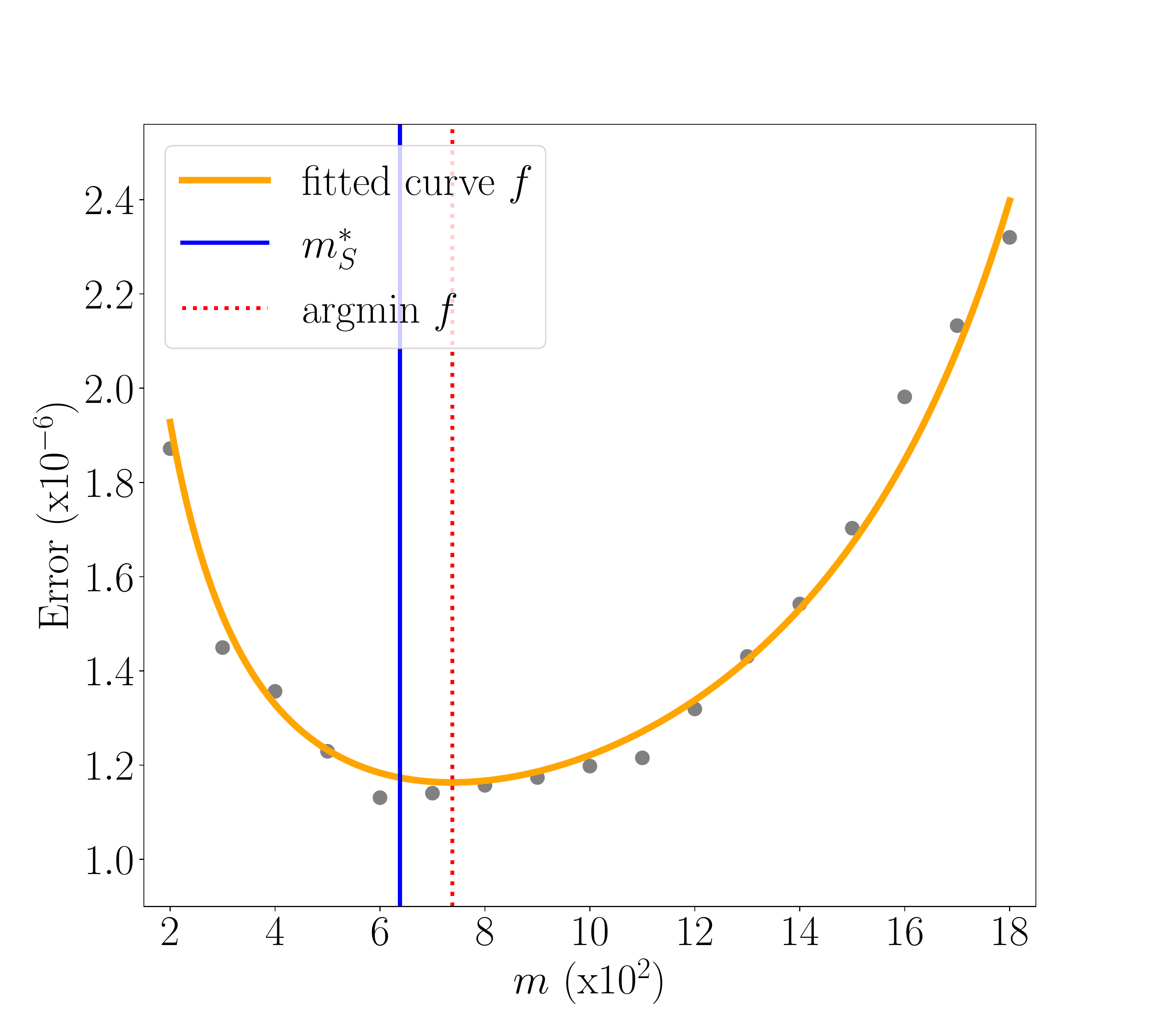}}\caption{}\label{pde_opt:(c)}
\end{subfigure} 
\caption{Linear elasticity. (a): Boxplots of the $\textsf{CVaR}_{0.95}(Y)$ computed using the estimated CDFs given by ECDF, AETC-d, and cvMDL*-sorted when $B = 10^7$ with 50 experiments. (b): Same but for $\textsf{CVaR}_{0.99}(Y)$.  (c): Inspection of how well the estimated loss $\widehat{L}_\S$ mimics the oracle loss curve as a function of $m$.  The discrete data are fitted using a function $f(m; a, b)$ of the form $\frac{a}{m} + \frac{b}{B-c_\ex m}$, with fitted values for $a$ and $b$ given by $2.64\times 10^{-4}$ and $5.57$, respectively. }\label{optimal}
\end{center}
\end{figure}

\subsubsection{Oracle versus estimated loss}
We investigate the model selection criteria used in cvMDL. 
Note for each $\S\subseteq \{1:n\}$, there is a discrepancy between the exact loss function versus the estimator $\widehat{L}_\S$ constructed with empirical data. We inspect if this approximation is reasonable. 
To numerically determine if the exploration-exploitation trade-off is optimal, we fix $B = 10^7$ and $\S = \{1,2,3,4\}$.
For a given value of $m$, we first take $m$ exploration samples to estimate the control variates parameters and then use them to build an estimator $\widetilde{F}_\S$ for $F_Y$ as in \eqref{fug}.  
We then compute the (exact) mean weighted $L^2$ loss associated with this value of $m$. 
We repeat the experiment $10$ times and compute the average loss value.
We compile results of the above for $m$ in the range from $200$ to $1800$.
The results are reported in Figure \ref{optimal} (c). 
It can be seen that the optimal exploration rate under oracle loss $L_\S$, $638$ (see \Cref{fig:struct}, right), almost matches the empirically identified optimal exploration rate, which is $736$. The small gap can be attributed to the underestimation of exploration error committed due to the finite-sample estimation of parameters.

\subsection{Extrema of Geometric Brownian Motion}\label{sec:num2}

Geometric Brownian motion is a continuous-time stochastic process that is widely used in financial modeling.    
In a simple setting, a geometric Brownian motion $S_t$ is a random process with a constant initial state $s_0 > 0$ whose evolution is described using the stochastic differential equation
\begin{align*}
  dS_t &= \mu S_t dt + \sigma S_t dW_t & t& \geq 0, & S_0 &= s_0,
\end{align*}
where both $\mu$ and $\sigma>0$ are constants, and $W_t$ is a standard Brownian motion process. 
A unique explicit solution for $S_t$ exists and can be written as 
\begin{align*}
S_t = s_0\exp\left(\left(\mu-\frac{\sigma^2}{2}\right)t + \sigma W_t\right).
\end{align*}
Set $\mu = 0.05, \sigma = 0.2, s_0 = 1$.
We are interested in estimating the joint distribution of the extreme values of $S_t$ over the time interval $[0,1]$:
\begin{align*}
&(S_{\min}, S_{\max})^\top\in\R^2 & S_{\min}:=\min_{0\leq t\leq 1}S_t, \;\; S_{\max}:=\max_{0\leq t\leq 1}S_t.
\end{align*}
We thus choose as the QoI the random vector $(S_{\min}, S_{\max})^\top$. We evaluate these quantities by discretizing the stochastic differential equation in time using the Euler--Maruyama scheme with time step $\Delta t$ and computing the discrete extrema.
The computational complexity (cost) of the corresponding procedure is proportional to the number of grid points used for discretization. 

We construct a multifidelity model for this problem based on time discretization. In particular, we consider four different time scales $\Delta t\in \{2^{-4}, 2^{-6}, 2^{-8}, 2^{-14}\}$, with the high-fidelity model $Y$ corresponding to $\Delta t = 2^{-14}$ and $X_1, X_2, X_3$ corresponding to $\Delta t = 2^{-8}, 2^{-6}, 2^{-4}$, respectively. This results in (normalized) model costs $(c_0, c_1, c_2, c_3) = (1024,16,4,1)$. The total budget $B$ takes values in $[10^4, 10^6]$. To generate joint samples, the randomness of $W_t$ is simulated from the same realization used in the high-fidelity model. The oracle CDF of the high-fidelity model is computed using MC with $10^5$ samples. The oracle correlations between the QoIs of the high- and low-fidelity models in Table \ref{1980}. 

\begin{table}[htbp]
\begin{center}
\small
\caption{Geometric Brownian motion. Oracle correlations between the high-fidelity and low-fidelity model QoIs computed using 50,000 samples. }\label{1980}
\begin{tabular}{l*{7}{c}r}
Model QoIs            & $S_{\min} (2^{-8})$ & $S_{\max} (2^{-8})$ & $S_{\min} (2^{-6})$ & $S_{\max} (2^{-6})$ & $S_{\min} (2^{-4})$ & $S_{\max} (2^{-4})$   \\
\hline
$S_{\min} (2^{-14})$  & 0.999 & 0.682 & 0.997 & 0.682 & 0.984 &0.680 \\
\hline
$S_{\max} (2^{-14})$     & 0.681 & 0.999 & 0.681 & 0.998 & 0.674 & 0.988 \\
\end{tabular} 
\end{center}
\end{table}
In this example, all model QoIs are two-dimensional random vectors so AETC-d cannot be directly applied. 
For cvMDL and its variants, setting $\omega(\x)\equiv 1$ violates \Cref{ass:omega}.
Instead, since $S_{\min}\leq s_0 = 1\leq S_{\max}$, we choose $\omega(\x) = \mathbf 1_{\mathcal T}(\x)$ as an indicator function on a two-dimensional bounded region $\mathcal T\subset \R^2$ where the most likely outcomes reside. For instance, here we take $\mathcal T = [0.5, 1]\times [1, 3]$. 
The statistics of the estimation errors and the selected models by cvMDL are reported in Figure \ref{sde-comp}(a),(b). Panel (c) shows that the correlation coefficient $\rho_\S(\x)$, is close to unity over the entire domain, suggesting that our chosen control variate \eqref{eq:rhoS-def} is a good choice.

\begin{figure}[htbp]
\centering 
\begin{subfigure}{0.325\textwidth}{\includegraphics[width=\linewidth, trim={0.02cm 0.2cm 1cm 2cm},clip]{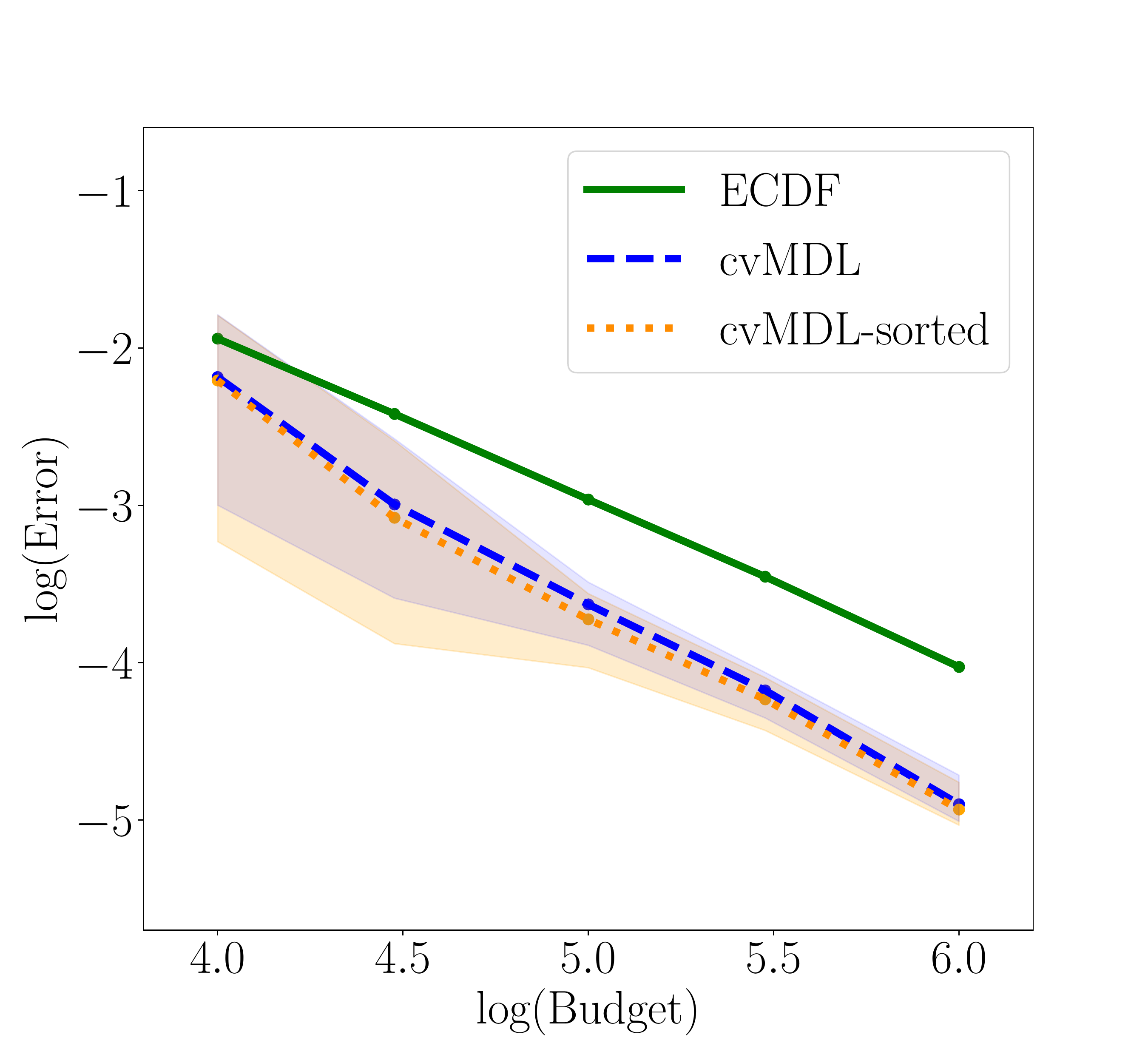}}\caption{}\label{sde_quantile:(a)}
\end{subfigure} 
\begin{subfigure}{0.325\textwidth}{\includegraphics[width=\linewidth, trim={0.02cm 0.2cm 1cm 2cm},clip]{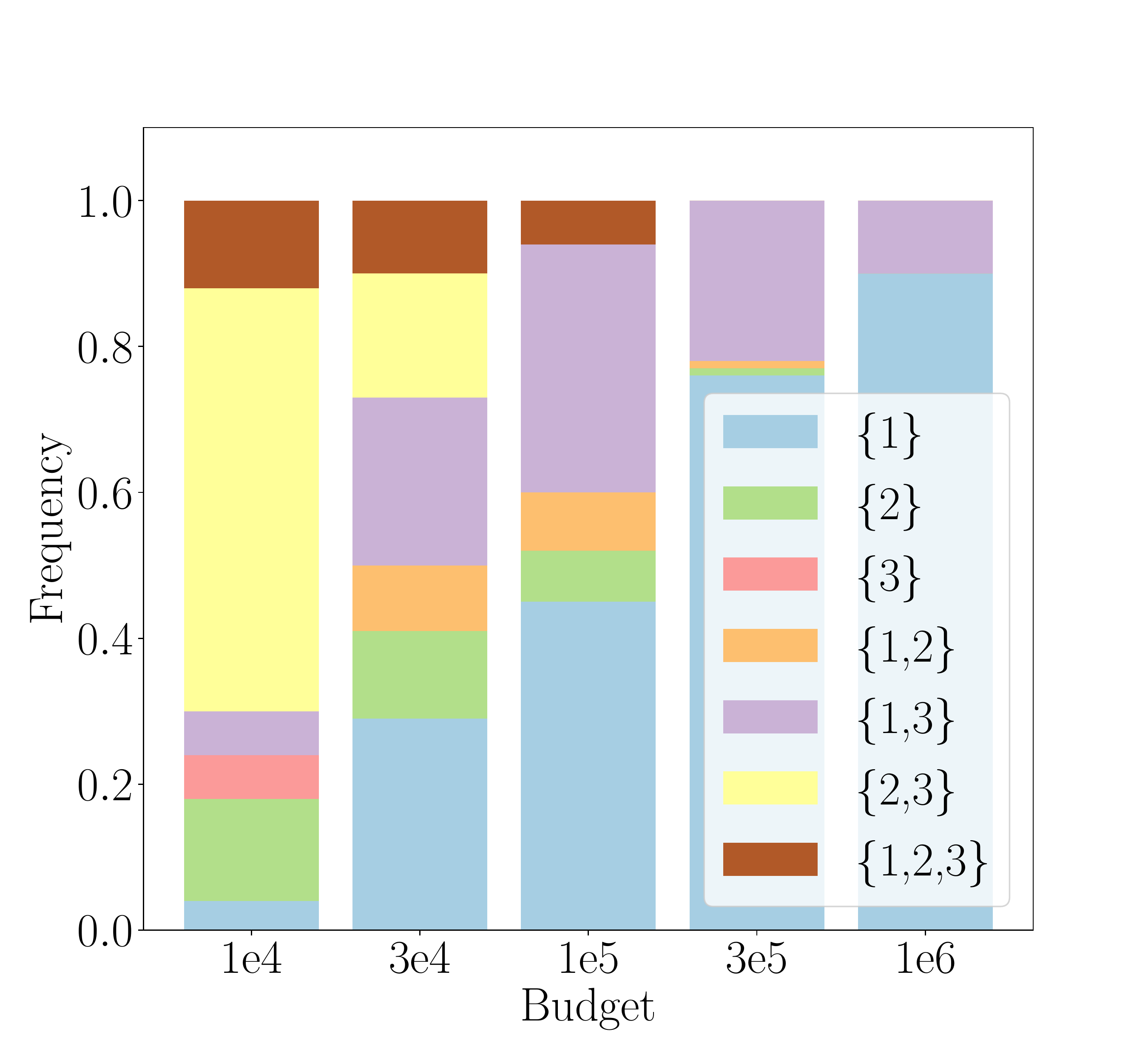}}\caption{}\label{sde_model:(b)}
\end{subfigure} 
\begin{subfigure}{0.325\textwidth}{\includegraphics[width=\linewidth, trim={0.02cm 0.2cm 1cm 2cm},clip]{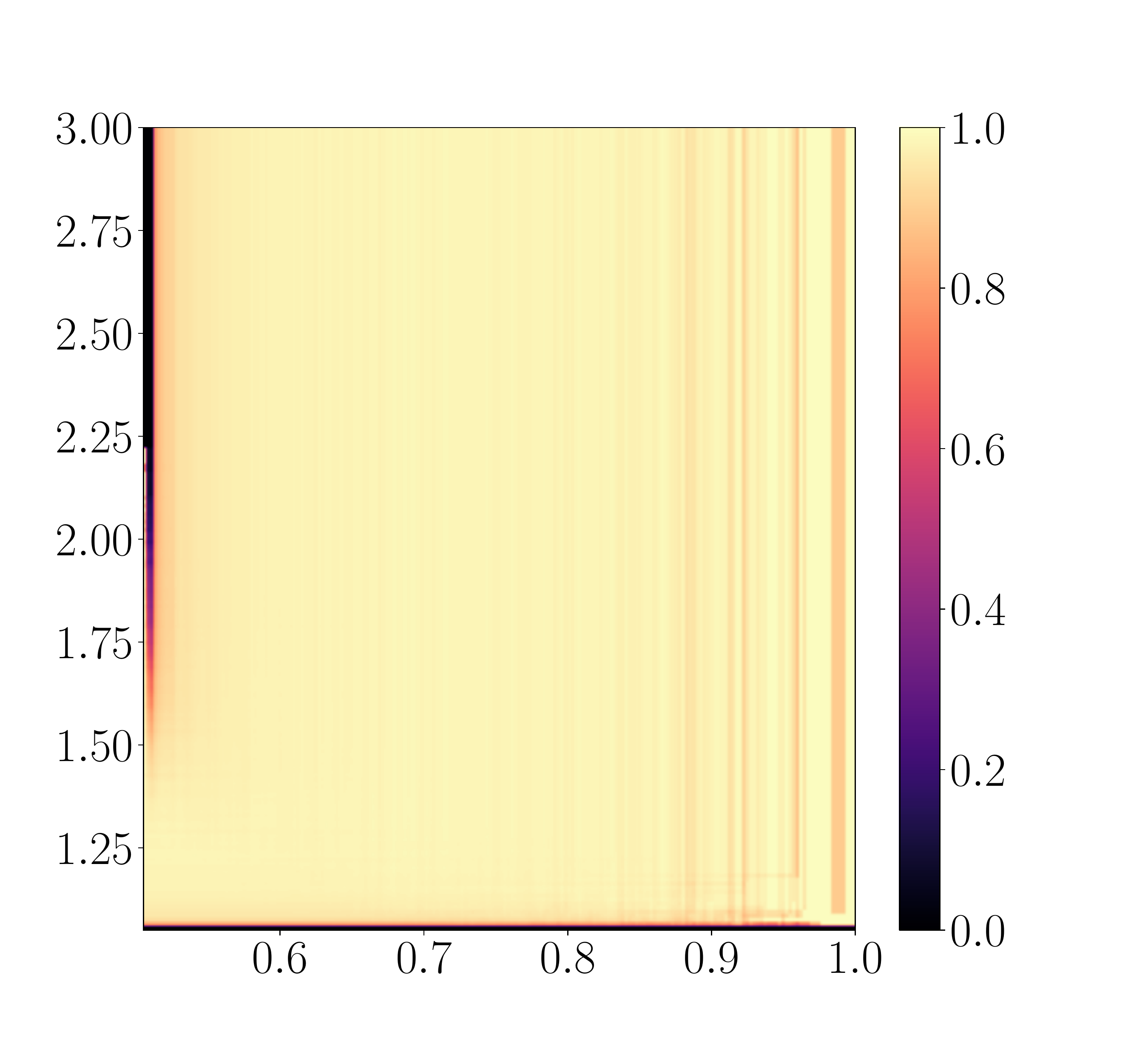}}\caption{}\label{sde_rho:(c)}
\end{subfigure} 
\caption{Geometric Brownian motion. (a): Mean $\omega(\x)$-weighted $L^2$ error between $F_Y$ and the estimated CDFs given by ECDF, cvMDL, and cvMDL-sorted, with the $5\%$-$50\%$-$95\%$ quantiles to measure the uncertainty.
  (b): Frequency of different models selected by cvMDL; cf. optimal model losses in \Cref{oracle}. (c) Estimated $\rho_\S(\x)$ from \eqref{eq:rhoS-def} when $\S = \{1\}$ using 50,000 i.i.d. samples for $\x\in \mathcal T$.} \label{sde-comp}
\end{figure}

\Cref{sde-comp} shows that cvMDL is consistent on the region $\mathcal T$, and the corresponding estimation error is on average much lower than that of ECDF. 
As the budget goes to infinity, the model selected by cvMDL converges to the single low-fidelity model $\{1\}$, which coincides with the optimal model computed using oracle statistics.  
With additional sorting to stabilize the algorithm, cvMDL-sorted slightly further reduces the errors of cvMDL, which is consistent with the observations in the 1d case. 
In the pre-asymptotic regime when the budget is small, the models selected by cvMDL have relatively large fluctuations, but these stabilize for larger budgets.
More results from this experiment are presented in \Cref{sec:num2-details}.

\subsection{Brittle fracture behavior of a fiber-reinforced matrix}\label{sec:num3}

We investigate a two-dimensional fiber-reinforced matrix, a subject commonly explored in the field of fracture mechanics. Our QoI is the maximum load that induces brittle fracture within the matrix region adjacent to the fiber. To obtain the QoI, we solve a quasi-static, two-dimensional finite element problem. 

\begin{figure}[htbp]
  \begin{minipage}[b][][b]{0.30\textwidth}
    {\includegraphics[width=\textwidth]{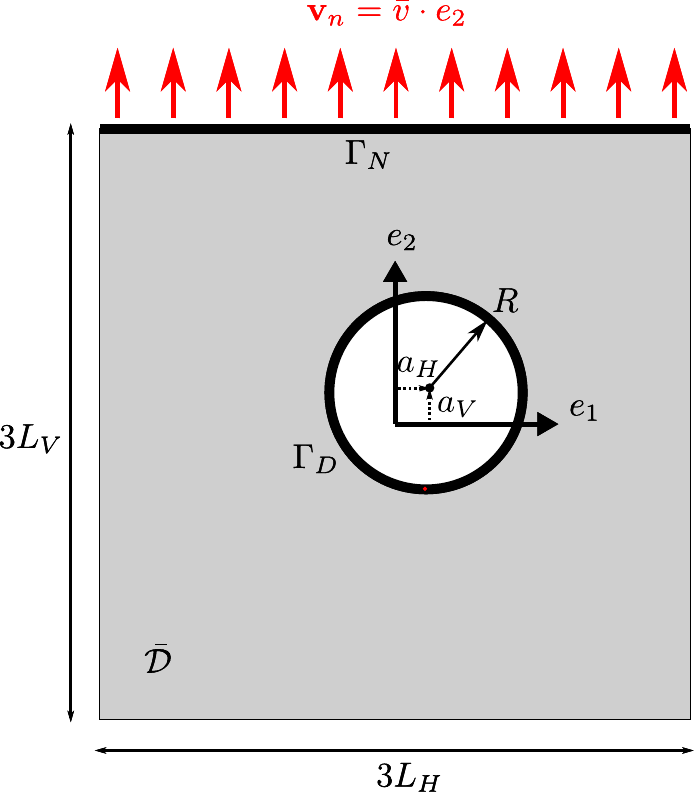}}
  \end{minipage}
  \begin{minipage}[b][][c]{0.68\textwidth}
    \resizebox{\textwidth}{!}{
    \begin{tabular}{ccccccc}
    \toprule 
    Random  & \multirow{2}{*}{Property} & \multirow{2}{*}{Mean} & \multirow{2}{*}{$\mathrm{COV}_i$ (\%)$^{\mathrm{a}}$} & Lower  & Upper & Probability \tabularnewline
    variable &  &  &  & boundary & boundary & distribution\tabularnewline
    \midrule
    $p_{1}$ & $\kappa$ (GPa) & $55$ & $5$ & $45$ & $60$ & Truncated normal\tabularnewline
    $p_{2}$ & $\nu$ & $0.25$ & $5$ & $0.20$ & $0.30$ & Truncated normal\tabularnewline
    $p_{3}$ & $G_{c}$ (GPa) & $1$ & $10$ & 0 & $\infty$ & Lognormal\tabularnewline
    $p_{4}$ & $L_{H}$(cm) & $1$ & $-$ & $0.9$ & $1.1$ & Uniform\tabularnewline
    $p_{5}$ & $L_{V}$ (cm) & $1$ & $-$ & $0.9$ & $1.1$ & Uniform\tabularnewline
    $p_{6}$ & $R$ (cm) & $0.5$ & $-$ & $0.4$ & $0.6$ & Uniform\tabularnewline
    $p_{7}$ & $a_{H}$ (cm)  & $0$ & $-$ & $-0.05$ & $0.05$ & Uniform\tabularnewline
    $p_{8}$ & $a_{V}$ (cm) & $0$ & $-$ & $-0.05$ & $0.05$ & Uniform\tabularnewline
    \bottomrule
    \end{tabular}
  }
    \begin{tablenotes}
    \scriptsize\smallskip 
  \item{a.} Coefficient of variation $\mathrm{COV}_i=100\times\sigma_i/\mathbb{E}[p_i]$ for $i=1,\ldots,8.$\\[3pt]
    \end{tablenotes}
  \end{minipage}
  \caption{Fiber-reinforced matrix. Left: Geometry, loading, and boundary condition. We consider the domain $\mathcal{D}=[-1.5L_H, 1.5L_H]\times[-1.5 L_V, 1.5L_V]\in\mathbb{R}^2$ including a circular hole of radius $R$ at $(a_H,a_V)$ in the ${e_1}$ and ${e_2}$ directions of the center lines. Right: Properties of the eight random inputs in the fiber-reinforced matrix. Here, $\kappa$ is the Young's modulus and $\nu$ is the Poisson ratio, see \Cref{sec:num3-details} for details.}
\label{fig:ex3.1}
\end{figure} 
Figure~\ref{fig:ex3.1} (left) shows a square plate of length $3L_H$ in the $e_1$ direction and $3L_V$ in  the direction with a circular inclusion of radius $R$. In the domain $\mathcal{D}=[-1.5L_H, 1.5L_H]\times[-1.5 L_V, 1.5L_V]\in\mathbb{R}^2$, the loading is given by an applied normal displacement ${\mathbf{v}_n}=\bar{v}\cdot e_2$ on the boundary $\Gamma_N$. The other boundaries, denoted $\Gamma_D$, are free, corresponding to zero displacement on $\Gamma_D$. 
The closure of the domain is $\bar{\mathcal{D}}\equiv\ \mathcal{D}\cup\Gamma$, with $\Gamma=\Gamma_D\cup\Gamma_N$. The unknowns are the displacement field $\mathbf{u}=(u,v)^{\top}\in\mathbb{R}^2$ and the scalar damage variable $\phi_d\in\mathbb{R}$ in the domain $\bar{\mathcal{D}}$ of the elastic body. 
This setup models the traction experiment of a fiber-reinforced matrix \cite{amor2009regularized,bourdin2000numerical}, with the corresponding boundary value problem as described in \cite{natarajan2019fenics}. The PDE formulation is: find $\mathbf{u}(\mathbf{x})$ and $\phi_d(\mathbf{x})$ for $\mathbf{x}=(x_1,x_2)^{\top}\in\bar{\mathcal{D}}$, such that
\begin{align} 
[(1-\phi_d(\mathbf{x}))^2+q]\nabla\cdot \bm{\sigma}(\mathbf{x})&=\bf{0},\label{disp_gov} \\
-G_c\ell_0\nabla^2\phi_d(\mathbf{x})+\left[\dfrac{G_c}{\ell_0}+2H(\mathbf{x})\right]\phi_d(\mathbf{x})&=2H(\mathbf{x}),\label{damage_gov}
\end{align} 
with corresponding boundary conditions on $\Gamma_N$ and $\Gamma_D$. The full model details, including definitions for $\ell_0$, $\bm{\sigma}$, $q$, and $H$ are shown in \Cref{sec:num3-details}. We consider eight input random variables, $p_1, \ldots p_8$, which are stemming from material properties and geometries, see Figure~\ref{fig:ex3.1} (right).

\subsubsection{High-fidelity and low-fidelity models}
\begin{figure*}
\centering
{\includegraphics[width=0.9\textwidth]{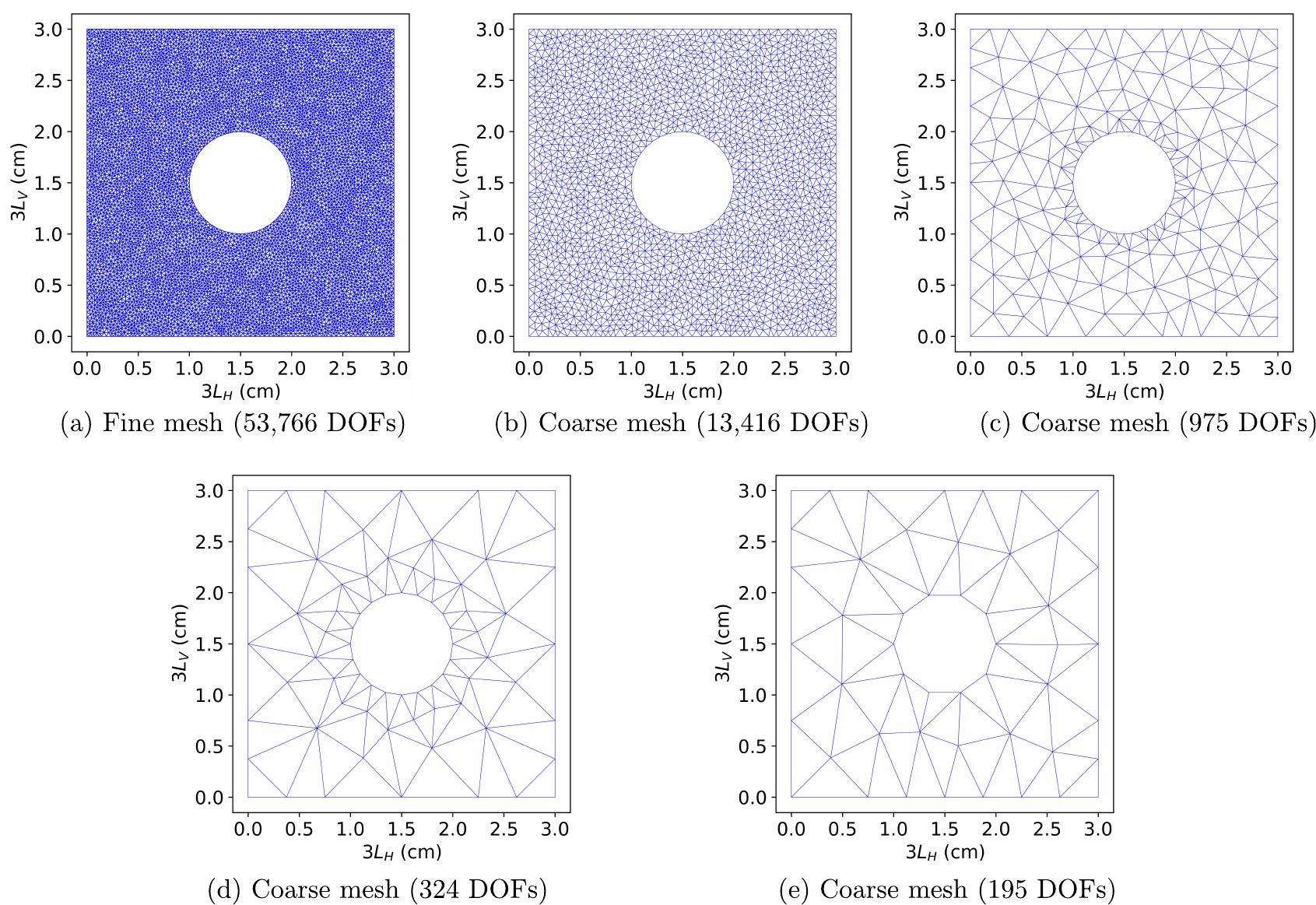}}\hfill
\caption{Fiber-reinforced matrix. The fine finite-element mesh in (a) is used to generate the high-fidelity QoI $Y$, while the coarse meshes in (b)--(e) are used to generate the low-fidelity 
QoIs $X_1,X_2,X_3,X4$, respectively.}
\label{fig:ex3.2}
\end{figure*} 
The model~\eqref{disp_gov} and \eqref{damage_gov} for brittle fracture analysis is solved using an iterative solver, wherein we solve for the scalar damage variable ($\phi_d$) using the displacement fields ($\mathbf{u}$). Subsequently, the updated damage variable is used to solve for the displacement field, and the process is repeated until the difference between the current and previous iterates becomes less than the user-defined tolerance $\delta \ll 1$. We set $\delta=5\times 10^{-3}$ for the high-fidelity model, and $\delta = 5\times 10^{-2},~0.2,~0.4$ for the lower-fidelity models. Figure~\ref{fig:ex3.2} shows a fine mesh and several coarse meshes used for the high and low-fidelity models, respectively. The details of the high and low-fidelity models, including CPU times to implement finite element analysis, are reported in Table~\ref{table:model}, which also reports (normalized) model costs. The oracle correlations between the QoI ($Y$) of the high-fidelity model and its low-fidelity QoIs $X_1$, $X_2$, $X_3$, $X_4$ are $0.96$, $0.93$, $0.87$, and $0.74$, respectively. 
\begin{table}[htbp]
\begin{center}
\caption{Fiber-reinforced matrix. Comparison of the high-fidelity and four different low-fidelity finite element models to compute the QoI.}
\small
\begin{tabular}{ccccc}
\toprule 
Model type & Tolerance ($\delta$) & DOFs & CPU time (s)$^\mathrm{a}$ & \multirow{1}{*}{Normalized cost$^\mathrm{b}$}\tabularnewline
\midrule
High-fidelity, $Y$ & $5\times10^{-3}$ & $\num[group-separator={,}]{53766}$ & $250.51$ & $108.9$\tabularnewline
Low-fidelity 1, $X_1$ & $5\times10^{-2}$ & $\num[group-separator={,}]{13416}$ & $20.95$ & $9.1$\tabularnewline
Low-fidelity 2, $X_2$ & $0.2$ & $975$ & $2.97$ & $1.3$\tabularnewline
Low-fidelity 3, $X_3$ & $0.2$ & $324$ & $2.50$ & $1.1$\tabularnewline
Low-fidelity 4, $X_4$ & $0.4$ & $195$ & $2.30$ & $1$\tabularnewline
\bottomrule
\end{tabular}
\begin{tablenotes}
\scriptsize\smallskip
\item{a.} The CPU time is averaged over $5$ trials.
\item{b.} The cost is normalized so that sampling $X_4$ has unit cost.
\end{tablenotes}
\label{table:model}
\end{center}
\end{table} 
\begin{figure}[h]
\centering
\begin{subfigure}{0.4\textwidth}{\includegraphics[width=\linewidth]{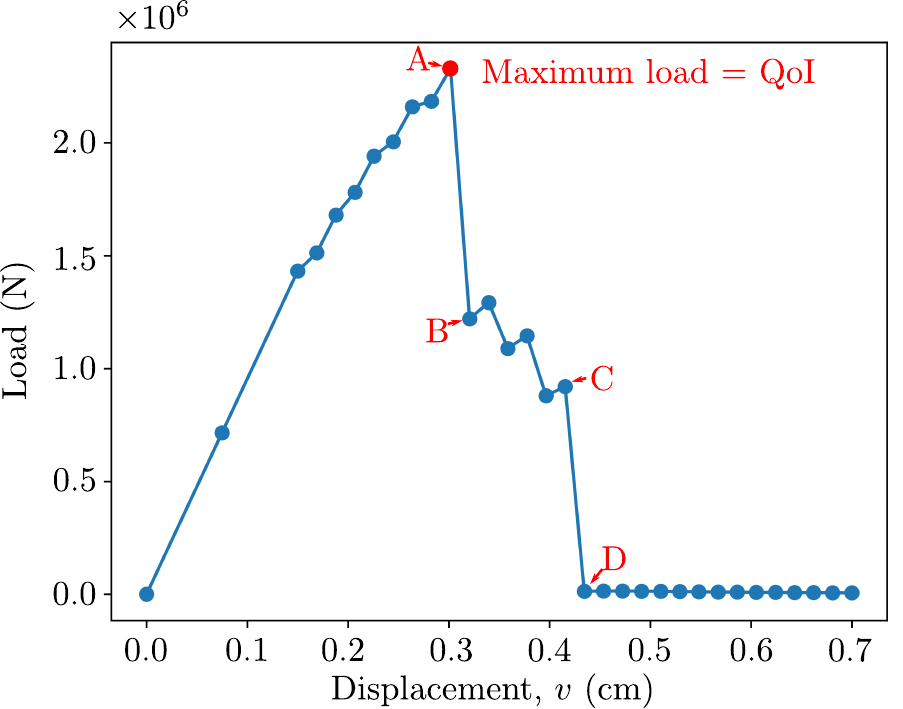}}\caption{Load-displacement curve}\label{brittle-QoI:(a)}
\end{subfigure}
\hspace{0.5cm}
\begin{subfigure}{0.4\textwidth}{\includegraphics[width=\linewidth]{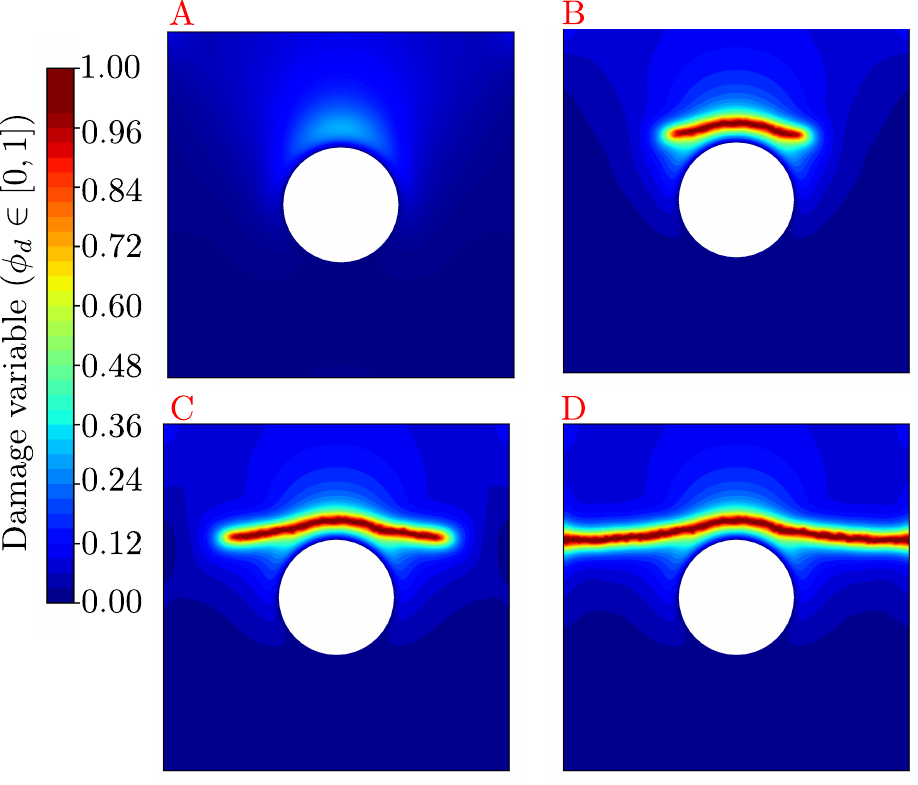}}\caption{Damage variable contour}\label{brittle-QoI:(b)}
\end{subfigure}
\caption{Fiber-reinforced matrix. Finite element analysis results: (a) The ultimate tensile load in the load-displacement curve is recorded as the QoI. (b) the damage variable contour shows the degree of damage ($0<\phi_d \leq 1$) that occurred in regimes `A'--`D' of the load-displacement curve, indicating that brittle fracture occurred at the top of circular hole advances in the regime `A' to `B' before a complete fracture occurs in regime `D'. }
\label{fig:ex3.3}
\end{figure} 

We measure the maximum tensile load as a QoI from the load-displacement curve. Figure~\ref{brittle-QoI:(a)} presents the relationship between the resulting load and the imposed displacement on the top of the fiber-reinforced matrix. As the applied displacement $\bar{v}$ at the top increases from $0$ cm to $7.5\times10^{-2}$~cm, the resulting load exhibits an almost linear increase until the structure begins to sustain damage. Upon reaching a peak load, the rate of change of the resulting load over displacement significantly decreases. This behavior is observed from regimes `A' to `B' in Figure~\ref{brittle-QoI:(a)}. The regimes occur due to the partial fracturing of the matrix, as indicated by a damage variable of value $\phi_d = 1$ in Figure~\ref{brittle-QoI:(b)}. There is another substantial drop in load from regimes `C' to `D', presenting complete fracture throughout the entire domain of the matrix. The maximum tensile load at `A' is represented as a scalar; thus, for the high-fidelity QoI we have $d=1$, while for the low-fidelity QoIs we have $d_1=d_2=d_3=d_4=1$.

\subsubsection{Results for CDF, mean, standard deviation, and CVaR estimation}
\begin{table*}[htbp]
\caption{
Fiber-reinforced matrix. Comparison of the accuracy of cvMDL{*}-sorted and ECDF in estimating CDF, mean, standard deviation, and $\mathrm{CVaR}$ at $\beta=0.99$ for the QoI (the ultimate tensile load).  We present the mean error of these estimates relative to oracle estimates obtained from $\num[group-separator={,}]{5500}$ i.i.d. high-fidelity samples ($B=\num[group-separator={,}]{577170}$) over 100 trials.}
\begin{center}
\vspace{0.1in}
\small
\begin{tabular}{cccccccc}
\toprule 
\multirow{1}{*}{Method} &  & CDF (error)$^{\mathrm{a}}$ &  & Mean (-)$^{\mathrm{b}}$ & Standard deviation (-)$^{\mathrm{b}}$ &  & $\mathrm{CVaR_{0.99}}$ (-)$^{\mathrm{b}}$\tabularnewline
\midrule 
\multicolumn{8}{c}{\textbf{Budget $B=\num[group-separator={,}]{20000}$}}\tabularnewline
cvMDL{*}-sorted &  & $2.192\times10^{-4}$ &  & $1.559\times10^{-3}$ & $2.291\times10^{-2}$ &  & $6.698\times10^{-3}$\tabularnewline
ECDF &  & $3.916\times10^{-4}$ &  & $2.991\times10^{-3}$ & $3.837\times10^{-2}$ &  & $9.435\times10^{-3}$\tabularnewline
\midrule
\multicolumn{8}{c}{\textbf{Budget $B=\num[group-separator={,}]{35000}$}}\tabularnewline
cvMDL{*}-sorted &  & $1.194\times10^{-4}$ &  & $1.168\times10^{-3}$ & $1.575\times10^{-2}$ &  & $5.531\times10^{-3}$\tabularnewline
ECDF &  & $2.269\times10^{-4}$ &  & $2.311\times10^{-3}$ & $2.955\times10^{-2}$ &  & $7.577\times10^{-3}$\tabularnewline
\midrule
\multicolumn{8}{c}{\textbf{Budget $B=\num[group-separator={,}]{50000}$}}\tabularnewline
cvMDL{*}-sorted &  & $8.441\times10^{-5}$ &  & $8.957\times10^{-4}$ & $1.340\times10^{-2}$ &  & $4.846\times10^{-3}$\tabularnewline
ECDF &  & $1.627\times10^{-4}$ &  & $1.983\times10^{-3}$ & $2.670\times10^{-2}$ &  & $7.102\times10^{-3}$\tabularnewline
\bottomrule
\end{tabular}
\label{brittle:result}
\end{center}
\vspace{0.1in}
\begin{tablenotes}
\scriptsize\smallskip 
\item{a.} We determine the mean $\omega(x)$-weighted $L^2$ error between $F_Y$ and the estimated CDFs given by cvMDL{*}-sorted and ECDF. The mean $\omega(x)$-weighted $L^2$ errors are averaged over independent $100$~trials.   
\item{b.} We use the mean relative error of the estimates in the comparison of the oracle estimates over independent $100$ trials.
\end{tablenotes}
\end{table*}
The high-fidelity simulations are costly enough here that we must approximate the oracle solution with limited samples: 6000 high-fidelity simulations are generated, and we randomly select 5500 to estimate a quantity. We generate an ensemble of 100 such instances and use the average as the oracle. For the multi-fidelity procedure, 6000 joint high- and low-fidelity samples are used as the pool from which model samples are drawn. We investigate three budget values as reported in \Cref{brittle:result}. Over the corresponding 100 trials, cvMDL*-sorted predominantly selects the model subset $\mathcal{S}=\{2,4\}$ (selected 95, 96, and 98 times for the 3 budget values, respectively) and less frequently selects the model subset $\mathcal{S}=\{4\}$ (selected 5, 4, and 2 times, respectively) from the model set $\{1, 2, 3, 4\}$. This process yields averaged optimal exploration sample numbers $m^*_{\mathcal{S}}=140$, $245$, and $350$ for each respective budget $B$.

In Table~\ref{brittle:result}, cvMDL{*}-sorted surpasses ECDF in terms of mean errors for CDF, mean, standard deviation, and CVaR at $a=0.99$ for the QoI. The second column of Table~\ref{brittle:result} reports the mean-weighted $L^2$ error for $F_Y$ over $Y\in[1.5\times 10^6,3\times10^6]$ (N). The last four columns of that table show that the proposed cvMDL approach yields nearly twice the accuracy compared to the ECDF method, and this increased accuracy extends to the estimated statistics and risk metric. 
\begin{figure}[htbp]
\centering 
\begin{subfigure}{0.325\textwidth}{\includegraphics[width=\linewidth]{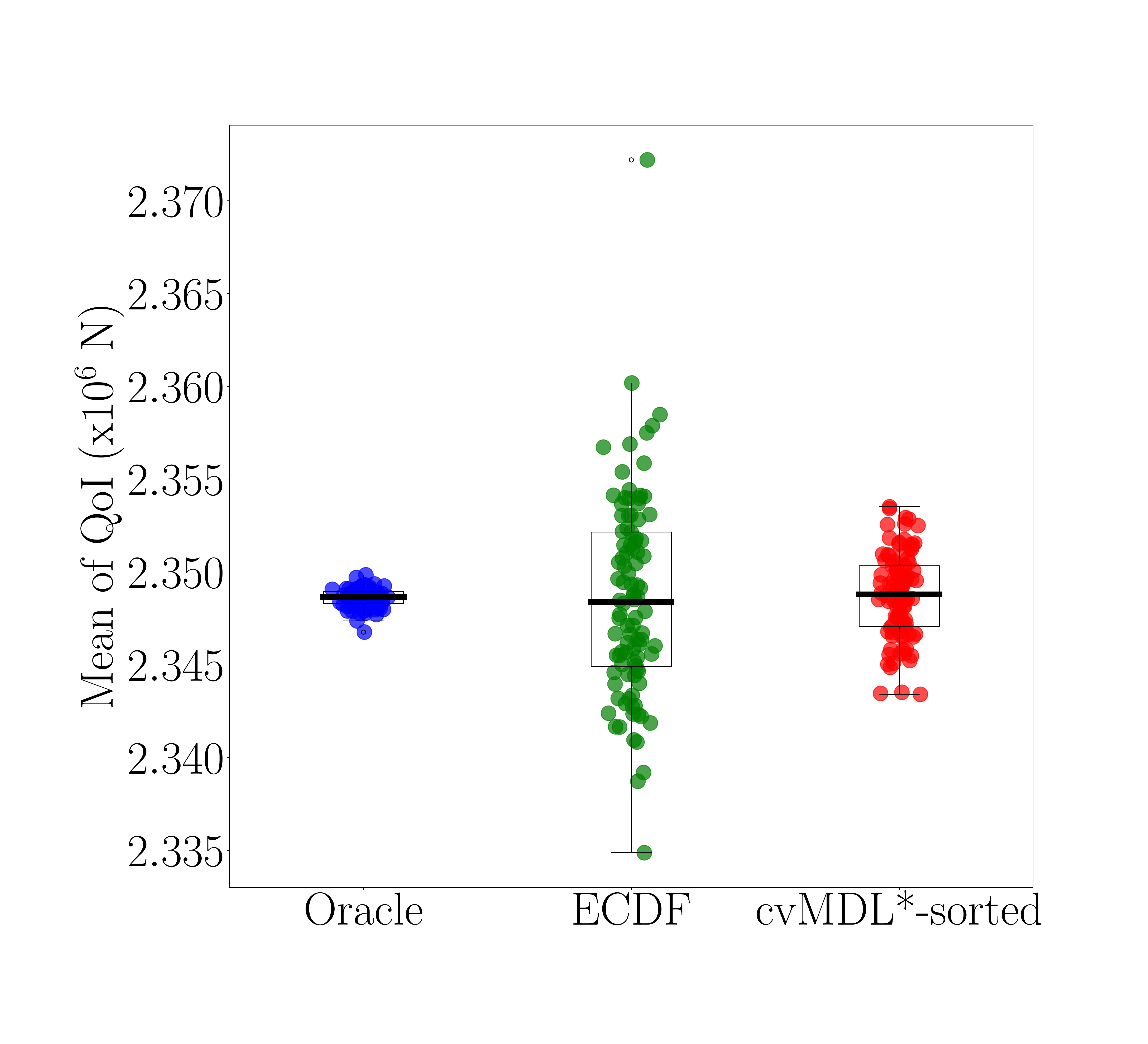}}\caption{Mean}\label{brittle:(a)}
\end{subfigure} 
\begin{subfigure}{0.325\textwidth}{\includegraphics[width=\linewidth]{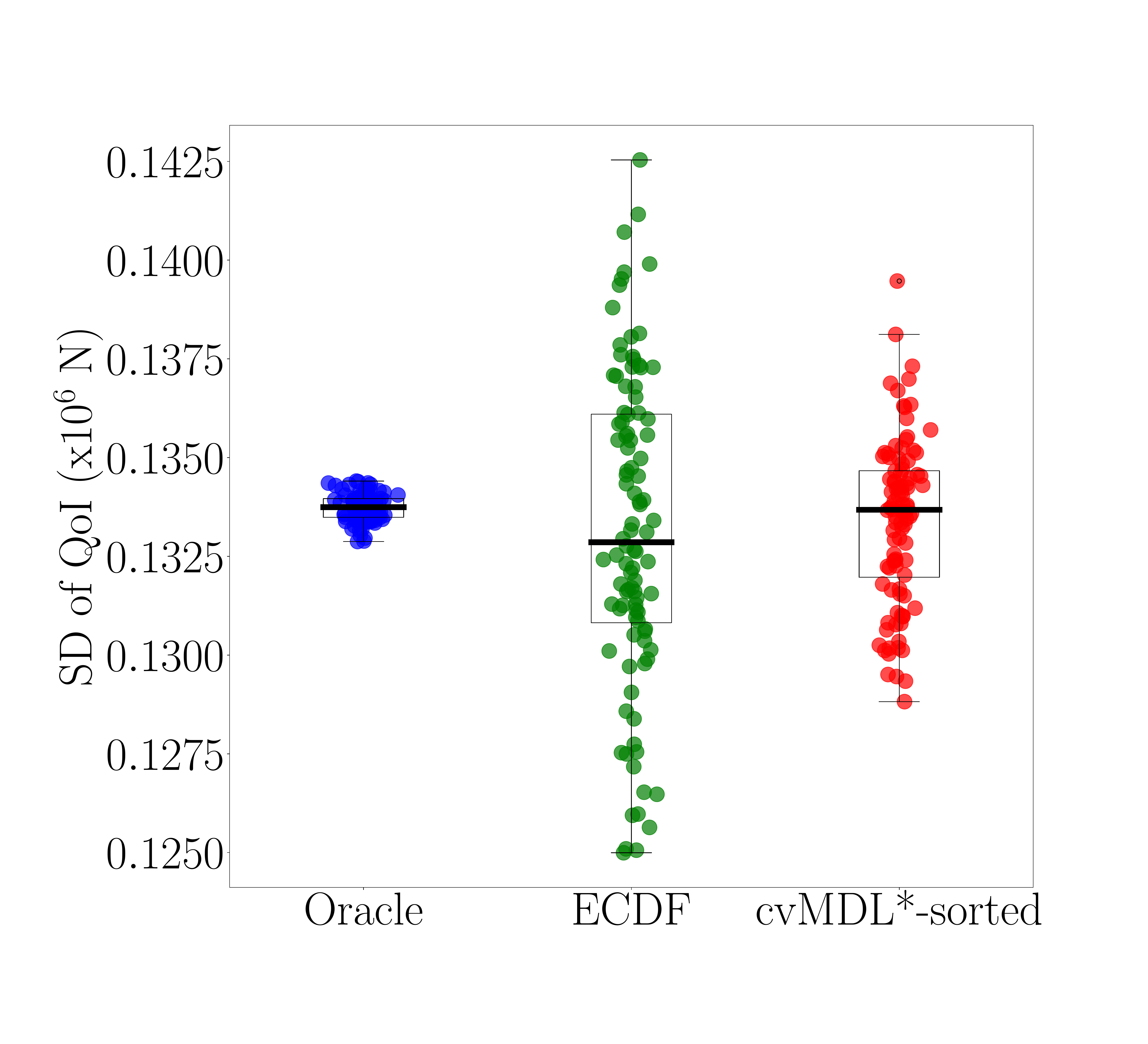}}\caption{Standard deviation}\label{brittle:(b)}
\end{subfigure} 
\begin{subfigure}{0.325\textwidth}
{\includegraphics[width=\linewidth]{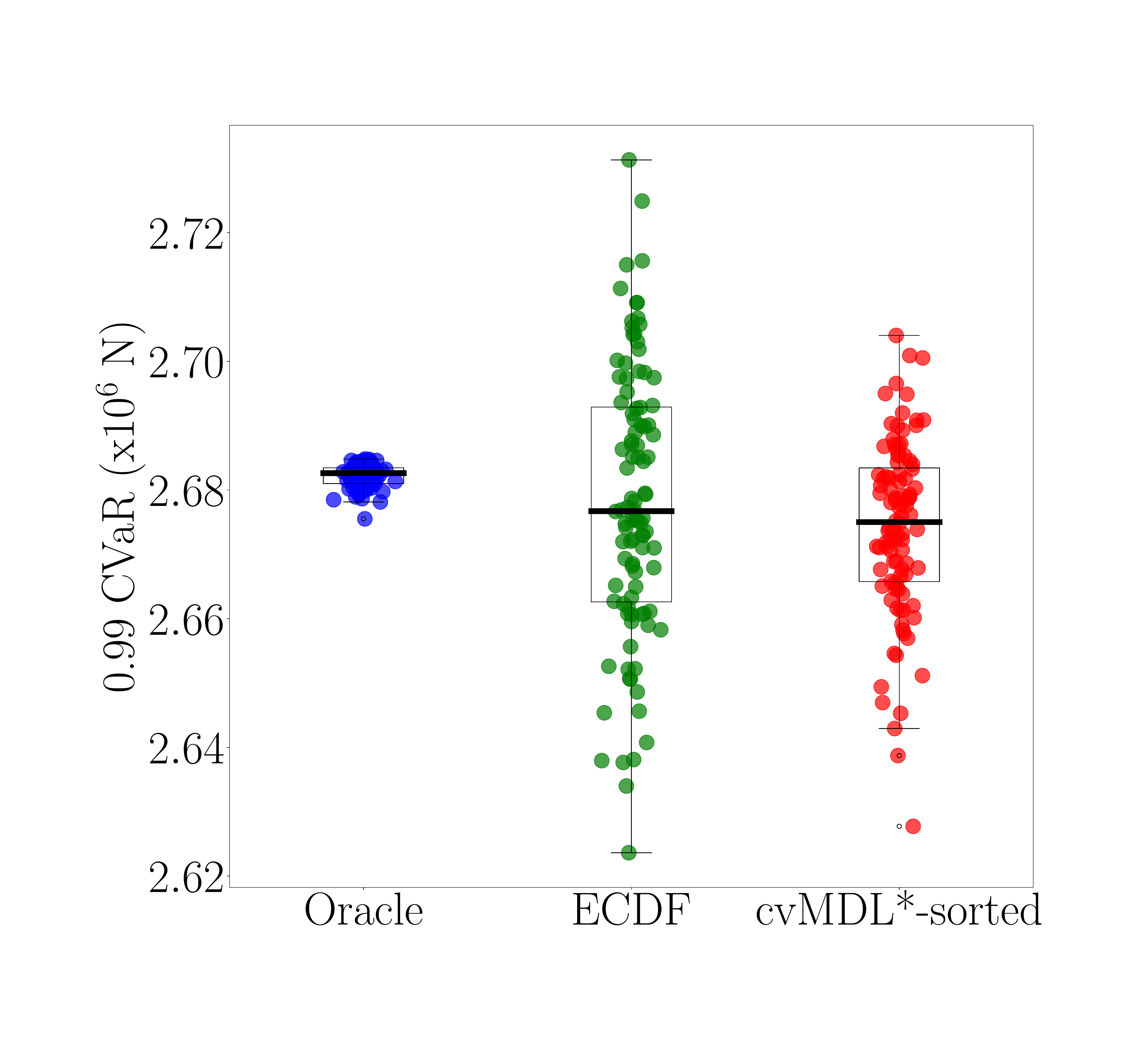}}\caption{$\mathrm{CVaR}_{0.99}$}\label{brittle:(c)}
\end{subfigure} 
\caption{Fiber-reinforced matrix. (a): Boxplots of the mean of QoI, computed by ECDF and cvMDL*-sorted when $B=\num[group-separator={,}]{50000}$ with 100 experiments. (b): Same boxplots for the standard deviation of QoI. (c): Same boxplots for CVaR (a=0.99).}\label{brittle_results}
\end{figure} 

Figures~\ref{brittle:(a)}--\ref{brittle:(c)} present the results of the statistical mean, standard deviation, and CVaR at $a=0.99$ via boxplots. The cvMDL*-sorted method achieves higher accuracy than the ECDF by a significant margin when compared to the oracle results. The box plots demonstrate that the statistical estimates by the cvMDL{*}-sorted method exhibit a lower spread than those of the ECDF, showing that cvMDL*-sorted estimates have smaller variance for this example.
\section{Conclusions}\label{sec:conc}
We developed a versatile framework for efficiently estimating the CDFs of QoI subject to a budget constraint.  
To implement this framework, we constructed a set of binary control variables based on linear surrogates and used them in an adaptive meta algorithm (cvMDL) that estimates the CDFs.
We established both uniform consistency and trade-off optimality for the corresponding algorithm as the budget tends to infinity. 

Although the proposed framework is built upon an existing bandit-learning paradigm, our treatment of exploration and exploitation distinguishes itself from the previous works.
In particular, the new approach employed in our framework leads to innovative estimators that dramatically relaxes the reliance on underlying model assumptions. Furthermore, the approach allows for the treatment of different types of QoIs, both vector- and scalar-valued.
To the best of our knowledge, our framework provides the first robust multifidelity CDF estimator under a budget constraint that can deal with both heterogeneous model sets and multi-valued QoIs at the same time, meanwhile requiring no a priori cross-model statistics.

\begin{appendix}

\section{Proofs of the main results}

\subsection{Proof of \Cref{thm:sort}}\label{ssec:sorting-proof}
Sort the entries of $\bm T$ in increasing order: $z^{(1)}<\cdots <z^{(M)}$, where $M$ is the product of each dimension $M_i$ of $\bm T$: $M = M_1\cdots M_d$. 
At the beginning of the algorithm, the index of $z^{(1)}$ is strictly decreasing in each direction. As a result, $z^{(1)}$ arrives at the entry of $\bm T$ with index $(1, \ldots, 1)$ after a finite number of iterations, and after that, it remains unchanged in the subsequent iteration.
In fact, for every $s< M$, assuming $z^{(1)}, \ldots, z^{(s)}$ have reached their final positions after which no change occurs, the index of $z^{(s+1)}$ is decreasing in each direction if the algorithm has not converged yet. 
The result follows by noting that $M$ is finite, and a stationary point must possess the desired monotonicity. 

\subsection{Proof of \Cref{lemma:parameter-consistency}}\label{sec:lemma-proof}
  This section contains the proofs of statements \ref{lemma:item:0} through \ref{lemma:item:5} in \Cref{lemma:parameter-consistency}. The proof of statement \ref{lemma:item:1}, the asymptotic consistency of $\widehat{\bm B}_{\S^+}$, is a direct result of the strong law of large numbers, so we omit this proof.

\subsubsection{Proof of Statement \ref{lemma:item:0}}\label{appx:1}

We only prove \eqref{0022} as the proof for \eqref{0021} is similar. 
Denote by $\bm e_i$ the $i$th unit vector in $\R^d$, i.e., $\bm e_d^{(j)} = \delta_{ij}, j\in \{1:d\}$, where $\delta$ is the Kronecker notation, and $\bm e := \sum_{i\in \{1:d\}}\bm e_i$ is the all-ones vector in $\R^d$.  
For fixed $\x\in\R^d$, without loss of generality, we assume $F_{V(\bm A)\vee Y}(\x)\leq F_{V(\bm B_{\S^+})\vee Y}(\x)$, as the other case is similar by reversing $F_{V(\bm A)\vee Y}(\x)$ and $F_{V(\bm B_{\S^+})\vee Y}(\x)$. 
Meanwhile, 
\begin{align*}
&V(\bm B_{\S^+})\vee Y\leq \x\Rightarrow V(\bm A)\vee Y \leq \x +\Delta \x\Rightarrow V(\bm A)\vee Y -\|\Delta \x\|_\infty\bm e \leq \x,
\end{align*}
where 
\begin{align*}
\Delta \x = \sum_{i\in \{1:d\}}|X_{\S^+}^\top(\bm A^{(i)}-\bm B_{\S^+}^{(i)})|\bm e_i.
\end{align*} 
Hence, under Assumptions \ref{ass1} and \ref{ass4}, for $t>0$, 
\begin{align*}
|F_{V(\bm A)\vee Y}(\x) - F_{V(\bm B_{\S^+})\vee Y}(\x)| &= F_{V(\bm B_{\S^+})\vee Y}(\x) - F_{V(\bm A)\vee Y}(\x)\\
&\leq F_{V(\bm A)\vee Y-\|\Delta \x\|_\infty\bm e}(\x) - F_{V(\bm A)\vee Y}(\x)\\
&\leq F_{V(\bm A)\vee Y- t\bm e}(\x) - F_{V(\bm A)\vee Y}(\x) + \P\left(\|\Delta \x\|_\infty\geq t\right)\\
& = F_{V(\bm A)\vee Y}(\x + t\bm e) - F_{V(\bm A)\vee Y}(\x) + \P\left(\|\Delta \x\|_\infty\geq t\right)\\
&\lesssim C\sqrt{d} t + \sum_{i\in \{1:d\}}\P\left(\|X_{\S^+}^\top(\bm A^{(i)}-\bm B_{\S^+}^{(i)})\|_2\geq t\right)\\
&\stackrel{\eqref{ass}}{\lesssim} C\sqrt{d} t + \frac{1}{t^2}\sum_{i\in \{1:d\}}\|\bm A^{(i)}-\bm B_{\S^+}^{(i)}\|^2_2,
\end{align*}
where the penultimate inequality follows from the Lipschitz condition on $F_{V(\bm A)\vee Y}$ and a union bound, 
and the last inequality follows from Markov's inequality. 
Taking $t = \|\bm A - \bm B_{\S^+}\|_F^{2/3}$ yields the desired result.

\subsubsection{Proof of Statement \ref{lemma:item:2}}\label{appx:2}

We only prove the first statement; the second can be proved similarly.
Recall that
\begin{align*}
	&\widehat{F}_{\widehat{H}_\S}(\x)  = G(\widehat{\bm B}_{\S^+}; \x)& G(\bm\A; \x) := \frac{1}{m}\sum_{\ell\in \{1:m\}} \mathbf 1_{\{(X^\top_{\ex, \ell, \S^+}\bm\A)^\top\leq \x\}}\ \ \A\in\R^{(d_\S+1)\times d},
\end{align*}
where $X_{\ex, \ell, \S^+}$ denotes the $\ell$th exploration sample of $X_\S$ with intercept. 
It follows from the direct computation that
\begin{align*}
&\sup_{\x\in\R^d}|\widehat{F}_{\widehat{H}_\S}(\x) - {F}_{{H}_\S}(\x)|\\
\leq&\   \sup_{\x\in\R^d}\biggl| G(\widehat{\bm B}_{\S^+}; \x) - \mathbb{P}((X_{\S^+}^\top\widehat{\bm B}_{\S^+})^\top \leq \x) \biggr| +  \sup_{\x\in\R^d}\biggl|  \mathbb{P}((X_{\S^+}^\top\widehat{\bm B}_{\S^+})^\top \leq \x) - \mathbb{P}((X_{\S^+}^\top{\bm B}_{\S^+})^\top \leq \x) \biggr|\\
\leq&\ \sup_{\bm\A\in\R^{(d_\S+1)\times d}}\sup_{\x\in\R^d}  \biggl| G(\bm A; \x) - \mathbb{P}((X_{\S^+}^\top\bm A)^\top \leq \x) \biggr| +  \sup_{\x\in\R^d}\biggl|  \mathbb{P}((X_{\S^+}^\top\widehat{\bm B}_{\S^+})^\top \leq \x) - \mathbb{P}((X_{\S^+}^\top{\bm B}_{\S^+})^\top \leq \x) \biggr|\\
 =&\  \underbrace{\sup_{\bm\A\in\R^{(d_\S+1)\times d}}  \biggl| G(\bm A; \bm 0) - \E[G(\bm A; \bm 0)] \biggr|}_{\Lambda_{m,1}} +  \underbrace{\sup_{\x\in\R^d}\biggl|  \mathbb{P}((X_{\S^+}^\top\widehat{\bm B}_{\S^+})^\top \leq \x) - \mathbb{P}((X_{\S^+}^\top{\bm B}_{\S^+})^\top \leq \x) \biggr|}_{\Lambda_{m,2}}
\end{align*}
where $X_{\S^+}$ a general notation that is independent of $\widehat{\bm B}_{\S^+}$ and $\bm 0$ is the all-zeros vector.
Note that $\Lambda_{m,1}$ has no supremum over $\bm x$ since one is able to alter the intercept coefficients in $\bm A$ to yield different values of $\bm x\in\R^d$ without changing the coefficients of $X_{\S}$. 
In what follows, we show that both $ \Lambda_{m,1} $ and $ \Lambda_{m,2} $ converge to $ 0 $ a.s.

To bound $ \Lambda_{m,1} $, we appeal to the empirical process theory.
For any $\bm\A \in \R^{(d_\S+1)\times d}$, the indicator function $ \mathbf 1_{\{(X^\top_{\ell, \S^+}\bm\A)^\top\leq \bm 0\}} \leq 1.$ 
According to Massart concentration inequality \cite[Theorem 14.2]{buhlmann2011statistics}, we have for any $ t > 0 $ such that
\begin{equation*}
	\mathbb{P}(\Lambda_{m,1} > \mathbb{E}[\Lambda_{m,1}] + t) \leq \exp(-mt^2/8).
\end{equation*}
Taking $ t = 4\sqrt{\log m/m} $, 
\begin{equation}\label{concentration}
	\mathbb{P}\left(\Lambda_{m,1} > \mathbb{E}[\Lambda_{m,1}] + 4\sqrt{\frac{\log m}{m}}\right) \leq m^{-2}.
\end{equation}
Since $ \sum_{m = 1}^{\infty} m^{-2} < \infty, $ by the Borel-Cantelli lemma, we conclude that 
\begin{align}
\Lambda_{m,1} \leq \mathbb{E}[\Lambda_{m,1}] + 4\sqrt{\frac{\log m}{m}}\label{hutf}
\end{align}
for all sufficiently large $m$ a.s.
To bound $\mathbb{E}[\Lambda_{m,1}]$, note that the supremum in $\mathbb{E}[\Lambda_{m,1}]$ is taken over all indicator functions defined on $d$ intersections of hyperplanes in $\R^{d_\S}$ (the constant dimension is only a shift), which has a finite Vapnik--Chervonenkis (VC) dimension of order $d_{\S}d\log d$ \cite{blumer1989learnability}. 
According to \cite[Theorem 8.3.23]{Vershynin2018},  there exists a universal constant $C'$ such that $\mathbb{E}[\Lambda_{m,1}]\leq C'\sqrt{d_{\S}d\log d/m}$. 
This combined with \eqref{hutf} shows that $\Lambda_{m,1}\to 0$ a.s.

To bound $\Lambda_{m,2}$,  note that by Statement \ref{lemma:item:1} in \Cref{lemma:parameter-consistency}, 
a.s., for all sufficiently large $m$, $\|\widehat{\bm B}_{\S^+}-\bm B_{\S^+}\|_F< \e$ where $\e$ is the same as in \Cref{ass4}. 
Since $X_{\S^+}$ is independent of $\widehat{\bm B}_{\S^+}$, conditioning on $\|\widehat{\bm B}_{\S^+}-\bm B_{\S^+}\|_F< \e$ and applying 
Statement \ref{lemma:item:0} of \Cref{lemma:parameter-consistency},
$\Lambda_{m,2}\lesssim\|\widehat{\bm B}_{\S^+}-\bm B_{\S^+}\|^{2/3}_F$.
Now taking $m\to\infty$ shows $\Lambda_{m,2}\to 0$ a.s.

\subsubsection{Proof of Statement \ref{lemma:item:3}}\label{appx:3}

Note $\mathcal K_1(\x) + \mathcal K_2(\x) = \widehat{F}_Y(\x)(1-\widehat{F}_Y(\x))$, which is a consistent estimator for $F_Y(\x)(1-F_Y(\x))$ for all $\x\in\R^d$ a.s. as a result of the strong law of large numbers. 
Therefore, it suffices to prove the consistency for $\mathcal K_2(\x)$ only. 

Note $\mathcal K_2(\x)$ in \eqref{my1} can be rewritten as 
\begin{align}
\mathcal K_2(\x) &= \widehat{\rho}^2_\S(\x)\widehat{F}_{Y}(\x)(1 -\widehat{F}_{Y}(\x))= \begin{cases}
\frac{(\widehat{F}_{Y\vee\widehat{H}_\S}(\x) - \widehat{F}_{Y}(\x)\widehat{F}_{\widehat{H}_\S}(\x))^2}{\widehat{F}_{\widehat{H}_\S}(\x)(1- \widehat{F}_{\widehat{H}_\S}(\x))} & \x\in (\text{supp}(\widehat{F}_{\widehat{H}_\S}))^\circ\\
0 & \text{otherwise}
\end{cases}\label{daidai}
\end{align}
where 
\begin{align*}
\widehat{\rho}^2_\S(\x) = 
\begin{cases}
\frac{(\widehat{F}_{Y\vee\widehat{H}_\S}(\x) - \widehat{F}_{Y}(\x)\widehat{F}_{\widehat{H}_\S}(\x))^2}{\widehat{F}_{\widehat{H}_\S}(\x)(1- \widehat{F}_{\widehat{H}_\S}(\x))\widehat{F}_{Y}(\x)(1 -\widehat{F}_{Y}(\x))} & \x\in (\text{supp}(\widehat{F}_{Y}))^\circ\cap(\text{supp}(\widehat{F}_{\widehat{H}_\S}))^\circ\\
0 & \text{otherwise}
\end{cases}
\end{align*}
is the empirical estimator for $\rho_\S^2(\x)$. 
On the other hand,
\begin{align}\label{taitai}
\rho_\S^2(\x)F_Y(\x)(1-F_Y(\x)) = \begin{cases}
\frac{(F_{Y\vee H_\S}(\x) - F_{Y}(\x)F_{H_\S}(\x))^2}{F_{H_\S}(\x)(1- F_{H_\S}(\x))} & \x\in (\text{supp}(F_{H_\S}))^\circ\\
0 & \text{otherwise}
\end{cases}
\end{align}
Comparing \eqref{daidai} and \eqref{taitai}, the desired result follows from statement \ref{lemma:item:2} in \Cref{lemma:parameter-consistency}.

\subsubsection{Proof of statement \ref{lemma:item:4}}\label{appx:4}
	We prove the consistency of $\widehat{k}_2(\S)$; the consistency of $\widehat{k}_1(\S)$ can be proved similarly.
        By statement \ref{lemma:item:3} in \Cref{lemma:parameter-consistency}, 
        $\mathcal K_2(\x)$ converges to $\rho_\S^2(\x)F_Y(\x)(1-F_Y(\x))$ for all $\x\in\R^d$ as $m\to\infty$ a.s.  
     
     We first prove the first case where $d=1$ and $\|\omega\|_{L^\infty(\R)} = C<\infty$, and we change the notation $\x$ to the lowercase $x$. 
     Under the moment condition in Assumption \ref{ass1}, according to \cite[Theorem 2.13]{bobkov2019one},
    \begin{align*}
    &W_1\left(F_Y, \widehat{F}_Y\right) = \int_\R|\widehat{F}_Y(x)-F_Y(x)|\text{d}x\to 0& m\to\infty, 
    \end{align*} 
    where $W_1$ is the Wasserstein-1 metric. 
    Fix an arbitrary trajectory in the sample space such that $\mathcal K_2(x) \to\rho_\S^2(x)F_Y(x)(1-F_Y(x))$ and $\int_\R|\widehat{F}_Y(x)-F_Y(x)|\text{d}x\to 0$. 
    In the following, we treat $\mathcal K_2(x)$ as a deterministic sequence. 
    
    To show the consistency of $\widehat{k}_2(\S)$, it remains to justify the change of order of taking limit and integration:
    \begin{align*}
    \lim_{m\to\infty}\widehat{k}_2(\S) = \lim_{m\to\infty}c_\S\int_{\R} \omega(x)\mathcal{K}_2(x) \text{d}x &=  c_\S\int_{\R} \lim_{m\to\infty}\omega(x)\mathcal{K}_2(x) \text{d}x\\
    & = c_\S\int_{\R}\omega(x)\rho_\S^2(x)F_Y(x)(1-F_Y(x))\text{d}x = k_2(\S), 
    \end{align*}
    for which we appeal to the Vitali convergence theorem. 
    To apply the Vitali convergence theorem, we need to verify that the sequence $\omega(x)\mathcal K_2(x)$ is uniformly integrable and has absolutely continuous integrals. 
    To this end, recall the representation $\mathcal K_2(x)$ in \eqref{daidai}.
    Since the square of the empirical correlation estimator is bounded by $1$, a.s., 
    \begin{align*}
    &\omega(x)\mathcal K_2(x)\leq \omega(x)\widehat{F}_{Y}(x)(1 -\widehat{F}_{Y}(x))\leq \frac{C}{4}<C.
    \end{align*}
    The absolutely continuous integrals part follows immediately from the uniform boundedness.
    For uniform integrability, we first observe  
    \begin{align*}
    \int_{\R}|\omega(x)F_Y(x)(1-F_Y(x)) - \omega(x)\widehat{F}_{Y}(x)(1 -\widehat{F}_{Y}(x))| \text{d}x&\leq\int_{\R}\omega(x)|\widehat{F}_{Y}(x)-F_Y(x)| \text{d}x\\
    &\leq C\int_{\R}|\widehat{F}_{Y}(x)-F_Y(x)|  \text{d}x\to 0.
    \end{align*}
    Thus, 
    \begin{align*}
    \int_{|x|>M}\omega(x)\mathcal K_2(x) \text{d}x &\leq \int_{|x|>M}\omega(x)\widehat{F}_{Y}(x)(1 -\widehat{F}_{Y}(x)) \text{d}x\\
    &\leq \int_{|x|>M}\omega(x) F_{Y}(x)(1 -F_{Y}(x)) \text{d}x + C\int_{|x|>M}|\widehat{F}_{Y}(x) - F_Y(x)| \text{d}x\\
    &\lesssim \int_{|x|>M}\frac{C}{x^2} \text{d}x + C\int_{\R}|\widehat{F}_{Y}(x) - F_Y(x)| \text{d}x,
    \end{align*}
    where the last step follows from Assumption \ref{ass1} and Chebyshev's inequality. 
    For every $\e>0$, we can choose $m$ and $M$ sufficiently large so the right-hand side is less than $\e$. 
    The uniform integrability follows by enlarging $M$ to accommodate the first $m$ terms.
    
    The proof for $(b)$ is similar. It suffices to verify the change of order for the sequence $\mathcal \omega(\x)\mathcal K_2(\x)$. Since $\mathcal \omega(\x)\mathcal K_2(\x)\leq \omega(\x)$ and the latter is integrable and independent of $m$, the dominated convergence does the rest. 
    
\subsubsection{Proof of statement \ref{lemma:item:5}}\label{appx:5}

For $\x\in (\text{supp}(F_{H_\S}))^\circ$, it is easy to show via a contradiction argument that $\x\in\text{supp}(\widehat{F}_{\widehat{H}_\S})$ for all sufficiently large $m$ a.s. 
By statement \ref{lemma:item:2} in \Cref{lemma:parameter-consistency},  
	$\widehat{F}_{Y\vee\widehat{H}_\S}$ and $ \widehat{F}_{\widehat{H}_\S} $ are consistent estimators. Meanwhile, $\widehat{F}_{Y}(\x)$ is a consistent estimator for ${F}_{Y}(\x)$ due to the strong law of large numbers. 
 Therefore, we obtain
\begin{equation*}
\widehat{\alpha}(\x) =  \frac{\widehat{F}_{Y\vee\widehat{H}_\S}(\x) - \widehat{F}_{Y}(\x)\widehat{F}_{\widehat{H}_\S}(\x)}{\widehat{F}_{\widehat{H}_\S}(\x)(1-\widehat{F}_{\widehat{H}_\S}(\x))} \to {\alpha}(\x) = \frac{{F}_{Y\vee{H}_\S}(\x) - {F}_{Y}(\x){F}_{{H}_\S}(\x)}{{F}_{{H}_\S}(\x)(1-{F}_{{\H}_\S}(\x))}
\end{equation*}
as $ m \to \infty $ almost surely.

\subsection{Proof of \Cref{unifexploi}}\label{appx:6}

	Recall in \eqref{fug} that
	\begin{align*}
	\widetilde{F}_\S(\x) &= 	\widehat{F}_Y(\x) -  \frac{1}{m}\sum_{\ell\in \{1:m\}}\left(\widehat{\alpha}(\x)\widehat{h}_\S(X_{\ex, \ell, \S} ;\x)- \frac{1}{N_\S}\sum_{j\in \{1:N_\S\}}\widehat{\alpha}(\x)\widehat{h}_\S(X_{\ext, \ell, \S}; \x)\right)\\
	& = \widehat{F}_Y(\x) - \widehat{\alpha}(\x)\widehat{F}_{\widehat{H}_\S}(\x) + \widehat{\alpha}(\x)\left(\frac{1}{N_\S}\sum_{j\in \{1:N_\S\}}\mathbf 1_{\{(X^\top_{\ext, j, \S^+}\widehat{\bm B}_{\S^+})^\top\leq \x\}}\right).
	\end{align*}
	Thus, 
	\begin{align*}
	\sup_{\x\in\R^d}|\widetilde{F}_\S(\x) - F_Y(\x)|&\leq \sup_{\x\in\R^d}|\widehat{F}_Y(\x)-F_Y(\x)| + \sup_{\x\in\R^d}|\widehat{\alpha}(\x)(\widehat{F}_{\widehat{H}_\S}(\x)-F_{H_\S}(\x))|\\
	&\ \ \ \ + \sup_{\x\in\R^d}\left|\widehat{\alpha}(\x)\left(\frac{1}{N_\S}\sum_{j\in \{1:N_\S\}}\mathbf 1_{\{(X^\top_{\ext,j,\S^+}\widehat{\bm B}_{\S^+})^\top\leq \x\}}-F_{H_\S}(\x)\right)\right|\\
	&\stackrel{\eqref{bd2}}{\leq}\underbrace{\sup_{\x\in\R^d}|\widehat{F}_Y(\x)-F_Y(\x)|}_{(i)} + \underbrace{\sup_{\x\in\R^d}|\widehat{F}_{\widehat{H}_\S}(\x)-F_{H_\S}(\x)|}_{(ii)}\\
	&\ \ \ \ + \underbrace{\sup_{\x\in\R^d}\left|\left(\frac{1}{N_\S}\sum_{j\in \{1:N_\S\}}\mathbf 1_{\{(X^\top_{\ext,j,\S^+}\widehat{\bm B}_{\S^+})^\top\leq \x\}}-F_{H_\S}(\x)\right)\right|}_{(iii)}.
	\end{align*}
	Note $(i)$ converges to $0$ as $m\to\infty$ due to the multivariate Glivenko-Cantelli theorem.
        $(ii)$ converges to 0 as $m\to\infty$ due to statement \ref{lemma:item:2} in \Cref{lemma:parameter-consistency}.
        A similar argument as in the proof of statement \ref{lemma:item:2} of \Cref{lemma:parameter-consistency} 
        can be used to prove that $(iii)$ converges to $0$ as $N_\S\to\infty$, which is not repeated here.

\subsection{Proof of Theorem \ref{main}}\label{aa1}

To reduce notational confusion with $m$, we use $t$ to denote the number of exploration samples. 
The exploration rate $m$ grows nonlinearly with respect to an index that counts the iterations of the exploration loop in \Cref{alg2-detailed}.
We let $q$ denote the loop iteration index, and $t_q$ the corresponding exploration rate, i.e., $t_1 = n+2$. 
Let $q(B)$ be the total number of exploration iteration steps in \Cref{alg2-detailed}, which is random. 
It follows from the definition that $t_{q(B)} = m(B)$ and 
\begin{align}
&n+1\leq t_q\leq t_{q+1}\leq 2t_q& 1\leq q< q(B).\label{myq}
\end{align}

We first show that $m(B)$ diverges as $B\to\infty$ a.s. 
According to statement \ref{lemma:item:4} in \Cref{lemma:parameter-consistency}, 
$\widehat{k}_1(\S)\to k_1(\S), \widehat{k}_2(\S)\to k_2(\S)$ for $\S\subseteq \{1:n\}$ a.s. 
As a result, for almost every realization $\omega\in\Omega$, where $\Omega$ denotes the product space of exploration samples, there exists an $0<L(\omega), L'(\omega)<\infty$,
\begin{align*}
&\sup_{t>n+1}\max_{\S\subseteq \{1:n\}}\widehat{k}_1(\S; \omega)<L(\omega)<\infty& \inf_{t>n+1}\min_{\S\subseteq \{1:n\}}\widehat{k}_2(\S; \omega)>L'(\omega)>0
\end{align*} 
The exploration stopping criterion of cvMDL in \Cref{alg2-detailed} requires that 
\begin{align*}
&m(B; \omega)\geq \widehat{m}^*_{\S(B; \omega)} \geq \frac{B}{c_{\ex}+\sqrt{\frac{c_{\ex}L(\omega)}{L'(\omega)}}}\to\infty&B\to\infty.
\end{align*} 
Thus, 
\begin{align}
&\lim_{B\to\infty}m(B; \omega)=\infty. 
\label{ksl1}
\end{align}

We now work with a fixed realization $\omega$ along which $m(B; \omega)\to\infty$ as $B\to\infty$, and $\widehat{k}_1(\S), \widehat{k}_2(\S)$ converge to the true parameters as $t\to\infty$. 
We prove that both \eqref{rr1} and \eqref{rr2} hold for such an $\omega$. 
Fix $\delta<1/2$ sufficiently small. 
Since $\S^*$ is assumed unique, a continuity argument implies that there exists a sufficiently large $T(\delta; \omega)$, such that for all $t\geq T(\delta; \omega)$, 
\begin{align}
\max_{(1-\delta)m^*_{\S^*}\leq m\leq (1+\delta)m^*_{\S^*}}{\widehat{L}}_{\S^*}(m; t)&< \min_{\S\subseteq \{1:n\}, \S\neq \S^*}\widehat{L}^*_\S(t).\label{112}\\
1-\delta&\leq \frac{\widehat{m}^*_\S(t; \omega)}{m^*_\S}\leq 1+\delta&\forall \S\subseteq \{1:n\}\label{113},
\end{align}
where $\widehat{L}_{\S^*}(\cdot;  t)$ is the estimated loss function for $\S^*$ using $t$ exploration samples, and $\widehat{L}_\S^*(t)$ is the  estimated $L_\S^*$ in \eqref{optm} using $t$ exploration samples.

Since $m^*_\S$ scales linearly in $B$ and $m(B;\omega)$ diverges as $B\to\infty$, there exists a sufficiently large $B(\delta;\omega)$ such that for $B>B(\delta;\omega)$, 
\begin{align}
\min_{\S\subseteq \{1:n\}}m^*_\S&> 4T(\delta;\omega)\label{2345}\\
t_{q(B)} = m(B;\omega) &> 4T(\delta;\omega).\label{3456}
\end{align}
  Consider $q'< q(B)$ that satisfies $t_{q'-1}< T(\delta;\omega) \leq t_{q'}$. Such a $q'$ always exists due to \eqref{3456}, and satisfies 
\begin{align*}
t_{q'}\stackrel{\eqref{myq}}{\leq} 2t_{q'-1}<2T(\delta;\omega)\stackrel{\eqref{2345}}{\leq}\frac{1}{2}\min_{\S\subseteq \{1:n\}}m^*_\S\stackrel{\eqref{113}, \delta<1/2}{\leq}\widehat{m}^*_\S(t_{q'}; \omega).
\end{align*} 
This inequality tells us that in the $q'$-th loop iteration, for all $\S\subseteq \{1:n\}$, the corresponding estimated optimal exploration rate is larger than the current exploration rate. 
In this case, 
\begin{align*}
&\widehat{L}_\S(t_{q'}\vee \widehat{m}^*_\S(t_{q'}; \omega); t_{q'}) = \widehat{L}_\S(\widehat{m}^*_\S(t_{q'}; \omega); t_{q'}) = \widehat{L}_\S^*(t_{q'})&\forall \S\subseteq \{1:n\}.
\end{align*}
This, along with \eqref{112} and \eqref{113}, tells us that $\S^*$ is the estimated optimal model in the current step, and more exploration is needed.  

To see what $t_{q'+1}$ should be, we consider two separate cases.
If $2t_{q'}\leq \widehat{m}^*_{\S^*}(t_{q'}; \omega)$, then
\begin{align*}
T(\delta;\omega)<t_{q'+1} = 2t_{q'}\leq \widehat{m}^*_{\S^*}(t_{q'}; \omega)\leq (1+\delta)m^*_{\S^*},
\end{align*}
which implies 
\begin{align}
(1-\delta)m^*_{\S^*}\stackrel{\eqref{113}}{\leq} t_{q'+1}\vee \widehat{m}^*_{\S^*}(t_{q'+1}; \omega)\leq (1+\delta)m^*_{\S^*}.\label{myku}
\end{align}
If $t_{q'}\leq \widehat{m}^*_{\S^*}(t_{q'}; \omega)<2t_{q'}$, then
\begin{align*}
t_{q'+1} = \left\lceil\frac{t_{q'} + \widehat{m}^*_{\S^*}(t_{q'}; \omega)}{2}\right\rceil\leq \widehat{m}^*_{\S^*}(t_{q'}; \omega)\stackrel{\eqref{113}}{\leq} (1+\delta)m^*_{\S^*},
\end{align*}
which also implies \eqref{myku}. 
But \eqref{myku} combined with \eqref{112} and \eqref{113} implies that $\S^*$ is again the estimated optimal model in the $(q'+1)$-th loop iteration. 
Applying the above argument inductively proves $\S(B) = \S^*$, i.e. \eqref{rr2}.
Note \eqref{myku} holds true until the algorithm terminates, which combined with the termination criteria $t_{q(B)}\geq \widehat{m}^*_{\S^*}(t_{q(B)}; \omega)\geq (1-\delta)m^*_{\S^*}$ implies 
\begin{align*}
1-\delta\leq\frac{m(B; \omega)}{m^*_{\S^*}} = \frac{t_{q(B)}}{m^*_{\S^*}}\leq 1+\delta.
\end{align*}
\eqref{rr1} now follows by noting that $\delta$ can be set arbitrarily small. 

Finally, let $\widetilde{F}'(\x; B)$ be chosen as in \eqref{fug} with $\S =  \S^*$, $m = m^*_{\S^*}$ and $N_\S = (B-c_{\ex}m^*_{\S^*})/c_{\S^*}$. 
Note both $m, N_\S$ are deterministic and diverge as $B\to\infty$. 
By the triangle inequality,
\begin{align*}
\sup_{\x\in\R^d}|\widetilde{F}(\x; B)-F_Y(\x)|\leq\sup_{\x\in\R^d}|\widetilde{F}(\x; B)-\widetilde{F}'(\x; B)| + \sup_{\x\in\R^d}|\widetilde{F}'(\x; B)-F_Y(\x)|.
\end{align*}
As $B\to\infty$, 
the first term on the right-hand side converges to $0$ due to \eqref{rr1} and \eqref{rr2} in Theorem \ref{main}, 
and the second term on the right-hand side converges to $0$ due to Theorem \ref{unifexploi}. This proves \eqref{rr3}.

\section{Additional numerical results}

\subsection{Additional results for geometric Brownian motion in \Cref{sec:num2}}\label{sec:num2-details}

We present two figures that provide experimental results to supplement those presented in \Cref{sec:num2}. A plot of the oracle CDF is visualized in Figure \ref{oracle} (left, middle). The oracle model loss and exploration sample count are in \Cref{oracle} (right). 
\Cref{sde-heat} shows an instance of a heatmap of the absolute estimation errors of ECDF, cvMDL, and cvMDL-sorted when $B = 10^6$, providing supporting evidence that cvMDL is more accurate than ECDF on $\mathcal T$.

\begin{figure}[htbp]
  \begin{minipage}[b][][b]{0.33\textwidth}
\begin{subfigure}{1\textwidth}{\includegraphics[width=\textwidth, trim={2cm 2.5cm 1cm 3cm},clip]{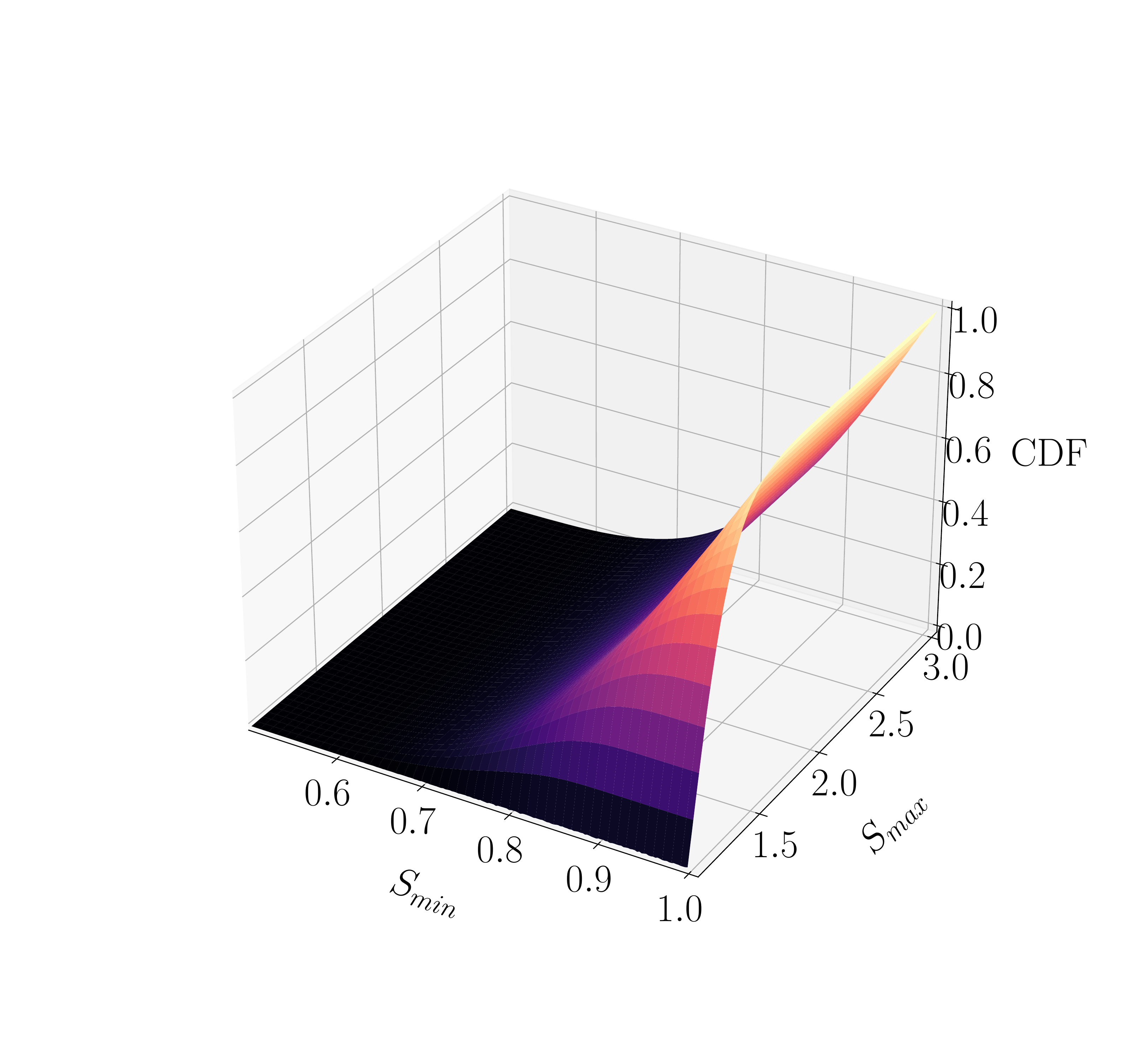}}\caption{}
\end{subfigure} 
  \end{minipage}
  \begin{minipage}[b][][b]{0.31\textwidth}
  \begin{subfigure}{1\textwidth}{\includegraphics[width=\textwidth, trim={0.2cm 0.2cm 1cm 1cm},clip]{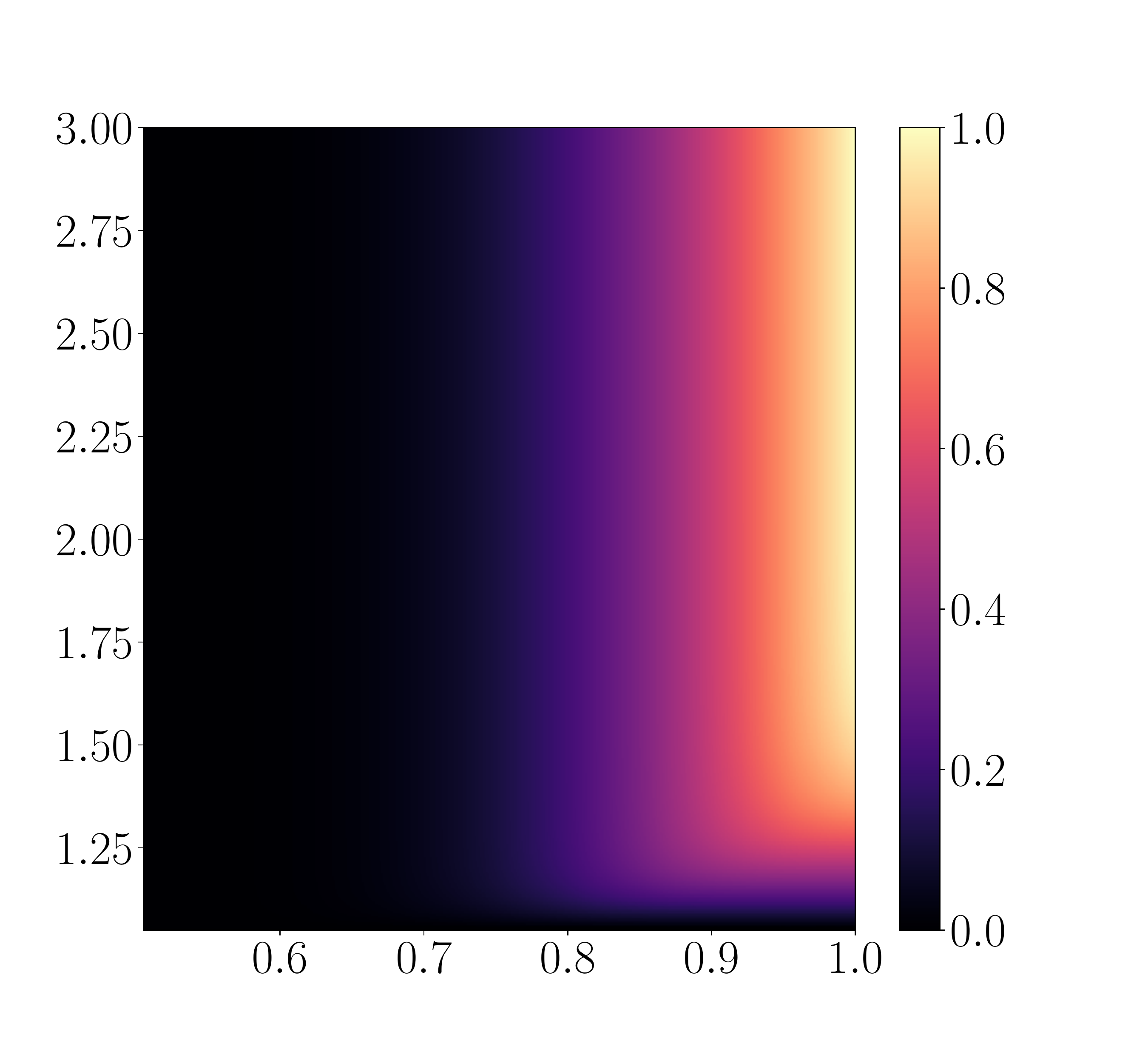}}\caption{}
  \end{subfigure}
  \end{minipage}
  \begin{minipage}[b][][b]{0.32\textwidth}
  \begin{subfigure}{1\linewidth}
\begin{tabular}{rrr}
      Model $\S$ & $\gamma_\S$ & $m^\ast_\S$ \\\hline
      $\textbf{\{1\}}$ & \textbf{11.3} & \textbf{613} \\
      $\{2\}$ & 13.7 & 23.2 \\
      $\{3\}$ & 23.2 & 902 \\
      $\{1, 2\}$ & 12.2 & 596 \\
      $\{1, 3\}$ & 11.6 & 606 \\
      $\{2, 3\}$ & 14.2 & 790 \\
      $\{1, 2, 3\}$ & 12.4 & 581\\\\
    \end{tabular}
    \caption{}
    \end{subfigure}
  \end{minipage}
 \caption{Geometric Brownian motion. (a)-(b): Oracle CDF of $(S_{\min}, S_{\max})$ in the high-fidelity model computed using $10^5$ Monte Carlo samples. (c): Oracle scaled loss $\gamma_\S$ \eqref{optm} and the optimal exploration sample count $m_\S^*$ \eqref{optm} for budget $B=10^6$, computed using 50,000 samples.} \label{oracle}
\end{figure}

\begin{figure}[htbp]
\centering 
\begin{subfigure}{0.31\textwidth}{\includegraphics[width=\linewidth, trim={0.5cm 0.5cm 1cm 1cm},clip]{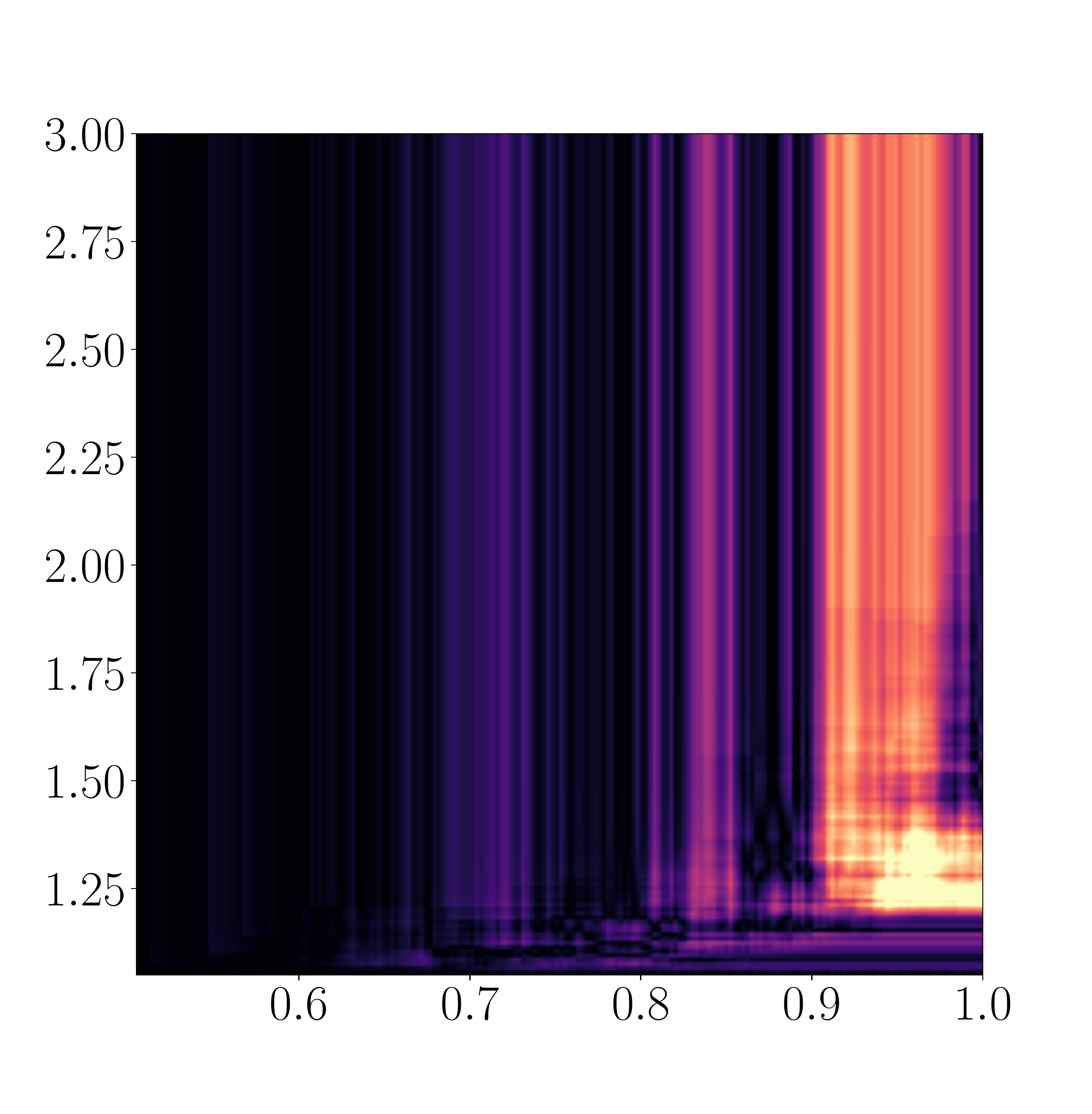}}\caption{}\label{heat1:(a)}
\end{subfigure} 
\begin{subfigure}{0.31\textwidth}{\includegraphics[width=\linewidth, trim={0.5cm 0.5cm 1cm 1cm},clip]{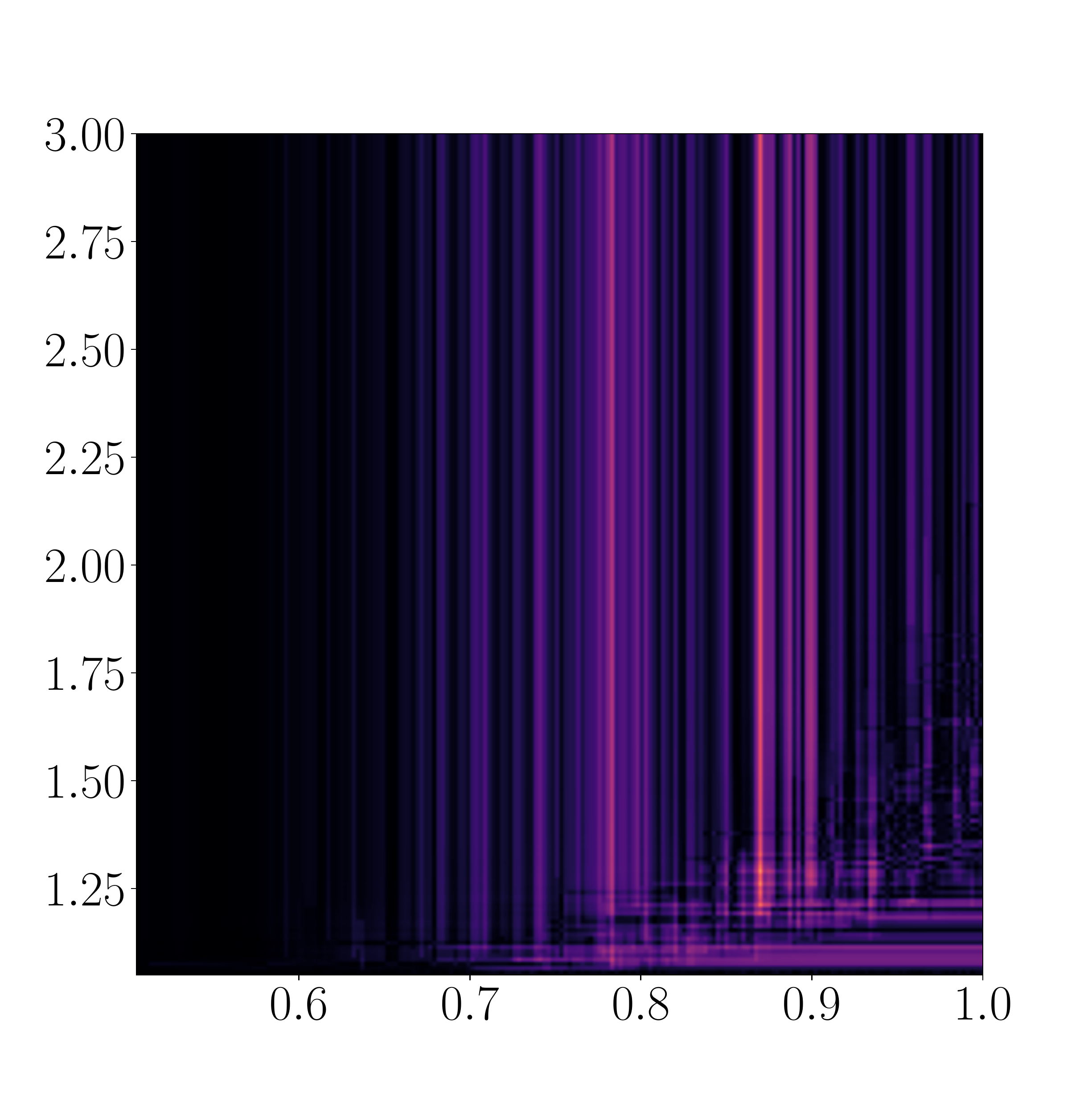}}\caption{}\label{heat2:(b)}
\end{subfigure} 
\begin{subfigure}{0.345\textwidth}{\includegraphics[width=\linewidth, trim={0.5cm 0.5cm 1cm 1cm},clip]{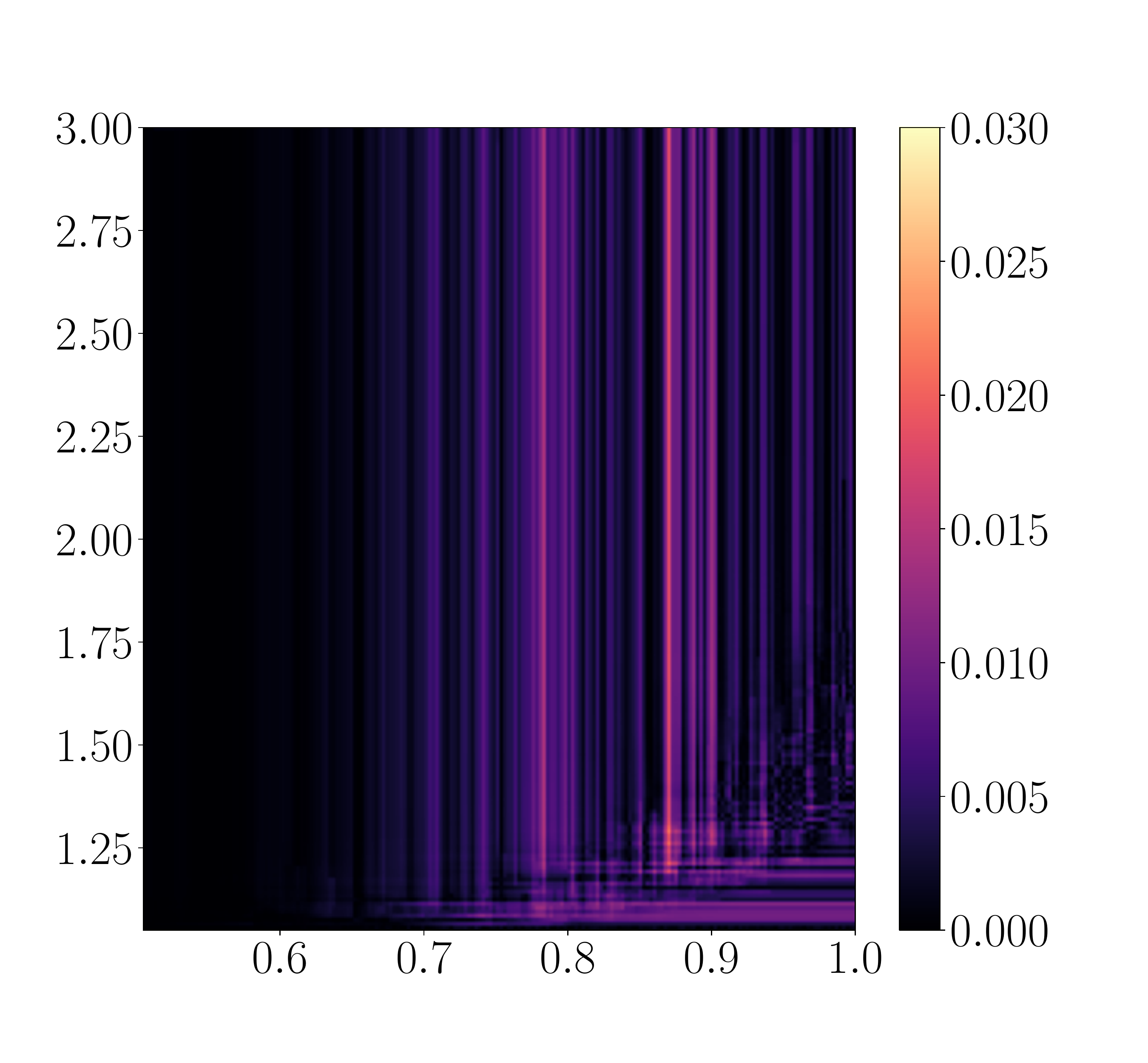}}\caption{}\label{heat3:(c)}
\end{subfigure} 
\caption{Geometric Brownian motion. An instance of absolute pointwise estimation errors of (a) ECDF, (b) cvMDL, and (c) cvMDL-sorted for budget $B = 10^6$.} \label{sde-heat}
\end{figure}

\subsection{Additional results for brittle fracture in \Cref{sec:num3}}\label{sec:num3-details}

We present additional experimental details that supplement those presented in \Cref{sec:num3}.

Recall the boundary value problem in ~\eqref{disp_gov} and \eqref{damage_gov}. The boundary conditions on $\Gamma_N$ and $\Gamma_D$ are
\begin{alignat*}{2}
[(1-\phi_d(\mathbf{x}))^2+q]\nabla\cdot \bm{\sigma}(\mathbf{x})&={\mathbf{v}_n},&\quad\mathbf{x}~\text{on}~\Gamma_N,\\
\mathbf{u}(\mathbf{x})&={\bf{0}},&\quad\mathbf{x}~\text{on}~\Gamma_D,
\\
\nabla\phi_d(\mathbf{x})\cdot\mathbf{n}&=0,&\quad\mathbf{x}~\text{on}~\Gamma_N,
\end{alignat*} 
where $q \ll 1$,  $\bm{\sigma}(\mathbf{x})=\dfrac{\partial\Psi(\bm{\epsilon}(\mathbf{x}))}{\partial{\bm{\epsilon}(\mathbf{x})}}$ is the Cauchy stress tensor, and $\Psi(\bm{\epsilon}(\mathbf{x}))=\dfrac{1}{2}\lambda(\mathrm{tr}(\bm{\epsilon}(\mathbf{x})))^2+\mu\mathrm{tr}(\bm{\epsilon}(\mathbf{x})^2)$ is the elastic energy density with $\mu$ and $\lambda$ the Lam\'e constants, i.e.,
\[
\lambda = \dfrac{\nu \kappa}{(1+\nu)(1-2\nu)},\qquad \mu = \dfrac{\kappa}{2(1+\nu)}
\]
with Young's modulus $\kappa$ and Poisson's ratio $\nu$, and $\bm{\epsilon}(\mathbf{x})=\dfrac{1}{2}\left[\nabla\mathbf{u}(\mathbf{x}) + \nabla\mathbf{u}(\mathbf{x})^{\top}\right]$ is the small strain tensor. In~\eqref{damage_gov} the history variable $H(\mathbf{x})$ is defined as:
\begin{equation*}
	H(\mathbf{x}) =
	\begin{cases}
		\Psi(\bm{\epsilon}(\mathbf{x})), & \Psi(\bm{\epsilon}(\mathbf{x}))<H_i(\mathbf{x}) \\
		H_i(\mathbf{x}), & \text{otherwise}
	\end{cases}
 ,\quad i=1,2,\ldots,n,
\end{equation*}
where $H_i(\mathbf{x})$ is the strain energy computed at $i$th step of the discretized load, which corresponds to the iterative solver stage $\bar{v}_i\cdot e_2$, with $\bar{v}_i\in[0,\bar{v}]$. 
\newpage
\end{appendix}

\bibliographystyle{siamplain}

\bibliography{BL-CDF}

\end{document}